\documentclass[11pt]{article}
\usepackage[dvipsnames]{xcolor}
\definecolor{Gred}{RGB}{219, 50, 54}
\definecolor{ToCgreen}{RGB}{0, 128, 0}

\usepackage[margin=1.0in]{geometry}
\usepackage[T1]{fontenc}
\usepackage{XCharter}

\usepackage{babel}
\usepackage[spacing=true,kerning=true,babel=true,tracking=true]{microtype}

\providecommand{\I}{\mathds{I}}

\providecommand{\diag}{\operatorname{diag}}

\usepackage[normalem]{ulem}

\usepackage{qcircuit}
\usepackage{graphicx}
\usepackage{amsmath,amssymb,amsthm,mathrsfs,amsfonts,dsfont}
\usepackage{mathtools}
\usepackage{subfigure, epsfig}
\usepackage{braket}
\usepackage{bm}
\usepackage{enumerate}
\usepackage{xcolor}
\usepackage{comment}
\usepackage{scalefnt}
\usepackage{authblk}
\usepackage{nicematrix}
\usepackage{booktabs}
\definecolor{Gred}{RGB}{219, 50, 54}
\definecolor{ToCgreen}{RGB}{0, 128, 0}
\usepackage[colorlinks]{hyperref}
\usepackage{cleveref}
\hypersetup{
  colorlinks = true,
  citecolor  = ToCgreen,
  linkcolor  = Sepia,
  filecolor  = Gred,
  urlcolor   = Gred
  }
\usepackage{makecell,multirow,tabularx,tabulary}

\usepackage{pagecolor}
\usepackage{lmodern}
\usepackage[utf8]{inputenc}

\usepackage{tikzscale}
\usepackage{pgfplots}
\pgfplotsset{compat=newest}
\usetikzlibrary{plotmarks}
\usetikzlibrary{arrows.meta}
\usetikzlibrary{external}
\usepgfplotslibrary{patchplots}
\usepackage{grffile}
\usepackage{enumitem}

\DeclareMathOperator{\tr}{Tr}
\DeclareMathOperator{\Tr}{Tr}

\newcommand{\coleq}{\mathrel{\mathop:}\nobreak\mkern-1.2mu=}
\newcommand{\eqcol}{\mkern-1.2mu=\mathrel{\mathop:}\nobreak}

\newcommand{\mc}{\mathcal}

\newcommand{\mr}{\mathrm}
\newcommand{\msf}{\mathsf}
\newcommand{\mbb}{\mathbb}

\newcommand{\expval}[1]{{\langle #1 \rangle}}
\newcommand{\ketbra}[2]{{\vert #1 \rangle\! \langle #2 \vert}}

\newcommand{\E}{\mathop{\mbb{E}}}

\newcommand{\var}{{\mathop{\mr{Var}}}}

\newcommand{\bb}{\begin{equation}\begin{aligned}\hspace{0pt}}
\newcommand{\bbb}{\begin{equation*}\begin{aligned}}
\newcommand{\ee}{\end{aligned}\end{equation}}
\newcommand{\eee}{\end{aligned}\end{equation*}}
\newcommand{\eqt}[1]{\stackrel{\mathclap{\text{\scriptsize \mbox{#1}}}}{=}}
\newcommand{\leqt}[1]{\stackrel{\mathclap{\text{\scriptsize \mbox{#1}}}}{\leq}}
\newcommand{\geqt}[1]{\stackrel{\mathclap{\text{\scriptsize \mbox{#1}}}}{\geq}}

\makeatletter
\newcommand{\pushright}[1]{\ifmeasuring@#1\else\omit\hfill$\displaystyle#1$\fi\ignorespaces}
\makeatother

\newtheorem{lemma}{Lemma}

\newtheorem{proposition}{Proposition}

\newtheorem{corollary}[lemma]{Corollary}%
\newtheorem{definition}[lemma]{Definition}

\newtheorem{theorem}{Theorem}

\definecolor{applegreen}{rgb}{0.55, 0.71, 0.0}

\definecolor{commentblue}{RGB}{70, 100, 140}
\definecolor{commentgreen}{RGB}{75, 130, 95}
\definecolor{commentred}{RGB}{160, 80, 80}

\newcommand{\comments}[1]{}
\usepackage{algorithm}  
\usepackage[noend]{algpseudocode}  
\newcommand{\algorithmfootnote}[2][\footnotesize]{%
  \let\old@algocf@finish\@algocf@finish%
  \def\@algocf@finish{\old@algocf@finish%
    \leavevmode\rlap{\begin{minipage}{\linewidth}
    #1#2
    \end{minipage}}%
  }%
}
\usepackage{xparse}
\NewDocumentCommand{\LeftComment}{s m}{%
  \Statex \IfBooleanF{#1}{\hspace*{\ALG@thistlm}}\(\triangleright\) #2}
\algnewcommand{\LineComment}[1]{\Statex // #1}

\begin{document}

\title{\makebox[\textwidth][c]{When are bosonic Gaussian states classical to learn?}}

\author[1]{Senrui Chen\thanks{Authors are listed in alphabetical order.}\thanks{\texttt{csenrui@gmail.com}}}
\author[2]{Antonio Anna Mele\thanks{\texttt{antoniomele.p@gmail.com}}}
\author[1]{Francesco Anna Mele\thanks{\texttt{fmele@caltech.edu}}}
\author[1]{John Preskill\thanks{\texttt{preskill@caltech.edu}}}

\affil[1]{\small \textit{Institute for Quantum Information and Matter, Caltech, Pasadena, CA 91125, USA}}
\affil[2]{\small \textit{Dahlem Center for Complex Quantum Systems, Freie Universität Berlin, 14195 Berlin, Germany}}

\date{September 22, 2026}
\maketitle
\begin{abstract}
A fundamental question in physics is: When does classical behavior emerge from quantum systems? Bosonic Gaussian states provide a natural setting to explore this quantum-classical boundary, as they capture both the classical field behavior and the intrinsic quantum nature of light. Here, we address this problem from a learning-theoretic perspective by asking: \emph{When are bosonic Gaussian states classical to learn?} That is, under what conditions (if any) can an $n$-mode bosonic Gaussian state be learned with as few samples, and with operations as simple, as are needed to learn a classical $2n$-variate Gaussian distribution? We establish a smooth crossover in learnability governed by the state's thermal fluctuations:
\begin{enumerate}
\item \emph{Cold Gaussian states are non-classical to learn.} When the covariance matrix satisfies $\Sigma \le (\frac12+O(\frac1n))\I$, i.e. close to the vacuum covariance, tomography under single-copy (i.e., non-entangled) measurements fundamentally requires $\Omega(n^3)$ copies—strictly exceeding the sample complexity $\Theta(n^2)$ of learning classical Gaussian distributions. We show that this hardness persists even when few-copy entangled measurements are allowed. 
\item \emph{Warm Gaussian states are classical to learn.} When thermal fluctuations exceed the vacuum noise, parameterized by $\Sigma \ge (\frac{1}{2} + \nu)\I$ for any parameter $\nu > 0$, we prove that single-copy tomography requires $N = \Theta\left(n^2 \min(n, 1 + \nu^{-1})\right)$ copies. This bound is tight and is achieved by simple, non-adaptive, unentangled heterodyne measurements. Crucially, for $\nu = \Omega(1)$, the sample complexity drops to $\Theta(n^2)$, matching the classical case.
\end{enumerate}
Our results tightly characterize a quantum-to-classical crossover in the learnability of bosonic Gaussian states, reveal a novel connection between fundamental physics and statistical learning theory, and have implications for real-world quantum sensing experiments.
\end{abstract}

\newpage
\tableofcontents

\section{Introduction}

The world is governed by the laws of quantum mechanics, but our daily life appears to be classical. 
It is now well accepted that such macroscopic classical behavior emerges from underlying microscopic quantum systems~\cite{zurek2003decoherence}. 
This phenomenon was first understood through investigations into the nature of light: At the macroscopic level, light behaves like a classical electromagnetic field as characterized by Maxwell's equations. At the microscopic level, light consists of photons and thus has quantized energy. 
The theoretical formalization and experimental validation built around the wave-particle duality of light are fundamental to the establishment of today's quantum physics.

In modern quantum information language, light can be described by bosonic continuous-variable quantum systems. For an $n$-mode bosonic system, the quantum state is either described via a quasi-probability distribution over $2n$ quadrature observables, known as the \emph{Wigner function}, or via the photon number occupation basis, known as the \emph{Fock basis}. These two representations are equivalent and reflect the viewpoints of waves and particles, respectively.
{
A class of bosonic quantum states that appear ubiquitously in nature (thanks to a central-limit-type argument~\cite{glauber1963coherent}) is the class of \emph{Gaussian states}, defined as those whose Wigner functions are Gaussian probability distributions. 
Gaussian states describe the most classical-field-like state of light. Indeed, they have $2n$ well-defined quadrature expectation values (given by the Gaussian mean) with added Gaussian noise (given by the Gaussian covariance). 
At the same time, they also capture the quantum nature of light, in that the covariance matrix is bounded from below due to the uncertainty principle.
The seminal works by Glauber and Sudarshan~\cite{glauber1963coherent,sudarshan1963equivalence} investigated the quantum-classical boundary for states of light. 
In particular, Glauber observed that an asymptotically large thermal fluctuation renders the quantum fluctuation negligible in comparison and thus makes the state essentially a classical random field.
}

The goal of the current work is to study when classical behavior emerges in quantum systems from the perspective of \emph{statistical learning theory}.
The basic question is the following: Given $N$ identical and independent copies of an unknown physical state (which can be a quantum state or a classical distribution) sampled from a certain family, how large must $N$ be to ensure that an observer can learn a classical description of the state to $\varepsilon$ \emph{distinguishability distance}\footnote{Here, \emph{distinguishability distance} refers to the \emph{trace distance} in quantum theory and the \emph{total variation distance} in classical probability theory. We write $D_\mr{tr}(\rho,\sigma):=\frac12\|\rho-\sigma\|_1=\frac12\|\rho-\sigma\|_\mr{tr}$, with both $\|\cdot\|_1$ and $\|\cdot\|_\mr{tr}$ denoting the trace norm. Within both theories, if two objects have $\varepsilon$ distinguishability distance, the maximal probability that an observer can distinguish them correctly is $(1+\varepsilon)/2$ in the balanced hypothesis testing setting.} with high probability (say $2/3$).  The answer is referred to as the sample complexity of the learning task.

For a family of $n$-mode bosonic Gaussian states, since their Wigner functions are given by $2n$-variate classical Gaussian distributions, it is natural to compare the sample complexity of learning both families. In particular, we say there is an emergent classical behavior if (1) the quantum and classical families have the same sample complexity; and (2) optimal learning of the quantum family can be achieved via simple ``classical-like'' measurements. While there can be multiple different definitions for classical-like measurement, here we require the measurement to at least be \emph{non-entangled}, i.e., no joint measurements across multiple copies are allowed. See Fig.~\ref{fig:measurement} for an illustration. Intuitively, when these criteria are satisfied, the family of quantum states is essentially \emph{classical to learn}.

The sample complexity for learning a $2n$-dimensional \emph{classical Gaussian} distribution is known to be $N=\Theta(n^2/\varepsilon^2)$. In fact, simple empirical estimators for sample mean and covariance suffice to achieve the optimal scaling~\cite{ashtiani2020near,devroye2020minimax}. Very recently, the sample complexity for learning an $n$-mode \emph{bosonic Gaussian} state has also been found to be $N=\Theta(n^2/\varepsilon^2)$~\cite{chen2026optimal_merged}. 
However, the only known protocol to achieve this uses a highly entangled measurement, violating the criterion of using ``classical-like'' measurements. In contrast, the best currently known non-entangled protocol uses $N=\widetilde O(n^3/\varepsilon^2)$ copies~\cite{bittel2025energy}, which does not match the classical learning complexity.
We thus ask the following question,
\begin{center}
    {\em
        Can bosonic Gaussian states be learned using $\Theta(n^2/\varepsilon^2)$ copies via classical-like measurements?\\
        If not, under what additional conditions does this hold?
    }
\end{center}
Beyond the theoretical interest of exploring the quantum-to-classical crossover, finding simpler sample-optimal learning algorithms for bosonic Gaussian states is of practical relevance. Indeed, learning and estimation of Gaussian states is a crucial subroutine for quantum sensing, with examples including gravitational wave and dark matter detection. It is also useful for benchmarking photonic quantum computers and quantum networks. 
For those practical quantum devices, a multi-copy entangled measurement is often out of reach. It is thus highly desirable to consider more experimentally friendly learning protocols without compromising the efficiency.

\begin{figure}[!tp]
    \centering
    \includegraphics[width=0.9\linewidth]{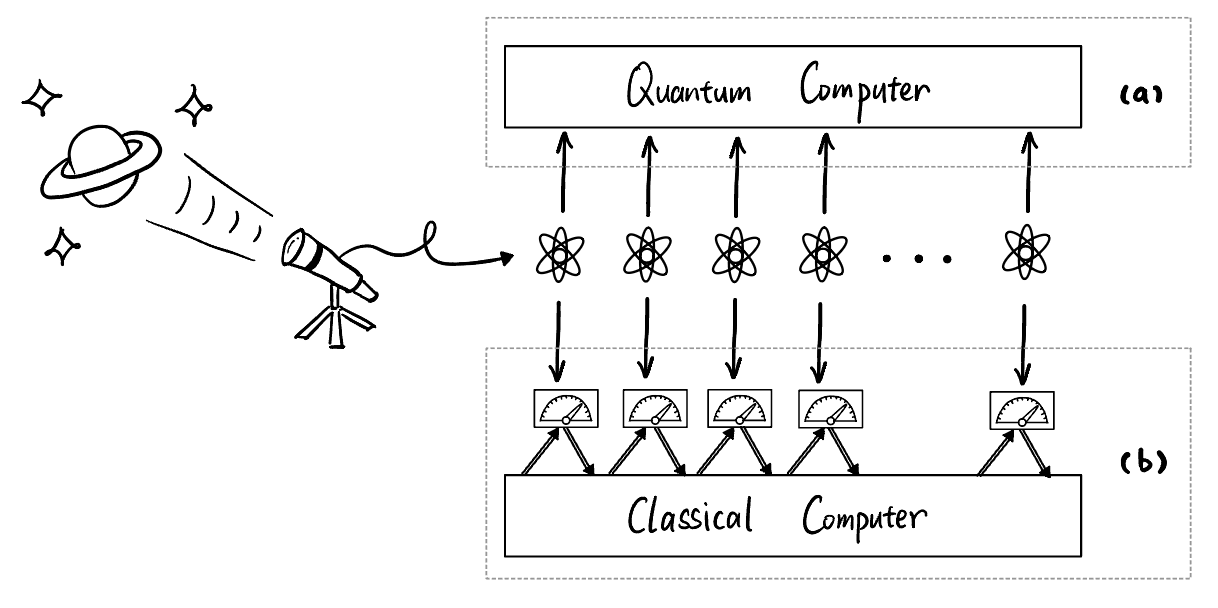}
    \caption{Illustration of entangled vs. non-entangled learning schemes. An experimental device (on the left) generates many i.i.d. copies of an unknown quantum state of interest, which can be processed in two different ways: (a) With a quantum computer, one can conduct entangled measurements across all copies, which is the most powerful measurement allowed by quantum mechanics; (b) With a classical computer only, one can conduct only single-copy non-entangled measurements. Note that each measurement may be chosen adaptively based on outcomes of previous measurements.}
    \label{fig:measurement}
\end{figure}

\subsection{Background}\label{sec:background}
We first introduce some background in bosonic quantum information that is necessary to state our results. Bosons are quantum particles whose states are invariant under exchange. An $n$-mode bosonic system consists of different photon number sectors, with the $m$-photon Hilbert space given by $\mc H_m^{(n)}\cong\mr{Sym}^m(\mbb C^n)$ and the full (infinite-dimensional) Hilbert space given by $\mc H^{(n)}=\bigoplus_{m=0}^\infty \mc H_m^{(n)}$.
A natural basis for $\mc H^{(n)}$ is the Fock basis $\{\ket{\bm m}\equiv\ket{m_1}\otimes\cdots\otimes\ket{m_n}\}_{\bm m\in \mbb N^n}$ which specifies the number of photons in each mode. We will also use the ladder operators $\{\hat a_j,\hat  a_j^\dagger\}_{j=1}^n$, defined as
\bb
    \hat a_j\ket{\bm m}=\sqrt{m_j}\ket{\bm m - \bm e_j},\quad
    \hat a_j^\dagger\ket{\bm m}=\sqrt{m_j+1}\ket{\bm m + \bm e_j},
\ee
and the photon number operators $\hat n_j\coleq \hat  a_j^\dagger\hat a_j$ which satisfy $\hat n_j\ket{\bm m} = m_j\ket{\bm m}$. 
The total photon number is given by $\hat n\coleq\sum_{j=1}^n\hat n_j$.
We use $\mc D(\mc H^{(n)})$ to denote the set of all density operators on $\mc H^{(n)}$.

\medskip
We also introduce the $2n$ quadrature operators $\hat{\bm r}\coleq(\hat q_1,\cdots,\hat q_n,\hat p_1,\cdots,\hat p_n)$ defined as
\bb
    \hat  q_j=\frac{\hat a_j + \hat a_j^\dagger}{\sqrt2},\quad
    \hat  p_j=\frac{\hat a_j - \hat a_j^\dagger}{i\sqrt2},
\ee
which satisfy the canonical commutation relation $[\hat r_k,\hat r_l]=i
\Omega_{kl}\I$ where $\Omega\coleq\begin{pmatrix}
    0_n & \I_n\\-\I_n & 0_n
\end{pmatrix}$.
Now, for an operator $\hat O$ acting on $\mc H^{(n)}$, define its characteristic function as $\chi_{\hat O}(\bm \alpha)\coleq\Tr(\hat O e^{i\bm\alpha^T\Omega\hat{\bm r}})$ and the \emph{Wigner function} as\footnote{
When $\hat O$ is not a trace-class operator, the characteristic function and Wigner function should in general be understood as generalized functions. Throughout this work, it suffices to focus on trace-class operators.
} 
\bb
W_{\hat O}(\bm\beta)\coleq\frac{1}{(2\pi)^{2n}}\int\mr{d}^{2n}{\bm\alpha}e^{i{\bm\alpha}^T\Omega\bm\beta}\chi_{\hat O}(\bm \alpha),\quad\forall\bm\beta\in\mbb R^{2n}.
\ee
For any density operator $\rho\in\mc D(\mc H^{(n)})$, its Wigner function satisfies two important properties: (i) Normalization, i.e. $\int \mr{d}^{2n}\bm\beta\,W_\rho(\bm\beta)=1$. Note that $W_\rho$ can in general contain negative values, and is therefore called a quasi-probability distribution. (ii) All observable expectation values can be computed using $\Tr(\rho\hat O)=(2\pi)^n\int\mr d^{2n}\bm\beta\, W_\rho(\bm\beta) W_{\hat O}(\bm\beta)$; thus the Wigner function completely describes a bosonic state.

\medskip
A \emph{Gaussian state} is a state whose Wigner function is a Gaussian distribution, i.e. $W_\rho = \mc N(\mu,\Sigma)$. Such a state is completely determined by its first and second moments,
\bb
\mu_j = \Tr(\rho \hat r_j),\quad \Sigma_{kl}\coleq\frac12 \Tr(\rho \{\hat r_k - \mu_k, \hat r_l - \mu_l\});
\ee
thus we denote a Gaussian state by $\rho(\mu,\Sigma)$. A pair $(\mu,\Sigma)$ defines a valid Gaussian state if and only if $\Sigma + i\Omega/2\ge0$~\cite{weedbrook2012gaussian}. 
A subclass of Gaussian states of particular interest is called \emph{passive} Gaussian states (or, \emph{gauge-invariant} Gaussian states), which are Gaussian states whose density operators are block diagonal in the total photon number sectors. {Passive Gaussian states are always zero-mean, and the eigenvalues of $\Sigma$ coincide with the \emph{symplectic eigenvalues}; each equals $\bar n_j +\frac{1}{2}$, where $\bar n_j$ is the expected photon occupation of the corresponding eigenmode, a proxy for its temperature.}
Equivalently, they are states that can be obtained by applying a passive Gaussian unitary on products of thermal states. 

{
    An alternative representation of bosonic states is the Glauber-Sudarshan $P$ function~\cite{glauber1963coherent,sudarshan1963equivalence} defined as $\rho=\int P(\alpha)\ketbra{\alpha}{\alpha}\,\mr d^{2n}\alpha$ where $\ket{\alpha}$ is the coherent state (\textit{i.e.}, eigenstates of $\{\hat a_j\}$).
    Conventionally, a state with positive $P$ function is viewed as a ``classical state'', as it is a mixture of coherent states and thus has no entanglement or squeezing.
    However, such a state is not automatically classical to learn, much as separable states may require entangled measurements for optimal learning~\cite{divincenzo2002quantum}.\footnote{
    {
    As Glauber pointed out~\cite{glauber1963coherent}, a $P$-positive Gaussian state can still behave very differently from a classical Gaussian, becoming classical only when the thermal fluctuation dominates the zero-point quantum fluctuation. Our goal is to characterize this crossover quantitatively in a statistical learning task.
    }
    }
    Specific to Gaussian states, a positive $P$ function is equivalent to requiring $\Sigma>\frac12\I$. 
    Another condition we will use is $\nu_-\I\le\Sigma-\frac12\I\le\nu_+\I$, which implies the expected number of photons in each mode of the centered state is bounded between $\nu_-$ and $\nu_+$. This is because $\tr(\hat n_j\rho) = \frac12\tr((\hat q_j^2 +\hat p_j^2 -1)\rho)$. We will thus refer to $\nu_{\pm}$ as the \emph{photon number} or \emph{temperature} parameter, the latter with slight abuse of terminology (but accurate for passive Gaussian states).

    All of our theorems stated in Sec.~\ref{subsec:results} hold for general POVM measurements on $n$-mode bosonic systems, which are allowed to be \emph{non-Gaussian}. Here, we review one specific type of Gaussian measurement, \emph{heterodyne measurement}, whose outcome distribution on an $n$-mode Gaussian state $\rho(\mu,\Sigma)$ is a $2n$-variate classical Gaussian $\mc N(\mu,\Sigma+\frac12\I)$; that is, the state's Wigner function convoluted with additional Gaussian noise.
}

\subsection{Results}
\label{subsec:results}

We first present our technical results and later discuss the physical implications and practical applications.
Our first result sets a limitation for non-entangled measurements in Gaussian state learning:
\begin{theorem}[Lower bound against single-copy measurements]\label{th:main_lower_1_copy}
    There exist universal constants $\varepsilon_0>0$ and $n_0\in\mbb N$ such that the following holds for all $n\ge n_0$ and $0<\varepsilon<\varepsilon_0$.
    A single-copy, possibly adaptive scheme that can learn any $n$-mode Gaussian state $\rho(\mu,\Sigma)$ to $\varepsilon$ trace distance with probability at least $2/3$ requires $N=\Omega(n^3/\varepsilon^2)$ copies. This holds even under the promise that $\rho$ is passive and that $\Sigma\le(\frac12+O(\frac1n))\I$.
    Furthermore, when restricted to non-adaptive schemes, the same lower bound holds even if the spectrum of $\Sigma$ is a priori known.
\end{theorem}
{\noindent We refer to Gaussian states satisfying $\Sigma\le(\frac12+O(\frac1n))\I$ as \emph{cold Gaussian states}:  their expected total number of photons is $O(1)$, so they are close to the (zero-temperature) vacuum.}
For passive Gaussian states, Theorem~\ref{th:main_lower_1_copy} matches the performance of heterodyne measurement~\cite{bittel2025energy}, which is a non-entangled, non-adaptive Gaussian measurement. For general Gaussian states, it matches the upper bound given by the adaptive single-copy Gaussian scheme in~\cite{bittel2025energy} up to a doubly-logarithmic energy factor. This result also shows that entangled measurements provide a provable advantage in learning Gaussian states, achieving $N=\Theta(n^2/\varepsilon^2)$ instead~\cite{chen2026optimal_merged}. Crucially, the hardness of learning can be established using an extremely low-temperature family of passive states whose total number of photons in expectation is $O(1)$.

One might wonder if Theorem~\ref{th:main_lower_1_copy} can be circumvented via few-copy entangled measurements, as there are examples where $2$-copy measurements are exponentially more efficient than $1$-copy measurements~\cite{huang2021information,chen2022exponential}. Our second result gives a negative answer:

{%

\begin{theorem}[Lower bound against $k$-copy entangled measurements]
\label{th:main_lower_k_copy}
There exist universal constants $\varepsilon_0>0$ and $n_0\in\mbb N$
such that the following holds. Let $n\ge n_0$, let
$0<\varepsilon<\varepsilon_0$, and let $k\ge1$ be any integer.
Any possibly adaptive scheme that learns every $n$-mode Gaussian state
$\rho(\mu,\Sigma)$ to trace-distance error $\varepsilon$ with probability
at least $2/3$, using only $k$-copy measurements, requires
\bb
    N=\Omega\left(
        \frac{n^3}{\varepsilon^2\min\{\sqrt{k},n\}}
    \right)
\ee
copies. Here each measurement acts on at most $k$ fresh copies, and no
quantum memory is retained between measurements. This holds even under
the promise that $\rho$ is passive and that
$\Sigma\le(\frac12+O(\frac1n))\I$.
\end{theorem}

\noindent In particular, for any constant $k$, learning via $k$-copy measurements still requires
$N=\Omega(n^3/\varepsilon^2)$ copies. Again, this result holds even for low-temperature passive Gaussian states.
}

Taken together, Theorem~\ref{th:main_lower_1_copy} and Theorem~\ref{th:main_lower_k_copy} give a negative answer to the first part of our main question. That is, general bosonic Gaussian states cannot be learned using classical-like sample complexity via classical-like measurements. 
The proof idea behind both theorems is inspired by the literature of learning finite-dimensional quantum states with single-copy or few-copy measurements~\cite{lowe2025lower,chen2023does,chen2024optimal2,keskin2026tight, nayak2026optimal}, together with novel techniques necessary to make things work for bosons. 
{%
Specifically, Theorem~\ref{th:main_lower_k_copy} is proved by establishing that tomography of $n$-dimensional mixed states reduces to tomography of $n$-mode passive Gaussian states, and thus the finite-dimensional no-go theorem of~\cite{keskin2026tight} applies (with the latter being a refinement of~\cite{chen2024optimal2}). The case $k=1$ also proves the adaptive part of
Theorem~\ref{th:main_lower_1_copy}. The direct bosonic analyses discussed later
further address the temperature dependence and the known-spectrum
promise for non-adaptive schemes (see Sec.~\ref{sec:techsum_2} for an overview and Secs.~\ref{sec:adaptive} and~\ref{sec:non-adaptive} for the corresponding proofs).} %

\begin{figure}[t]
    \centering
    \includegraphics[width=0.65\linewidth]{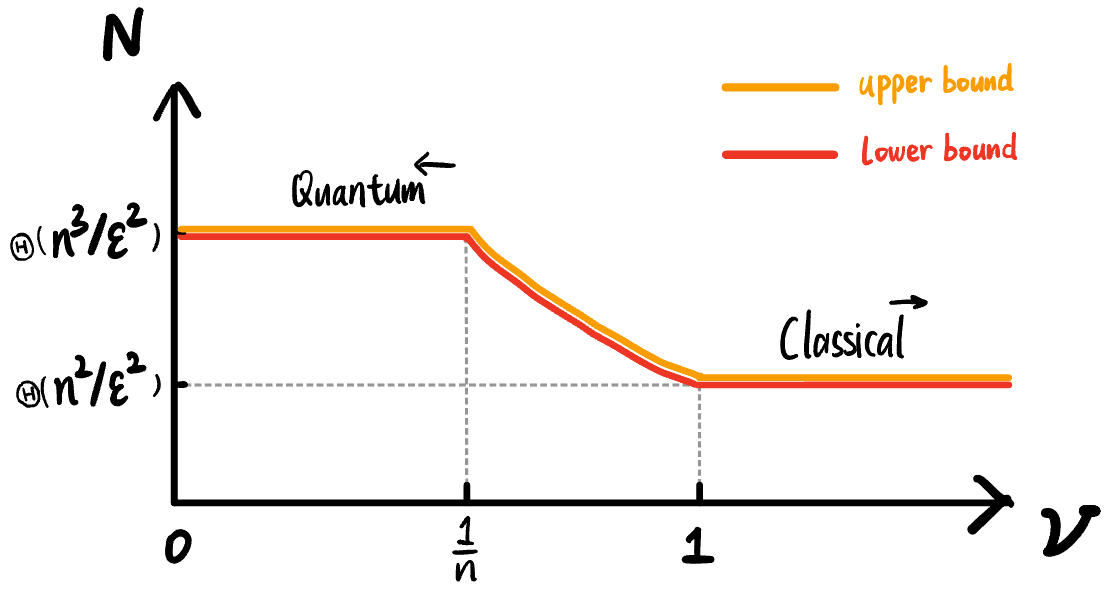}
    \caption{Quantum-to-classical crossover in the learnability of bosonic Gaussian states. The vertical axis denotes the sample complexity for learning to $\varepsilon$ trace distance with $2/3$ probability using single-copy (i.e. non-entangled) measurements. The horizontal axis denotes a minimal-temperature parameter $\nu$ such that the family of Gaussian states is promised to satisfy $\Sigma\ge(\frac12+\nu)\I$. 
    Our Theorem~\ref{th:main_transition} gives a tight characterization of the sample complexity in the entire regime of $\nu\in(0,\infty)$ as $N=\Theta\!\left({n^2}{\varepsilon^{-2}}\min\{n,~1+\nu^{-1}\}\right)$, and reveals a smooth learnability crossover from $\nu\sim 1/n$ to $\nu\sim 1$.}
    \label{fig:transition}
\end{figure}

\medskip
Now we turn to the question of what additional conditions are needed for Gaussian states to be classical to learn. Intuitively, a state with sufficiently large fluctuations compared to vacuum should look classical~\cite{glauber1963coherent}. We precisely quantify this intuition in the following theorem:
\begin{theorem}[Tight transition of learnability]  \label{th:main_transition}
    There exist universal constants $\varepsilon_0>0$ and $n_0\in\mbb N$ such that the following holds for all $n\ge n_0$ and $0<\varepsilon<\varepsilon_0$.
    Under the promise that an $n$-mode Gaussian state $\rho(\mu,\Sigma)$ satisfies $\Sigma\ge(\frac12+\nu)\I$ for some $\nu>0$ that can depend on $n$, the necessary and sufficient number of copies to learn $\rho$ to $\varepsilon$ trace distance with probability at least $2/3$ using single-copy measurements is given by
    \bb
        N = \Theta\!\left(\frac{n^2}{\varepsilon^2}\min\{n,~1+\nu^{-1}\}\right).
    \ee
    Or, equivalently,
    \bb
        N = \left\{\begin{aligned}
              & \Theta\!\left(\frac{n^3}{\varepsilon^2}\right),\quad &&0<\nu<\frac1n,
            \\& \Theta\!\left(\frac{n^2}{\nu\varepsilon^2}\right),\quad && \frac1n\le\nu \le1,
            \\& \Theta\!\left(\frac{n^2}{\varepsilon^2}\right),\quad &&\nu > 1.
        \end{aligned}\right.
    \ee
    This bound can be achieved by heterodyne measurements for all values of $\nu>0$.
\end{theorem}
As introduced in Sec.~\ref{sec:background}, $\nu$ is a parameter that lower bounds the minimal (expected) number of photons in each mode. Alternatively, it is a lower bound on the minimal fluctuation of a Gaussian state in any quadrature, offset by the vacuum variance $\frac12$. 
Our results tightly characterize the quantum-to-classical crossover in the learnability of bosonic Gaussian states for any temperature parameter $\nu>0$.\footnote{Even squeezed Gaussian states, whose covariance matrix has eigenvalues smaller than $1/2$, can be learned using $N=\widetilde{\Theta}(n^3/\varepsilon^2)$ samples via non-entangled measurements, where the tilde hides a $\mr{poly}\log\log$ term depending on the energy~\cite{bittel2025energy}.} 
We refer to Gaussian states with $\nu=\Omega(1)$ as \emph{warm Gaussian states}: every eigenmode carries a constant number of thermal photons. Such warm Gaussian states can be learned by simple classical-like measurements with classical-like sample complexity. We have answered the second part of our main question: a constant number of thermal photons per mode suffices for Gaussian states to be classical to learn. We illustrate the crossover in Fig.~\ref{fig:transition}.

{To our knowledge, the only previously known bound for single-copy learning of Gaussian states is a uniform upper bound of $O(n^3/\varepsilon^2)$ for heterodyne measurements~\cite{bittel2025energy, fanizza2025}, which we show is tight only in the $\nu=O(1/n)$ regime. (Here we focus on $\nu>0$ so there is no additional energy-dependent factor.)
} %
The lower bound in all regimes and the upper bound for $\nu = \omega(1/n)$ are our new contributions. 
The component bounds are proved in Sections~\ref{sec:adaptive}, \ref{sec:warm-upper}, \ref{sec:warm-lower} and combined at the end of Section~\ref{sec:warm-lower} to establish Theorem~\ref{th:main_transition}.

{
\subsection{Discussion}

Our results answer the question posed at the outset. Bosonic Gaussian states are classical to learn when each eigenmode contains at least a constant number of thermal photons in expectation. They are maximally hard to learn when the occupation number is of order $1/n$ per eigenmode, so that the total number of thermal photons is constant.
There is a smooth crossover between these two regimes, corresponding to $\nu\sim 1/n$ and $\nu\sim 1$.

Our results can be intuitively understood as a manifestation of wave-particle duality. In the cold regime ($\nu\sim 1/n$), since the total photon number is a constant, the state's behavior is largely dominated by the single-particle sector, which is isomorphic to an $n$-level quantum particle. Tomography of such an $n$-dimensional state using single-copy measurements is known to require $\Theta(n^3)$ samples because single-copy measurements cannot be aligned with the state's unknown eigenbasis, and measurements in different bases are mutually incompatible~\cite{chen2023does}. See Sec.~\ref{sec:techsum_1} for a more rigorous development of this intuition. 

In the warm regime ($\nu \sim 1$), the higher-photon-number sectors become significant, which causes the state to more closely resemble a classical electromagnetic wave, and the complexity of learning classical Gaussian distributions, $\Theta(n^2)$, kicks in. 
Note that in both the $\nu\sim 1/n$ and the $\nu\sim 1$ cases, the noise introduced by heterodyne measurements $\frac12\I$ is comparable to the state's fluctuation $\Sigma$, which means one can learn $\Sigma$ to the same relative precision using similar sample complexity in both cases.
{
This separation can be understood through the relation between quantum trace distance and the total variation distance of Gaussian Wigner functions. For states with the same mean, Wigner total variation distance is locally governed by the relative Frobenius error of the covariance~\cite{devroye2018total}, whereas quantum trace distance can be larger by a factor of order $\sqrt n$~\cite[Lemma~1]{chen2026towards}. Under the promise $\Sigma\succeq(\frac12+\nu)\I$, our bound reduces this conversion factor to $O(\sqrt{1+\nu^{-1}})$ for covariance perturbations satisfying $\Sigma'\succeq\Sigma$. The bound thus interpolates smoothly from $O(\sqrt n)$ at $\nu\sim1/n$ to $O(1)$ at $\nu=\Omega(1)$, recovering the classical Wigner-distance scale up to constant factors. Our estimator enforces the required covariance ordering by slightly inflating the empirical covariance, while preserving the order of the relative covariance error. See Sec.~\ref{sec:technical-summary-warm} for further discussion.
}

\begin{table}[t]
\centering
\small
\setlength{\tabcolsep}{4pt}
\renewcommand{\arraystretch}{1.2}

\begin{tabularx}{\textwidth}{
    @{}
    >{\raggedright\arraybackslash}p{0.28\textwidth}
    >{\raggedright\arraybackslash}p{0.22\textwidth}
    >{\raggedright\arraybackslash}X
    >{\raggedright\arraybackslash}p{0.14\textwidth}
    @{}
}
\toprule
\textbf{Experiments}
&
\textbf{Estimated $n$}
&
\textbf{Estimated $\nu$}
&
\textbf{Regime}
\\
\midrule

ADMX axion search
\cite{du2018search}
&
$n\simeq260$
\newline
$\Delta f=25\,\mathrm{kHz}$,
\newline
$\delta f=96\,\mathrm{Hz}$.
&
$\nu\simeq4.3$
\newline
$f\simeq650\,\mathrm{MHz}$,
\newline
$T_{\mathrm{cav}}\simeq150\,\mathrm{mK}$.
&
\textbf{Warm}
\newline
$\nu>1$
\\
\addlinespace[5pt]

Photon-counting axion haloscope
\cite{braggio2025quantum}
&
$n\simeq12$
\newline
$\Delta f\simeq1.1\,\mathrm{MHz}$,
\newline
$\delta f\simeq95\,\mathrm{kHz}$.
&
$\nu\simeq3.2\times10^{-4}$
\newline
$f\simeq7.37\,\mathrm{GHz}$,
\newline
$T_{\mathrm{eff}}\simeq44\,\mathrm{mK}$.
&
\textbf{Cold}
\newline
$\nu\ll1/n$
\\
\bottomrule
\end{tabularx}
\caption{
Examples of mode counts and photon occupations in two dark-matter search experiments, together with their classification into the warm and cold regimes according to our criteria.
Here, $\Delta f$ denotes the analysis or effective detection bandwidth, and $n\simeq\Delta f/\delta f$. For~\cite{du2018search}, $\delta f$ is the spectral-bin width; for~\cite{braggio2025quantum}, it is the effective frequency scale $1/\tau$ associated with the detection window $\tau=10.5\,\mu\mathrm{s}$.
For both cases $\nu$ is estimated using the Bose-Einstein statistics $\nu = \left(\exp(hf/k_B T)-1\right)^{-1}$.
}
\label{tab:experiments}
\end{table}

Let us now connect our results to real-world quantum sensing settings, by citing examples from dark-matter search experiments in the literature. In each example, we estimate the number of modes $n$ and the average photon number per mode $\nu$. Then, by comparing $\nu$ with $1/n$ and $1$, we infer whether the experiments operate in the cold or warm regime in the sense we have defined. To be clear, our rigorous results are for full tomography of an unstructured Gaussian state family (with temperature constraints), which do not immediately apply to those practical setups where the states of interest can lie in a much lower-dimensional subspace. Nevertheless, we believe this comparison provides guidance on whether the experiments operate in a quantum or classical regime.

The comparison is summarized in Table~\ref{tab:experiments}. For both experiments, $n$ denotes the number of effective frequency-bin modes.\footnote{Another relevant setup not discussed here is optical imaging. There $n$ refers to the number of spatial modes.} The average photon number per mode $\nu$ is estimated using Bose-Einstein statistics (see~\cite[App J]{braggio2025quantum}),
$\nu = \left(\exp(hf/k_B T)-1\right)^{-1}$. By our criteria, the ADMX experiment~\cite{du2018search} operates in the warm regime where heterodyne-type readout is provably optimal among single-copy measurements. In contrast, the axion haloscope~\cite{braggio2025quantum} is in the very cold regime (with $\nu$ far below the crossover scale $1/n$), where quantum-limited Gaussian readout carries an unavoidable penalty, and indeed that experiment opts for non-Gaussian photon counting instead.
These examples demonstrate that, depending on the experimental setup, a practical sensing task can operate on either side of our crossover. They also show that practical experiments are genuinely multi-mode, with $n$ ranging from tens to hundreds of bosonic modes.

}

\subsection{Related works}

\paragraph{Bosonic and Gaussian state learning.} 

Bosonic systems have a long history in research on quantum mechanics, quantum optics, and quantum information theory. 
Tomography of bosonic quantum states has long been investigated both in theory and experiment~\cite{lvovsky2009continuous}.
A recent series of works has begun to study the non-asymptotic sample complexity for learning bosonic systems with provable guarantees in trace distance closeness; see~\cite{mele2026advances} for a review.  
In particular, \cite{mele2025learning} shows that learning general energy-constrained bosonic states is extremely inefficient, with a sample complexity that scales exponentially in the number of modes. In contrast, energy-constrained Gaussian states, a practically motivated subclass of bosonic states~\cite{weedbrook2012gaussian,mele2025learning}, can be learned using a polynomial number of samples.

Subsequent works have significantly advanced the understanding of Gaussian state learning. 
\cite{bittel2025energy} proposes an algorithm based on adaptive single-copy Gaussian measurements that learns any $n$-mode Gaussian state with energy at most $E$ to $\varepsilon$ trace distance using $\widetilde{O}(n^3/\varepsilon^2+n\log\log E)$ samples. \cite{chen2026towards} proves this $\Omega(n^3/\varepsilon^2)$ scaling is inevitable for any scheme based on Gaussian measurements (more generally, any Wigner-positive measurement), even for learning the subclass of passive Gaussian states. They also show that non-Gaussian measurements can break the $n^3$ barrier, achieving $\widetilde{O}(n^2/\varepsilon^2+n\log\log E)$ for learning passive Gaussian states. More recently, \cite{chen2026optimal_merged} completely settles the sample complexity of Gaussian state learning to be $\Theta(n^2/\varepsilon^2)$. A concurrent work~\cite{chen2026log} obtains similar complexity for pure Gaussian states, and also proves the surprising $\Omega(\log\log E)$ energy dependence is inevitable when restricted to Gaussian measurements.
It is worth mentioning that an important foundation for these results is the recently developed set of bounds on the trace distance between Gaussian states~\cite{Bittel_2025,holevo2024estimates,bittel2025energy,fanizza2025}.

While previous works have established a separation between Gaussian and non-Gaussian measurements for Gaussian state learning~\cite{chen2026towards,chen2026optimal_merged}, our work establishes a separation between entangling and non-entangling measurements for the same task. Combining these two aspects, we conclude that the $n^2$ sample complexity for learning general Gaussian states is only possible with both non-Gaussian and entangling measurements. Removing either one of these two resources restores the $n^3$ scaling.
On the other hand, for learning warm Gaussian states as in Theorem~\ref{th:main_transition}, both separations disappear, and a Gaussian, non-entangling, non-adaptive scheme suffices to achieve the optimal $n^2$ scaling. Together, these give a complete picture of when bosonic Gaussian states are classical to learn.

\paragraph{Quantum learning theory in finite-dimensional systems.}
Our work belongs to the rapidly growing research program of quantum learning theory. One central focus there is to find provable advantages of quantum resources (e.g., entanglement, magic, quantum memory) in learning certain properties of certain families of quantum states or processes (see e.g.~\cite{huang2022quantum,liu2025quantum,huang2026vast}).
Most of these works are developed for finite-dimensional quantum systems. Among those, a particularly relevant series of works explore the sample complexity for learning $d$-dimensional states using non-entangled or partially-entangled  measurements~\cite{chen2023does,lowe2025lower,chen2024optimal2,keskin2026tight} as opposed to fully-entangled measurements~\cite{Haah2017,ODonnell2016,pelecanos2025}.
As will be reviewed in Sec.~\ref{sec:techsum_1} and~\ref{sec:techsum_2}, the techniques for establishing our lower bounds are inspired by those works, with crucial distinctions unique to the bosonic setting.
To our knowledge, the quantum-to-classical crossover of learnability has not been explored in finite-dimensional systems. 
Indeed, the hard instances for full tomography of $d$-dimensional states are perturbations of the maximally mixed state, which already has infinite-temperature~\cite{chen2023does}; hence there is no warmer family of states that is easier to learn. How to properly generalize the crossover to finite-dimensional systems is thus an interesting open question.

\paragraph{Quantum sensing and imaging.} 
Bosonic and Gaussian quantum information theory are directly relevant to real-world quantum sensing and imaging applications. As some examples, squeezed Gaussian states are implemented in state-of-the-art gravitational-wave interferometers to suppress quantum noise~\cite{mcculler2020frequency}. 
Dark-matter searches can be modeled as finding a weak narrow-band bosonic signal in a wide-band thermal background~\cite{backes2021quantum}. Spatially incoherent optical sources studied in quantum imaging can be described by passive Gaussian states~\cite{tsang2016quantum}, etc. These tasks motivate studying efficient characterization and parameter estimation of multi-mode bosonic Gaussian states.

Parameter estimation for bosonic states has long been studied in quantum metrology, mainly using the tools of (quantum) Fisher information and the (quantum) Cramér-Rao bound~\cite{giovannetti2011advances,monras2013phase,fadel2025quantum}. 
Those tools focus on the asymptotic regime where the number of samples is abundant and the precision is extremely high, so that a local analysis around the ground truth is sufficient. In contrast, quantum learning theory mainly operates in the non-asymptotic regime, which is particularly relevant when there are many underlying model parameters but a limited number of  available samples (see~\cite{wainwright2019high} for discussion).
Nevertheless, quantum metrology and quantum learning theory share similar overarching goals. It is a fruitful direction to further connect these two research thrusts (see e.g.~\cite{chen2026instance,chen2026log} for some attempts along this line).

\paragraph{Emergent classicality from quantum systems.}
There are different notions of emergent classicality in quantum systems. 
A traditional dynamical point of view is that quantum states become classical due to decoherence, induced by interaction with the environment~\cite{zurek2003decoherence}.
More recently, people have investigated notions such as the temperature above which a quantum Hamiltonian exhibits classical properties. As one example, \cite{bakshi2024high} shows that the Gibbs state of a local Hamiltonian becomes separable above a constant threshold of temperature (called the ``sudden death of entanglement''). See also \cite{putterman2026quantum} for the sudden death of other quantum properties including magic and classical tractability.

Specific to bosonic systems, a quantum-optical notion of classicality is given by the non-negativity of the Glauber–Sudarshan P function~\cite{glauber1963coherent,sudarshan1963equivalence}.
A quantum state with a positive P function is a mixture of coherent states with different amplitudes, which means it has no entanglement or squeezing.
Every passive Gaussian state has a non-negative P function.
Also, every bosonic state becomes P-positive if it undergoes an isotropic additive Gaussian noise channel in all quadratures.

Our work introduces a complementary statistical learning-theoretic notion.
We call a quantum state family \emph{classical to learn} if it can be reconstructed using the same sample complexity and ``classical-like'' measurements as for its classical analogue. 
Under this criterion, we see that cold passive Gaussian states are \emph{non-classical to learn}, and they only become classical to learn when above a temperature threshold.
This is in sharp contrast to the P function criterion, which views all passive Gaussian states as classical.

\subsection{Outlook}

Many interesting open problems are left for future work. Here, we list some promising directions.

\paragraph{Closing the $k$-entangled gap.} While we obtain a tight bound
for single-copy measurements in Gaussian state learning, for $k$-entangled measurements we only obtain a lower bound, mainly to show constant $k$ cannot help circumvent the $\Omega(n^3)$ lower bound.
We have not pursued a matching $k$-entangled upper bound, because our focus is on ``classical-like'' measurements. By Theorem~\ref{th:main_lower_k_copy}, improving on the $\Theta(n^3)$ sample complexity requires $k$ to grow with $n$, a regime where $k$-copy measurements become impractical. 
Still, it would be interesting to obtain a tight sample-complexity characterization for $k$-entangled measurements as has been explored in finite-dimensional systems~\cite{chen2024optimal2,pelecanos2025,keskin2026tight,nayak2026optimal}.
Extending those results to the bosonic case seems possible but challenging. One difficulty is that the random purification channel for bosonic systems is much more complicated compared to the finite-dimensional case~\cite{chen2026optimal_merged}. 

\paragraph{Beyond full-state tomography.} The task we study is learning the full description of a general Gaussian state (the only assumption is on the minimal fluctuations). In practical quantum sensing and benchmarking applications, the states to be learned usually come from a much more structured subset, and one may only want to learn a few properties rather than a complete description. It is thus an interesting open direction to generalize the quantum-classical learnability separation beyond full-state tomography. 
{
In particular, while heterodyne measurements are optimal among single-copy measurements for full tomography, they may be far from optimal for structured estimation tasks in the cold regime, where tailored non-Gaussian protocols can provide substantial advantages~\cite{tsang2016quantum}.
}

\section{Technical Summary}
In this section, we give a high-level overview of the techniques used to establish our main results.

\subsection{Encode finite-dimensional states into Gaussian states}\label{sec:techsum_1}

A key ingredient in Theorem~\ref{th:main_lower_k_copy} is a novel protocol that shows that finite-dimensional tomography reduces to passive Gaussian-state tomography. To describe this reduction, we focus on passive Gaussian states, which suffice to establish all of our lower bounds. We write the $n$-mode bosonic Fock space as
$\mc H^{(n)}=\bigoplus_{m=0}^{\infty}\mc H_m^{(n)}$, where
$\mc H_m^{(n)}\cong\mr{Sym}^m(\mbb C^n)$ is the sector with exactly $m$
photons. With respect to this decomposition, every $n$-mode passive Gaussian
state admits a particularly simple parametrization~\cite[Sec.~17.1.4]{DerezinskiGerard2013},
which we refer to as the \emph{fugacity matrix representation}. Namely, for
any $n$-by-$n$ Hermitian matrix ${\sf R}$ satisfying
$0\le{\sf R}<\I_n$, called the \emph{fugacity matrix}, the corresponding
passive Gaussian state is
\bb
    \rho
    =
    \det(\I_n-{\sf R})
    \bigoplus_{m=0}^{\infty}\Gamma_m({\sf R}),
\ee
where
$\Gamma_m({\sf R})
\coleq
{\sf R}^{\otimes m}|_{\mr{Sym}^m(\mbb C^n)}
\in\mc L(\mc H_m^{(n)})$
is the $m$-th symmetric power of ${\sf R}$ acting on the $m$-photon sector.
In particular, $\Gamma_0({\sf R})=1$ and
$\Gamma_1({\sf R})={\sf R}$. Conversely, every passive Gaussian state
admits a unique representation of this form.

Our key observation is that this representation provides a natural
low-energy embedding of finite-dimensional states into passive Gaussian
states. Concretely, for every $n$-dimensional density matrix
$\sigma\in\mc D(\mbb C^n)$, let $\rho_\sigma$ denote the $n$-mode passive
Gaussian state with fugacity matrix ${\sf R}=\sigma/4$. The following three properties make this encoding useful:
\begin{enumerate}
    \item \emph{Low photon number.}
    The encoded state satisfies
    $\Tr(\rho_\sigma\hat n)\le1/3$.
    Moreover, if $\sigma\le2\I_n/n$ (as in the hard family used to prove the finite-dimensional lower bound in \cite{chen2024optimal2,keskin2026tight}), then $\Sigma_{\rho_\sigma}
        \le
        \left(
            \frac12+\frac{1}{2n-1}
        \right)\I_{2n}$, so the passive Gaussian state $\rho_\sigma$ is ``cold''.

    \item \emph{Recover $\sigma$ from $\rho_\sigma$.}
Let $\widehat\rho$ be any $n$-mode state satisfying
$\frac12\|\widehat\rho-\rho_\sigma\|_1\le\varepsilon\le1/100$.
Let $P_1$ denote the projector onto the one-photon sector
$\mc H_1^{(n)}\cong\mbb C^n$. Then the normalized one-photon block
\[
    \widehat\sigma
    \coleq
    \frac{P_1\widehat\rho P_1}
    {\Tr(P_1\widehat\rho)}
\]
is well defined and satisfies
$\frac12\|\widehat\sigma-\sigma\|_1\le16\varepsilon$.

    \item \emph{Simulate $\rho_\sigma$ from copies of $\sigma$.} {There is an algorithm, independent of $\sigma$, that exactly outputs one
copy of $\rho_\sigma$ using a random number $T$ of copies of $\sigma$, with $\E T\le4/9$ (so, only a constant number of copies of $\sigma$ is needed on average); see \cref{lem:exact-passive-simulator}.}
    \end{enumerate}

{These three properties give a direct black-box reduction from
finite-dimensional tomography to passive Gaussian-state tomography.
Suppose that $\mc A$ is a possibly adaptive $k$-copy protocol that learns
every passive Gaussian state in the cold family to $\varepsilon$ trace
distance using at most $N$ copies, with probability at least $2/3$.
Starting from copies of an unknown $n$-dimensional state
$\sigma\le2\I_n/n$, we use the construction of the simulator above to
generate, on demand, the copies of $\rho_\sigma$ required by $\mc A$.
We then run each measurement block of $\mc A$ on the simulated copies and,
from the final estimate $\widehat\rho$, recover an estimate of $\sigma$ by
taking its normalized one-photon block.

The only subtlety is that producing one copy of $\rho_\sigma$ requires a
random number of copies of $\sigma$. Controlling this random cost while
remaining within the allowed measurement model requires some care; the
details are given in the proof of \cref{th:few-copy-passive-lower} in the appendix.
The resulting finite-dimensional protocol uses $O(N)$ copies of
$\sigma$ in total and learns $\sigma$ to $16\varepsilon$ trace distance
with probability at least $1/2$. At the same time, the protocol retains the dependence on $k$:
each measurement performed by $\mc A$ acts on at most $k$ Gaussian
copies, and simulating it may require several finite-dimensional
measurement blocks, since failed batches are discarded and repeated.
Let $t_j$ denote the (random) number of copies of $\sigma$ used in the $j$-th
such block, including failed batches. As shown in the appendix, uniformly over $\sigma\le2\I_n/n$,
\[
    \E \sum_j t_j\min\{\sqrt{t_j},n\}
    =
    O\!\left(N\min\{\sqrt{k},n\}\right).
\]
The left-hand side is precisely the quantity appearing in the
finite-dimensional lower bound of \cite{keskin2026tight}, in the form
stated in \cref{lem:fd-k-copy-lower}. We can therefore apply that result
to the induced finite-dimensional protocol. Since the recovery step
learns $\sigma$ to accuracy $16\varepsilon$, the lemma gives
\[
    \Omega\!\left(\frac{n^3}{(16\varepsilon)^2}\right)
    \le
    O\!\left(N\min\{\sqrt{k},n\}\right).
\]
Absorbing the constant factor into the universal constants and
rearranging yields
\[
    N=\Omega\left(
        \frac{n^3}{\varepsilon^2\min\{\sqrt{k},n\}}
    \right).
\]
The full argument is given in Sec.~\ref{sec:k-copy}.  }\footnote{An earlier version of our proof is based on reduction to the finite-dimensional $t$-copy tomography lower bound in~\cite[Theorem 1.2]{chen2024optimal2}. The main limitation there is that every block of measurements cannot exceed $t$ copies, causing an additional $\log N$ factor in our simulation argument. In contrast, the lower bound of \cite{{keskin2026tight}} only needs to control a total budget of measurement block size, i.e. $\sum_{j=1}^T t_j\min\{\sqrt{t_j},n\}$. This allows us to remove the additional log overhead in our simulation argument.}

\medskip
\noindent \textbf{Note added.} A similar reduction was independently developed in concurrent work \cite[Sec. 9 and Remark 9.7]{rubin26} to establish lower bounds for learning bosonic and fermionic Gaussian states from finite-dimensional quantum state tomography lower bounds.

\subsection{Bayesian posterior anti-concentration for bosons}
\label{sec:techsum_2}

{ The above reduction-based proof already gives the $\Omega(n^3/\varepsilon^2)$
single-copy lower bound for cold Gaussian states. To establish the temperature-dependent lower bound, we analyze single-copy measurements directly on bosonic systems.}

Our proof of the lower bound of Theorem~\ref{th:main_transition} (and also the adaptive part of Theorem~\ref{th:main_lower_1_copy}) follows the Bayesian posterior anti-concentration approach introduced in~\cite{chen2023does}, but needs to address several significant technical challenges to extend this approach to bosonic systems.
At a high level, we design an appropriate prior distribution of passive bosonic Gaussian states of appropriate temperature $\nu$. We then show that, with high probability over a ground-truth state sampled from the prior and over $N$ outcomes from any (potentially adaptive) measurement schedules, the posterior distribution will not concentrate over an $\varepsilon$-trace distance ball centered at the ground-truth state unless $N = \Omega(\frac{n^2}{\nu\varepsilon^2})$, which then implies the desired lower bound of sample complexity.

To be concrete, for $1/n\le\nu\le1$, consider the following family of passive Gaussian states defined via fugacity:
\bb
    \msf R_X\coleq r({\I + \xi X}),\quad\text{where}~r\asymp \frac{\nu}{1+\nu},\quad\xi\asymp \varepsilon/\sqrt{n\nu},
\ee
with $0<\xi\le1/32$, and $X\in\mr{Herm}(\mbb C^n)$ is sampled from the traceless Gaussian Unitary Ensemble conditioned on $\|X\|_\mr{op}\le4$.
This choice ensures every $\rho_{\msf R_X}$ has symplectic eigenvalues lower bounded by $\frac12+\nu$ and hence $\Sigma\ge(\frac12+\nu)\I$.
A key technical Lemma~\ref{le:fugacity_stab} we introduce states that learning $\rho_{\mr R_X}$ to $\varepsilon$ trace distance implies learning $X$ to $O(n)$ trace distance.
Therefore, we just have to show hardness of learning $X$. Intuitively, the space of $X$ has volume of order $n^2$ that needs to be compensated in the posterior tilt.

While $\msf R_X$ looks almost the same as the $n$-dimensional states ensemble used in~\cite{chen2023does} except for the $r$ factor, a crucial difference is that $\rho_{\msf R_X}$ is not linear in $X$, unlike the finite-dimensional case. As a consequence, while in the finite-dimensional case it suffices to consider a linear additive perturbation on the ground-truth (i.e., $\rho_0 + Z$) to prove posterior anti-concentration, such perturbation becomes hard to work with in our case. Instead, we introduce the following relative local parameterization $\Phi$ around any ground-truth $X_0$:
\bb
    \msf R_{\Phi_{X_0}(Z)}\coleq r\,\frac{(\I+\xi {X_0})^{\frac12}(\I+\xi Z)(\I+\xi X_0)^{\frac12}}{\Tr[(\I+\xi X_0)(\I+\xi Z)/n]},
\ee
where $Z\in\mr{Herm}(\mbb C^n)$ is a traceless operator satisfying $\|Z\|_\mr{op}\le 0.1$ denoting the perturbation, and it is clear that $\msf R_{\Phi_{X_0}(0)}=\msf R_{X_0}$. This parameterization allows us to decouple $Z$ from $X_0$ when analyzing the $m$-photon sector, where we can single out the factor of $\Gamma_m(\I+\xi Z)$. We also show that $\Phi$ is a bi-Lipschitz map with bounded Jacobian determinant, so it will not significantly alter the Lebesgue volume that is needed in the proof.

The next crucial step is to control the averaged log likelihood ratio in order to control the posterior tilt. As passive Gaussian states are block-diagonal in total photon number sectors, the optimal measurements can without loss of generality be assumed to first project into a fixed total photon number sector. To bound the likelihood ratio in, say, the $m$-photon sector, we need to compute moments of Haar measure associated with the $m$-th symmetric power representation of $\mbb U(n)$. This is in contrast to the finite-dimensional case, where one only needs to evaluate the moments with the standard representations of $\mbb U(n)$~\cite{chen2023does}. In our proof, we evaluate and bound those moments by connecting them to Schur polynomials, a central object in the representation theory with many useful properties~\cite{fulton2013representation}. We also use the log-Sobolev inequalities for the unitary group~\cite{meckes2013spectral}.
Consequently, we are able to bound the averaged posterior tilt per state copy by $O(n\nu^2\xi^2)=O(\nu\varepsilon^2)$. This implies $N=\Omega(n^2/(\nu\varepsilon^2))$ copies are necessary to ensure a sufficiently large tilt that compensates the $n^2$ prior volume factor.
The complete proof is presented in Sec.~\ref{sec:adaptive}.

\paragraph{Other lower bounds.}
In Sec.~\ref{sec:non-adaptive}, we also present an alternative proof for Theorem~\ref{th:main_lower_1_copy} that only works for \emph{non-adaptive} measurements, but holds even if the spectrum of $\Sigma$ is assumed to be \emph{a priori} known. This proof is also considerably simpler. The main technique is Fano's argument: first construct an $\varepsilon$-net of cold passive Gaussian states of size $2^{\Theta(n^2)}$, then show that with $N$ rounds of non-adaptive measurement one can only extract $N\cdot O(\varepsilon^2/n)$ bits of information from the net, and finally use Fano's inequality to conclude $N=\Omega(n^3/\varepsilon^2)$ is necessary for learning to succeed. Our proof resembles the finite-dimensional tomography bound studied in~\cite{lowe2025lower}, plus that we still need to use Schur polynomials to help evaluate certain moments in a fixed photon number sector, and to carefully control the probability of high-photon-number event.

In Sec.~\ref{sec:warm-lower}, we prove a matching lower bound of $\Omega(n^2/\varepsilon^2)$ for learning Gaussian states with promise $\Sigma\ge (\frac12+\nu)\I$ for arbitrarily large $\nu$ using arbitrary measurements, thus completing the proof for Theorem~\ref{th:main_transition}. The lower bound is based on standard Fano's argument and Holevo's theorem. Similar proof techniques are used in~\cite{chen2026towards,fanizza2025}.

\subsection{Learning warm Gaussian states and a new trace-distance bound}
\label{sec:technical-summary-warm}

We now give a high-level overview of the upper bound in
Theorem~\ref{th:main_transition}. Details are presented in Sec.~\ref{sec:warm-upper}.  In our convention, heterodyne detection on an
$n$-mode Gaussian state $\rho(\mu,\Sigma)$ produces a classical Gaussian sample
$Y\sim\mathcal{N}(\mu,C)$ with covariance $C:=\Sigma+\I/2$~\cite{bittel2025energy,fanizza2025}.
Thus, after the measurement, the statistical problem is classical: from independent samples we estimate
the mean $\mu$ and covariance $C$ of a $2n$-dimensional Gaussian distribution.

The difficulty is to turn these classical estimates back into an accurate quantum state. There are two
related issues. First, simply subtracting the heterodyne noise $\I/2$ from an empirical estimate of $C$
need not produce a valid quantum covariance matrix. Second, the map from a covariance matrix to the
corresponding Gaussian state becomes increasingly sensitive near the pure-state boundary. We deal with
the first problem by modifying the estimator, and with the second by exploiting the warmness promise.

The key idea is to slightly \emph{inflate} the empirical covariance before subtracting the heterodyne
noise. Let $\widehat\mu$ and $\widehat C$ be the empirical mean and centered empirical covariance of
$N$ heterodyne outcomes. Standard Gaussian concentration
\cite[Lemma~9, Eqs.~(91), (94), and (95)]{bittel2025energy} gives
$(1-\zeta)C\preceq\widehat C\preceq(1+\zeta)C$ with high probability, where
$\zeta=O(\sqrt{(n+\log(1/\delta))/N}+(n+\log(1/\delta))/N)$.
Instead of using $\widehat C-\I/2$, we define
$\widehat\Sigma:=\widehat C/(1-\zeta)-\I/2$. On the same event,
\begin{align}
    0
    \preceq
    \widehat\Sigma-\Sigma
    \preceq
    \frac{2\zeta}{1-\zeta}C.
    \label{eq:technical-summary-inflated-estimator}
\end{align}
Thus $\widehat\Sigma\succeq\Sigma$: the estimator can only \emph{add} covariance to the true state.
In particular, $\widehat\Sigma$ is automatically physical. More importantly, this one-sided relation is
exactly the property that allows us to obtain a sharper trace-distance estimate.

Our main technical ingredient is the following bound. Let $\bar\Sigma$ have symplectic eigenvalues
at least $\nu_->1/2$, and set $\gamma_-:=1-(2\nu_-)^{-1}$. Whenever
$\Sigma'\succeq\bar\Sigma$,
\begin{align}
    \frac12
    \bigl\|
        \rho(0,\Sigma')-\rho(0,\bar\Sigma)
    \bigr\|_1
    \le
    \frac{1}{2\sqrt{2\gamma_-}}
    \bigl\|
        \bar\Sigma^{-1/2}
        (\Sigma'-\bar\Sigma)
        \bar\Sigma^{-1/2}
    \bigr\|_{\mathrm F}.
    \label{eq:technical-summary-one-sided-bound}
\end{align}
The important point is that the prefactor depends only on the distance from the pure-state boundary,
and not on the number of modes or on the largest covariance eigenvalue.
For the promise in Theorem~\ref{th:main_transition}, we take $\nu_-=\frac12+\nu$ and hence $\gamma_-=2\nu/(1+2\nu)$.

{
To relate this bound to the classical distance, consider $\rho=\rho(0,\bar\Sigma)$ and $\rho'=\rho(0,\Sigma')$ with $\Sigma'\succeq\bar\Sigma\succeq(\frac12+\nu)\I$.
For these equal-mean Gaussian Wigner distributions, the TV distance is locally comparable to the relative Frobenius norm $\bigl\|
        \bar\Sigma^{-1/2}
        (\Sigma'-\bar\Sigma)
        \bar\Sigma^{-1/2}
    \bigr\|_{\mathrm F}$ with universal constants~\cite[Theorem~1.1]{devroye2018total}.
Meanwhile, Ref.~\cite[Lemma~1]{chen2026towards} shows the Gaussian state trace distance $D_{\mathrm{tr}}(\rho,\rho')$ is no less than and can be a factor of $\sqrt n$ larger than their Wigner TV distance.
Combining these with Eq.~\eqref{eq:technical-summary-one-sided-bound} thus gives (for $D_{\mathrm{tr}}(\rho,\rho')$ smaller than an absolute constant),
\begin{equation}
    \operatorname{TV}(W_\rho,W_{\rho'})\lesssim D_{\mathrm{tr}}(\rho,\rho')
    \lesssim
    \min\!\left\{\sqrt n,\sqrt{1+\nu^{-1}}\right\}
    \,\operatorname{TV}(W_\rho,W_{\rho'}).
\end{equation}
In other words, as the minimal temperature parameter $\nu$ increase, the quantum distance becomes closer to and eventually coincides with the classical distance, up to universal constants.
In our tomography argument, the required covariance ordering is ensured by the inflated estimator in Eq.~\eqref{eq:technical-summary-inflated-estimator}; the mean error is handled separately below.
}

Let us briefly sketch the proof of the trace distance bound and explain why the one-sided assumption helps. 
{ Similar approaches using resolvents, symmetric relative entropy, and Pinsker’s inequality are first explored in~\cite{fanizza2025}. Yet, our tighter analysis is necessary to obtain the optimal temperature-dependent upper bound.}
A faithful Gaussian state can be written as
a Gibbs state of a quadratic Hamiltonian, whose Hamiltonian matrix satisfies
$H(\Sigma)=Jf(\Sigma J)$, where $J=i\Omega$ and
$f(z)=\int_{-1/2}^{1/2}(z-t)^{-1}dt$~\cite[Eq.~(6)]{BBP15}.
After putting the reference covariance $\bar\Sigma$ in Williamson normal form,
$\Sigma'\succeq\bar\Sigma$ becomes positivity of a relative perturbation $E\succeq0$.
This positivity implies that the perturbation cannot reduce the spectral gap controlling the relevant
resolvents, and therefore cannot move the estimate towards the singular pure-state boundary.

Keeping the exact $t$-dependent resolvent gap throughout the integral gives a Hamiltonian stability
bound proportional to $1/\gamma_-$. For a general two-sided perturbation, one instead needs a local
small-error assumption to control the perturbed resolvent, leading to a weaker dependence on the gap.
Finally, the symmetric relative-entropy identity and Pinsker's inequality
\cite[Theorem~5.41]{Wat18} turn the Hamiltonian estimate into
Eq.~\eqref{eq:technical-summary-one-sided-bound}. The complete proof is given in
Section~\ref{sec:one-sided-relative-covariance}.

We can now return to tomography. Under the promise
$\Sigma\succeq(1/2+\nu)\I$, the heterodyne covariance $C=\Sigma+\I/2$ is comparable to $\Sigma$
up to a factor depending only on $\nu$. Therefore
Eq.~\eqref{eq:technical-summary-inflated-estimator} implies that the normalized covariance error
$\Sigma^{-1/2}(\widehat\Sigma-\Sigma)\Sigma^{-1/2}$ has operator norm of order $\zeta$.
Since this is a $2n\times2n$ matrix, its Frobenius norm is larger by at most a factor of order
$\sqrt n$. The one-sided trace-distance bound then shows that the covariance contribution scales as
$O(\sqrt{1+1/\nu}\,\zeta\sqrt n)$. Since
$\zeta=O(\sqrt{(n+\log(1/\delta))/N})$ in the relevant regime, this is the dominant source of error
and leads to the quadratic dependence on the number of modes.

The first moment is cheaper to estimate. The equal-covariance Gaussian fidelity formula
\cite[Eq.~(9)]{BBP15}, together with the Fuchs--van de Graaf inequality
\cite[Theorem~3.36]{Wat18}, gives
$\frac12\|\rho(\widehat\mu,\Sigma)-\rho(\mu,\Sigma)\|_1
\le \frac12\|\Sigma^{-1/2}(\widehat\mu-\mu)\|_2$.
Gaussian concentration makes the right-hand side only
$O(\sqrt{(n+\log(1/\delta))/N})$, which is smaller than the covariance contribution by a factor of
order $\sqrt n$. In particular, no assumption on the size of the displacement $\mu$ is required.

Finally, we insert the intermediate state $\rho(\widehat\mu,\Sigma)$ and use the triangle inequality:
the difference between $\rho(\widehat\mu,\widehat\Sigma)$ and $\rho(\widehat\mu,\Sigma)$ is the
covariance error controlled by our new one-sided bound, while the difference between
$\rho(\widehat\mu,\Sigma)$ and $\rho(\mu,\Sigma)$ is the first-moment error above. Combining the two
contributions gives
\begin{align}
    N
    =
    O\left(
        \left(1+\frac1\nu\right)
        \frac{
            n\bigl(n+\log(1/\delta)\bigr)
        }{\epsilon^2}
    \right).
    \label{eq:technical-summary-warm-sample-complexity}
\end{align}
For constant failure probability and $\nu$ bounded below by an absolute positive constant, this becomes
$O(n^2/\epsilon^2)$, so independent heterodyne measurements achieve the same sample complexity as
learning a classical Gaussian distribution.

Combining this estimate with the general heterodyne upper bound of
Ref.~\cite[Theorem~10]{bittel2025energy} gives
$O(n^2\min\{n,1+1/\nu\}/\epsilon^2)$ samples for constant failure probability. Thus the available
upper bound interpolates between the general $O(n^3/\epsilon^2)$ regime near the vacuum and the
$O(n^2/\epsilon^2)$ warm regime.

\phantomsection
\section*{Acknowledgments} 
\addcontentsline{toc}{section}{Acknowledgments}
We thank Lennart Bittel, Sitan Chen, Yanbei Chen, Marco Fanizza, Yaroslav Herasymenko, Hsin-Yuan Huang, Alfred Li, Zachary Mann, Zhihan Zhang, Andrew Zhao for helpful discussions. 
S.C., F.A.M.\, and J.P.\ acknowledge funding provided by the Institute for Quantum Information and Matter, an NSF Physics Frontiers Center (NSF Grant PHY-2317110), and the
U.S. Department of Energy, Office of Science, National Quantum Information Science Research Centers, Quantum Systems Accelerator. 
A.A.M.\ and F.A.M. thank California Institute of Technology (Caltech) for the hospitality in March 2026, where part of this work was developed with the other authors.  A.A.M. acknowledges support from a 2025 Google PhD Fellowship. A.A.M.\ further acknowledges support from the BMFTR projects DAQC, MuniQC-Atoms, QuSol, Hybrid++, and PasQuops; the Clusters of Excellence ML4Q and MATH+; QuantERA; the Munich Quantum Valley; Berlin Quantum; the Quantum Flagship projects Millenion and Pasquans2; the DFG through CRC 183 and SPP 2514; the European Research Council through DebuQC; and PraktiQOM.

\paragraph{AI usage declaration.}
The conceptualization of the project and the identification and selection of gaps in the literature were carried out by the authors in March 2026. ChatGPT~5.5 was subsequently used to assist in closing technical gaps in preliminary proof attempts. More recently, GPT-5.6 Sol, GPT-6 Astra, and Google DeepReview (via the STOC 2027 pre-submission feedback program) were used to strengthen several of our technical results, and to help check the mathematical arguments and improve the exposition. All results are verified by the authors, who take full responsibility for the content of this work.

\paragraph{Note added.} The analysis of the high-temperature regime in Hamiltonian learning of Gaussian states will be considered, with methods related to our trace distance bound, in an independent work~\cite{Marco_upcoming}. We thank the authors for communicating with us about this.

One day before we submitted v1 of this manuscript, an independent preprint~\cite{rubin26} appeared on arXiv that reports results on single-copy tomography of bosonic and fermionic Gaussian states. For the bosonic part, \cite[Sec.~9]{rubin26} derives related single-copy lower bounds using a reduction from finite-dimensional quantum state tomography similar to ours, and discusses a log-free refinement using weighted block costs in \cite[Remark~9.6]{rubin26}.
Our Section~4 provides a complete log-free treatment of this reduction for arbitrary $k$-copy measurements. We additionally establish matching temperature-dependent upper and lower bounds for bosonic Gaussian tomography and the resulting quantum-to-classical crossover, which are not contained in \cite{rubin26}. For the fermionic part, \cite[Theorem~1.1]{rubin26} establishes an $\Omega(n^3/\varepsilon^2)$ lower bound for adaptive single-copy tomography of mixed $n$-mode fermionic Gaussian states, separating its sample complexity from the known optimal sample complexity of $\Theta(n^2/\varepsilon^2)$ via entangled measurements. Our work contains no results about fermions.

\newpage

\section{Additional notations and preliminaries}\label{sec:pre}

\paragraph{Facts from representation theory.} For $X\in\mr{GL}_n(\mbb C)$, we denote by $\Gamma_m(X)$ the symmetric power representation of $X$ on $\mc H^{(n)}_m\equiv\mr{Sym}^m(\mbb C^n)$, and by $\Gamma(X)$ the direct sum of actions for all $m$, i.e., $\Gamma(X)\coleq\bigoplus_{m=0}^\infty\Gamma_m(X)$. In particular, for $U\in\mbb U(n)$, $\Gamma(U)$ is the passive Gaussian unitary associated with $U$, whose action on the annihilation operators is specified by $\Gamma(U)^\dagger \hat a_i \Gamma(U) = \sum_{j=1}^n U_{ij} \hat a_j$. The projector onto $\mc H_m^{(n)}$ is given by $P_m^{(n)} = \Gamma_m(\I)$.

More generally, let $\lambda$ be a partition of $m$ of length no more than $n$ sorted in a non-increasing order. We use $S^{(n)}_\lambda \subseteq (\mbb C^n)^{\otimes m}$ to denote the Weyl module associated with $\lambda$, which is an irreducible representation of $\mr{GL}_n(\mbb C)$. Formally, $S_\lambda^{(n)}$ can be defined as the image of a Young symmetrizer $c_{\lambda}$ associated with a tableau of shape $\lambda$ in $(\mbb C^n)^{\otimes m}$. See~\cite[Lecture 6]{fulton2013representation}. We denote by $\Gamma_\lambda(X)$ the natural induced action of $X\in\mr{GL}_n(\mbb C)$ on $S_\lambda^{(n)}$, and by $P_\lambda^{(n)}=\Gamma_\lambda(\I)$ the projector onto $S_\lambda^{(n)}$.

\paragraph{Facts about passive Gaussian states.}
We use $\tau_\nu$ to denote a single-mode thermal state with average photon number $\nu$. In the Fock basis, $\tau_\nu$ can be expressed as
\bb\label{eq:thermal_normal_form}
\tau_\nu = \frac{1}{\nu+1} \sum_{m=0}^\infty \left(\frac{\nu}{\nu+1}\right)^m |m\rangle\langle m|.
\ee

\medskip
\noindent We will use the photon-number sector decomposition of the $n$-mode bosonic Hilbert space,
\bb\label{eq:photon-number-sector-decomposition}
\mathcal{H}^{(n)} = \bigoplus_{m=0}^\infty \mathcal{H}^{(n)}_m,
\ee
where $\mathcal{H}^{(n)}_m\cong\mr{Sym}^m(\mbb C^n)$ is the subspace of the Hilbert space with exactly $m$ photons. The superscript $(n)$ is sometimes omitted when there is no ambiguity. 

\medskip
\noindent A passive linear optical transformation, or a passive Gaussian unitary, is described by a unitary $U=(u_{ij})\in\mr{U}(n)$ with the following action on the annihilation operators:
\bb\label{eq_U_rep}
\Gamma(U)^\dagger \hat a_i \Gamma(U) = \sum_{j=1}^n u_{ij} \hat a_j, \quad\mr{for~}i=1,\cdots,n.
\ee
Viewing this as a representation of $\mbb{U}(n)$ on the $n$-mode bosonic Hilbert space $\mc H^{(n)}$, it is known that Eq.~\eqref{eq:photon-number-sector-decomposition} is the decomposition into irreducible representations. That is, $\mbb U(n)$ acts irreducibly on each photon-number sector $\mathcal{H}^{(n)}_m$.

\medskip
\noindent A passive Gaussian state is a Gaussian state with zero mean and covariance matrix $\Sigma$ satisfying $\Sigma=KDK^T$, where $K$ is a symplectic orthogonal matrix and $D=\mr{diag}\{d_1,\cdots,d_n,d_1,\cdots,d_n\}\ge 1/2\I$. The average total photon number of the state is given by $\sum_{i=1}^n (d_i-1/2)$. All passive Gaussian states are block-diagonal in the photon-number sector decomposition. That is, $\rho=\bigoplus_{m=0}^\infty p_m \rho_m$, where $\rho_m$ is a density matrix supported on $\mathcal{H}^{(n)}_m$ and $p_m$ is the probability of having $m$ photons.
One way to see this is to note that any passive Gaussian state can be generated by applying a passive Gaussian unitary to a product of single-mode thermal states. The latter are block-diagonal in the photon-number sector decomposition, and the passive Gaussian unitary preserves the total photon number.

\medskip\noindent An alternative way to represent a passive Gaussian state is via the \emph{photon number occupation} matrix $\msf N\in\mbb C^{n\times n}$, defined as~\cite{holevo1999evaluating}
\bb
    \msf N_{ij}\coleq \Tr(\rho \hat a_j^\dagger\hat a_i)
\ee
A simple relation exists between the photon number occupation matrix $\msf N$ and the covariance matrix $\Sigma$:
\bb
    \Sigma = \frac12\I_{2n}+\Phi_\mr{re}(\msf N) \equiv \frac12\I_{2n}+\begin{pmatrix}
        \mr{Re}\,\msf N &-\mr{Im}\,\msf N
        \\ \mr{Im}\,\msf N &\mr{Re}\,\msf N
    \end{pmatrix}.
\ee
We refer to $\Phi_\mr{re}:\mbb C^{n\times n}\to \mbb R^{2n\times 2n}$ as the \emph{realification} map. It is also the natural isomorphism from $\mbb U(n)$ to $\mr{SP}(2n)\cap \mr O(2n)$. The correctness of the above equation can be verified by noticing that $\expval{\hat a_i\hat a_j} =\expval{\hat a_i^\dagger\hat a_j^\dagger}=0$ for all passive Gaussian states and then directly using the relation between $(\hat p_i,\hat q_i)$ and $\hat a_i$.
The fugacity matrix $\msf R$ is related to these representations by $\msf N=\msf R(\I_n-\msf R)^{-1}$, or equivalently $\msf R=\msf N(\I_n+\msf N)^{-1}$.

\section{\(k\)-copy lower bound for cold Gaussian states}\label{sec:k-copy}

The goal of this section is to establish \Cref{th:main_lower_k_copy} from the main text, whose proof is given in \Cref{th:few-copy-passive-lower} below.
The proof is
a black-box reduction from finite-dimensional tomography to passive
Gaussian-state tomography. The reduction has three ingredients. First, we
encode an \(n\)-dimensional state \(\sigma\) into a passive Gaussian state
\(\rho_\sigma\) with low mean total photon number. Second, we show that
\(\sigma\) can be stably recovered from \(\rho_\sigma\). Third, we show how
to generate one copy of \(\rho_\sigma\) exactly from a random number of
copies of \(\sigma\). { To combine these ingredients, we group independent simulation trials into
batches, with the finite-dimensional input size chosen before each batch is measured. A $k$-copy Gaussian learner using $N$ copies then induces a finite-dimensional learner with $O(N)$ total copies and with a controlled ``expected weighted block cost'' (as defined below in Lemma~\ref{lem:fd-k-copy-lower}). The result follows from
the finite-dimensional lower bound in \cite[Eq.~(40)]{keskin2026tight} (recalled below), which is a refinement of~\cite{chen2024optimal2}.}
\footnote{See independent concurrent work \cite[Sec. 9 and Remark 9.7]{rubin26} for a similar reduction.}

We work with the \(k\)-copy measurement model for quantum state tomography
introduced in~\cite{chen2024optimal2}, defined as follows. The input copies
are processed in blocks of at most \(k\) fresh copies. At each round, the
protocol applies a POVM to the current block, and the choice of POVM may
depend on all previous classical outcomes. No quantum system is retained
from one round to the next. Thus, the protocol may be arbitrarily adaptive,
but only through classical information. This coincides with the model of separate, adaptively chosen
$t$-entangled measurements in \cite[Theorem~1.2]{chen2024optimal2}
and with \cite[Definition~1.1]{keskin2026tight}.

{
We use the following weighted form of the finite-dimensional lower bound
in~\cite[Eq.~(40)]{keskin2026tight}.

\begin{lemma}[Finite-dimensional lower bound with variable block sizes]
\label{lem:fd-k-copy-lower}
There exist universal constants $\varepsilon_{\mathrm{fd}}>0$ and
$d_0\in\mbb N$ such that the following holds. Let $d\ge d_0$ and
$0<\varepsilon<\varepsilon_{\mathrm{fd}}$.
Consider a protocol with a deterministic total copy budget $M\in\mbb N$.
In round $j$, it applies a POVM to an integer number $t_j\ge1$ of fresh
copies. Both $t_j$ and the POVM are chosen before accessing those copies,
using only previous classical outcomes and internal randomness.
No quantum memory is retained between rounds. The number of rounds $T$
may be random, but $\sum_{j=1}^T t_j\le M$ almost surely.

Suppose that, for every $\sigma\in\mc D(\mbb C^d)$ satisfying
$\sigma\le2\I_d/d$, the protocol outputs a state $\widehat\sigma$ such that
\[
    \Pr_\sigma\left(
        \frac12\|\widehat\sigma-\sigma\|_1\le\varepsilon
    \right)\ge\frac12.
\]
Define the weight of a block of size $t$ by
$w_d(t)\coleq t\sqrt{\min\{t,d^2\}}$. Then
\bb\label{eq:fd-k-copy-lower-bound}
    \sup_{\substack{\sigma\in\mc D(\mbb C^d)\\\sigma\le2\I_d/d}}
    \E_\sigma\left[\sum_{j=1}^T w_d(t_j)\right]
    =\Omega\left(\frac{d^3}{\varepsilon^2}\right).
\ee
Here $\Pr_\sigma$ and $\E_\sigma$ are taken over the protocol's
measurement outcomes and internal randomness, with $\sigma$ held fixed.
Thus, the supremum is over input states, whereas the expectation is
over runs of the protocol on each fixed state.
The sum includes every executed measurement round.
\end{lemma}
\begin{proof}
Let $\Delta\sim\mu$ have the prior from~\cite[Sec.~6.1]{keskin2026tight},
and set $\sigma_\Delta\coleq\I_d/d+\Delta$. The prior is supported on
traceless Hermitian matrices satisfying
$\|\Delta\|_{\mathrm{op}}\le C\varepsilon/d$ for a universal constant $C$.
Choose $\varepsilon_{\mathrm{fd}}$ so that
$C\varepsilon_{\mathrm{fd}}\le1/2$; then
$\I_d/(2d)\le\sigma_\Delta\le3\I_d/(2d)\le2\I_d/d$.
Run the protocol on $\sigma_\Delta$, and let $Y$ be its full classical
transcript, including internal randomness.

Since the protocol uses at most $M$ copies in total, each executed block
also satisfies $t_j\le M$. Thus it fits Definition~1.1 of~\cite{keskin2026tight} with both the total copy budget and the allowed maximum block
size set to $M$. This allows us to apply Proposition~7.1 of~\cite{keskin2026tight}.
We use its first inequality in Eq.~(57) of~\cite{keskin2026tight}, which retains the actual block
sizes $t_j$ rather than replacing them by the upper bound $M$.
Combining this with Corollary~6.2 of~\cite{keskin2026tight}, which allows success probability
$1/2$, gives
\[
    \Omega(d^2)\le I(\Delta;Y)
    \le O\left(
        \frac{\varepsilon^2}{d}\,
        \E_{\Delta\sim\mu}\left[
            \E_{\sigma_\Delta}\left[
                \sum_{j=1}^T w_d(t_j)
            \right]
        \right]
    \right).
\]
Here only the outer expectation averages over the prior; the inner
expectation is over the protocol run on the fixed state $\sigma_\Delta$.
Since every $\sigma_\Delta$ satisfies the promise, the prior average
is at most the supremum in \eqref{eq:fd-k-copy-lower-bound}.
Rearranging proves the claim.
\end{proof}
}

We now develop our argument, showing how $n$-dimensional mixed state tomography reduces to $n$-mode passive Gaussian-state tomography. Let us start with the following characterization of passive Gaussian states
(see also~\cite[Sec.~17.1.4]{DerezinskiGerard2013}):

\begin{lemma}[Fugacity matrix representation]\label{le:fugacity}
There is a bijection between $n$-mode passive Gaussian states $\rho$ and
$n$-by-$n$ Hermitian matrices $0\le{\sf R}<\I$.

Given ${\sf R}$, which we call the \emph{fugacity matrix}, the corresponding
passive Gaussian state is
\bb
    \rho
    =
    \det(\I-{\sf R})\Gamma({\sf R}),
\ee
where $\Gamma(X)\coleq\bigoplus_{m=0}^{\infty}\Gamma_m(X)$ and
$\Gamma_m(X)\coleq X^{\otimes m}|_{\mr{Sym}^m(\mbb C^n)}$, with
$\Gamma_0(X)\coleq1$.

Conversely, if $P_0$ and $P_1$ denote the projectors onto the vacuum and
one-photon sectors, respectively, then
\[
    {\sf R}
    =
    \frac{P_1\rho P_1}{\Tr(P_0\rho)},
\]
where the one-photon sector is naturally identified with $\mbb C^n$.

Moreover, the mean total photon number of $\rho$ is
\bb
    \Tr(\rho\hat n)
    =
    \Tr[{\sf R}(\I-{\sf R})^{-1}]\,.
\ee
\end{lemma}

\begin{proof}
Consider the spectral decomposition
${\sf R}=U\mr{diag}(r_1,\ldots,r_n)U^\dagger$, where
$r_j\in[0,1)$ for every $j$. Using the multiplicativity of $\Gamma$ and
the action of $\Gamma(\mr{diag}(r_1,\ldots,r_n))$ on the Fock basis, we obtain
\bb
    \det(\I-{\sf R})\Gamma({\sf R})
    &\eqt{(i)}
    \Gamma(U)
    \left(
        \sum_{\bm m\in\mbb N^n}
        \prod_{j=1}^n(1-r_j)r_j^{m_j}
        \ketbra{\bm m}{\bm m}
    \right)
    \Gamma(U)^\dagger
    \\
    &\eqt{(ii)}
    \Gamma(U)
    \bigotimes_{j=1}^n
    \left(
        (1-r_j)
        \sum_{\ell=0}^{\infty}
        r_j^\ell\ketbra{\ell}{\ell}
    \right)
    \Gamma(U)^\dagger
    \\
    &\eqt{(iii)}
    \Gamma(U)
    \bigotimes_{j=1}^n
    \tau_{\frac{r_j}{1-r_j}}
    \Gamma(U)^\dagger .
\ee
Here, (i) uses the spectral decomposition of ${\sf R}$,
$\Gamma(U^\dagger)=\Gamma(U)^\dagger$, and
$\det(\I-{\sf R})=\prod_{j=1}^n(1-r_j)$; (ii) factorizes the sum over the
Fock basis; and (iii) uses the definition of a single-mode thermal state.
Thus, every Hermitian matrix $0\le{\sf R}<\I$ defines a passive Gaussian
state.

Moreover, since $\Gamma(U)$ preserves the total photon number and
$\tau_\nu$ has mean photon number $\nu$, the last line gives
$\Tr(\rho\hat n)=\sum_{j=1}^n r_j/(1-r_j)
=\Tr[{\sf R}(\I-{\sf R})^{-1}]$.

Conversely, as recalled after Eq.~\eqref{eq_U_rep}, every passive Gaussian
state can be written as
\bb
    \rho=\Gamma(U)\bigotimes_{j=1}^n\tau_{\nu_j}\Gamma(U)^\dagger
\ee
for some $U\in\mbb U(n)$ and $\nu_j\ge0$. Define
$r_j\coleq\nu_j/(1+\nu_j)$ and
${\sf R}\coleq U\mr{diag}(r_1,\ldots,r_n)U^\dagger$. Then
$r_j\in[0,1)$ and $r_j/(1-r_j)=\nu_j$, so reversing the calculation above
gives $\rho=\det(\I-{\sf R})\Gamma({\sf R})$.

Finally, since $\Gamma_0({\sf R})=1$ and $\Gamma_1({\sf R})={\sf R}$,
we have $\Tr(P_0\rho)=\det(\I-{\sf R})$ and
$P_1\rho P_1=\det(\I-{\sf R}){\sf R}$. Hence
${\sf R}=P_1\rho P_1/\Tr(P_0\rho)$, which proves uniqueness and therefore
the claimed bijection.
\end{proof}
We now use the fugacity representation to embed $n$-dimensional states into
passive Gaussian states. Our key idea is that any $n$-dimensional density matrix $\sigma$ can be encoded into an $n$-mode passive Gaussian state with low mean photon number, so that the corresponding $n$-mode Gaussian state behaves similarly to the encoded $n$-dimensional state. We will
use the following encoding throughout this section.

\begin{definition}[Encoding of a finite-dimensional state into a passive Gaussian state]
\label{de:gaussian-encoding}
For every $n$-dimensional state $\sigma\in\mc D(\mbb C^n)$, we denote by
$\rho_\sigma$ the $n$-mode passive Gaussian state with fugacity matrix
${\sf R}=\sigma/4$, namely
$\rho_\sigma\coleq\det(\I-\sigma/4)\Gamma(\sigma/4)$.
\end{definition}

This definition is well posed: since $0\le\sigma\le \I$, we have
$0\le\sigma/4<\I$, and hence \Cref{le:fugacity} guarantees that
$\rho_\sigma$ is a valid passive Gaussian state. Moreover, by
\Cref{le:fugacity},
$\Tr(\rho_\sigma\hat n)
=\Tr[(\sigma/4)(\I-\sigma/4)^{-1}]
\le\frac13\Tr\sigma=\frac13$.
Thus, the encoded state has low mean total photon number.

We next show that the encoding can be stably inverted. The normalized state in
the one-photon sector of $\rho_\sigma$ is exactly $\sigma$, and this remains
true up to a controlled error if $\rho_\sigma$ is replaced by a nearby state.

\begin{lemma}[Recovering $\sigma$ from $\rho_\sigma$]
\label{lem:recover-sigma-stability}
Let $\sigma\in\mc D(\mbb C^n)$, let
$\widehat\rho\in\mc D(\mc H^{(n)})$, and let
$0<\varepsilon\le1/100$. If
$\frac12\|\widehat\rho-\rho_\sigma\|_1\le\varepsilon$, then
$\widehat p\coleq\Tr(P_1\widehat\rho)>0$, and the normalized one-photon block
\[
    \widehat\sigma
    \coleq
    \frac{P_1\widehat\rho P_1}{\widehat p}
\]
satisfies
$\frac12\|\widehat\sigma-\sigma\|_1\le16\varepsilon$.

In particular, if $\sigma'\in\mc D(\mbb C^n)$ and
$\frac12\|\rho_\sigma-\rho_{\sigma'}\|_1\le\varepsilon$, then
$\frac12\|\sigma-\sigma'\|_1\le16\varepsilon$.
\end{lemma}

\begin{proof}
Let $p_\sigma\coleq\Tr(P_1\rho_\sigma)$. Since
$\Gamma_1(\sigma/4)=\sigma/4$, we have
$P_1\rho_\sigma P_1=p_\sigma\sigma$, where
$p_\sigma=\frac14\det(\I-\sigma/4)$.

If $\lambda_1,\ldots,\lambda_n$ are the eigenvalues of $\sigma$, then
$\det(\I-\sigma/4)=\prod_j(1-\lambda_j/4)
\ge1-\frac14\sum_j\lambda_j=3/4$. Hence $p_\sigma\ge3/16$.

By H\"older's inequality for Schatten norms,
\[
    \|P_1\widehat\rho P_1-p_\sigma\sigma\|_1
    \le
    \|P_1\|_\infty^2
    \|\widehat\rho-\rho_\sigma\|_1
    \le
    2\varepsilon.
\]
Taking the trace and using $|\Tr X|\le\|X\|_1$ gives
$|\widehat p-p_\sigma|\le2\varepsilon$. Using the assumption $\varepsilon\le1/100$, we have
$\widehat p\ge3/16-2\varepsilon>1/8$, so $\widehat\sigma$ is well defined. Moreover, adding and subtracting $p_\sigma\sigma/\widehat p$, we obtain
\[
\begin{aligned}
    \|\widehat\sigma-\sigma\|_1
    &=
    \left\|
        \frac{P_1\widehat\rho P_1}{\widehat p}
        -
        \sigma
    \right\|_1
    \\
    &\le
    \frac{
        \|P_1\widehat\rho P_1-p_\sigma\sigma\|_1
    }{\widehat p}
    +
    \left|
        \frac{p_\sigma}{\widehat p}-1
    \right|
    \|\sigma\|_1
    \\
    &=
    \frac{
        \|P_1\widehat\rho P_1-p_\sigma\sigma\|_1
        +
        |p_\sigma-\widehat p|
    }{\widehat p}
    \\
    &\le
    \frac{4\varepsilon}{\widehat p}
    \le
    32\varepsilon,
\end{aligned}
\]
where we used $\|\sigma\|_1=1$, the bounds above, and
$\widehat p>1/8$.
Thus $\frac12\|\widehat\sigma-\sigma\|_1\le16\varepsilon$.

The last claim follows by taking $\widehat\rho=\rho_{\sigma'}$, whose
normalized one-photon block is exactly $\sigma'$.
\end{proof}
We now show that the encoded passive Gaussian state $\rho_\sigma$ can be generated exactly from
copies of the original finite-dimensional state $\sigma$. The construction is simple:
we sample a photon number, project the corresponding number of copies of
$\sigma$ onto the symmetric subspace, and repeat only if the projection fails.
The sampling probabilities are chosen so that the first successful trial
outputs exactly $\rho_\sigma$.

{ 
We state the following lemma mainly to present the simulation algorithm in
a clean and self-contained way. The lemma itself is not used as a black
box in the proof of the $k$-copy lower bound; there, we rely on more
detailed properties of the simulation than just the expectation bound
stated below.
\begin{lemma}[Simulate $\rho_\sigma$ with copies of $\sigma$]
\label{lem:exact-passive-simulator}
For every $\sigma\in\mc D(\mbb C^n)$, there is an algorithm
(independent of $\sigma$) that exactly outputs one copy of $\rho_\sigma$
using a random number $T$ of copies of $\sigma$,
with $\E_\sigma T\le\frac49$.
\end{lemma}
}
\begin{proof}
Set $d_\sigma\coleq\det(\I-\sigma/4)$. By the definition of $\rho_\sigma$
and the homogeneity
$\Gamma_\ell(\sigma/4)=4^{-\ell}\Gamma_\ell(\sigma)$,
\bb\label{eq:embedded-passive-state}
    \rho_\sigma
    =
    d_\sigma
    \bigoplus_{\ell=0}^{\infty}
    4^{-\ell}\Gamma_\ell(\sigma).
\ee
More generally, taking the trace of the fugacity representation with
fugacity matrix $x\sigma$ gives
\bb\label{eq_sum}
    \sum_{\ell=0}^{\infty}
    x^\ell\Tr\Gamma_\ell(\sigma)
    =
    \det(\I-x\sigma)^{-1},
    \qquad
    0\le x<1.
\ee
A single trial of the simulator proceeds as follows. Sample an integer
$L\ge0$ according to
$\Pr(L=\ell)=\frac34\,4^{-\ell}$. If $L=0$, regard the trial as successful and output the vacuum without
consuming any copy of $\sigma$. If $L\ge1$, take $L$ fresh copies of $\sigma$ and project them onto
$\mr{Sym}^L(\mbb C^n)$. If the projection succeeds, output the resulting
state, identifying $\mr{Sym}^L(\mbb C^n)$ with the $L$-photon sector
$\mc H_L^{(n)}$; otherwise, discard these copies and start a new trial.

For a fixed value $L=\ell$, the unnormalized post-measurement state is
$\Gamma_\ell(\sigma)$, so the projection succeeds with probability
$\Tr\Gamma_\ell(\sigma)$. Hence, using \eqref{eq_sum} with $x=1/4$, the
success probability of one trial is
\[
    q_\sigma
    =
    \frac34
    \sum_{\ell=0}^{\infty}
    4^{-\ell}\Tr\Gamma_\ell(\sigma)
    =
    \frac{3}{4d_\sigma}.
\]
Conditioned on success, and after discarding the classical value of $L$, the
output state of one trial is
\[
    \frac{1}{q_\sigma}
    \frac34
    \bigoplus_{\ell=0}^{\infty}
    4^{-\ell}\Gamma_\ell(\sigma)
    =
    d_\sigma
    \bigoplus_{\ell=0}^{\infty}
    4^{-\ell}\Gamma_\ell(\sigma)
    =
    \rho_\sigma.
\]
Since $d_\sigma\le1$, we have $q_\sigma\ge3/4$. 

{ Moreover, let $R$ denote the index of the first successful trial, and let
$L_j$ be the number of copies of $\sigma$ sampled in the $j$-th trial.
The total number of copies used by the simulator is therefore
\[
    T=\sum_{j=1}^{R}L_j
    =
    \sum_{j\ge1}L_j\,\mathds{1}_{\{R\ge j\}}.
\]
The event $\{R\ge j\}$ means that the first $j-1$ trials have all failed.
It therefore depends only on the randomness of the preceding trials and
is independent of the fresh length $L_j$ sampled in trial $j$. Hence
\[
\begin{aligned}
    \E_\sigma T
    &=
    \sum_{j\ge1}
    \E_\sigma\!\left[
        L_j\,\mathds{1}_{\{R\ge j\}}
    \right]=
    \sum_{j\ge1}
    \Pr_\sigma(R\ge j)\,\E L.
\end{aligned}
\]
Since every trial succeeds independently with probability $q_\sigma$,
we have
$\Pr_\sigma(R\ge j)=(1-q_\sigma)^{j-1}$. Therefore,
using $\E L=1/3$,
\[
\begin{aligned}
    \E_\sigma T
    &=
    \E L
    \sum_{j\ge1}(1-q_\sigma)^{j-1}
    =
    \frac{\E L}{q_\sigma}=
    \frac{1/3}{3/(4d_\sigma)}
    =
    \frac49 d_\sigma
    \le
    \frac49.
\end{aligned}
\]
Since $q_\sigma\ge3/4$, the first successful trial is reached almost
surely. Together with the calculation above, this proves the claim.}
\end{proof}
We now prove the main result of this section, stated in Theorem~\ref{th:few-copy-passive-lower}, via the black-box reduction below, whose main idea is as follows. Suppose that
$\mc A$ is a $k$-copy protocol that learns every passive Gaussian state in
the promised cold family using $N$ copies. We use $\mc A$ as a subroutine to learn an unknown $n$-dimensional state
$\sigma$, so that the finite-dimensional lower bound in
\Cref{lem:fd-k-copy-lower} can be applied. Starting from copies of
$\sigma$, we generate the copies of $\rho_\sigma$ requested by $\mc A$, run
each measurement block of $\mc A$ on the simulated inputs, and finally recover
$\sigma$ from the one-photon block of the estimate returned by $\mc A$. { The only complication is that generating a copy of the passive Gaussian
state $\rho_\sigma$ requires a random number of copies of $\sigma$.
We address this issue carefully in the proof below.}

Thus, finite-dimensional tomography reduces to passive Gaussian-state
tomography: an unexpectedly efficient Gaussian learner would give an
unexpectedly efficient finite-dimensional learner. The reduction is
summarized in \Cref{fig:fd-gaussian-reduction}.

\begin{figure}[t]
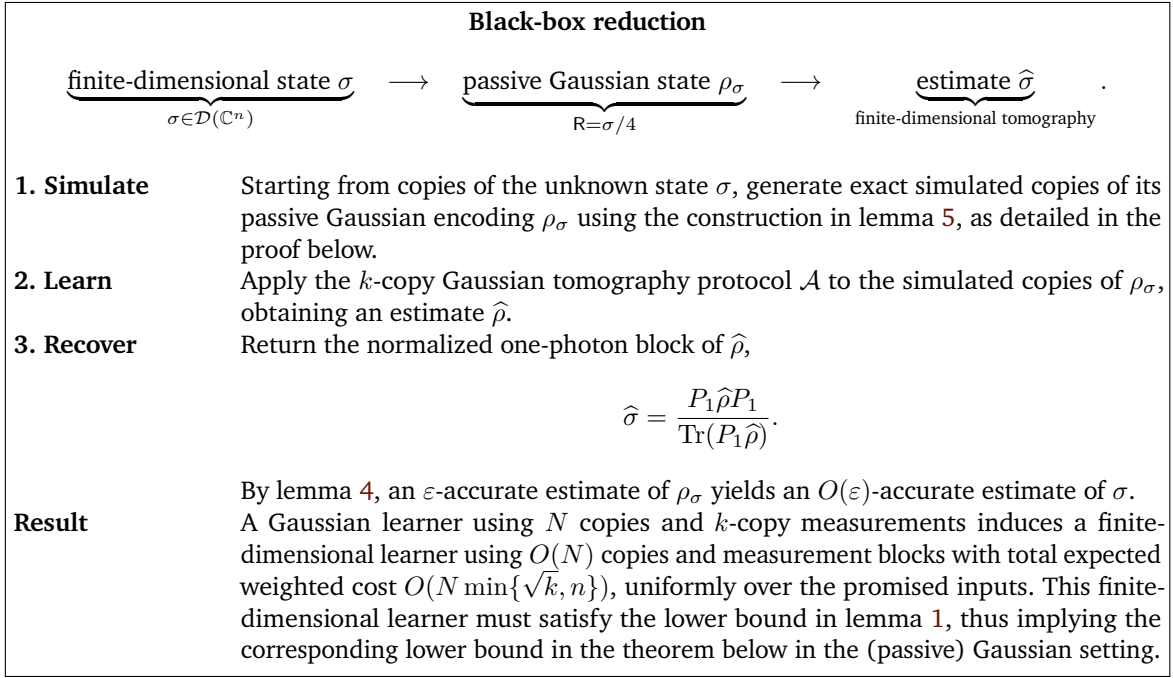

\centering
\fbox{%
\begin{minipage}{0.92\linewidth}
\small

\begin{center}
\textbf{Black-box reduction}
\[
    \underbrace{\text{finite-dimensional state }\sigma}_{%
        \sigma\in\mc D(\mbb C^n)}
    \quad\longrightarrow\quad
    \underbrace{\text{passive Gaussian state }\rho_\sigma}_{%
        {\sf R}=\sigma/4}
    \quad\longrightarrow\quad
    \underbrace{\text{estimate }\widehat\sigma}_{%
        \text{finite-dimensional tomography}} .
\]
\end{center}

\begin{tabularx}{\linewidth}{@{}p{0.17\linewidth}X@{}}

\textbf{1. Simulate}
& { Starting from copies of the unknown state $\sigma$, generate exact simulated
copies of its passive Gaussian encoding $\rho_\sigma$ using the construction
in \Cref{lem:exact-passive-simulator}, as detailed in the proof below.}
\\[6pt]

\textbf{2. Learn}
&
Apply the $k$-copy Gaussian tomography protocol $\mc A$ to the simulated
copies of $\rho_\sigma$, obtaining an estimate $\widehat\rho$.
\\[6pt]

\textbf{3. Recover}
&
Return the normalized one-photon block of $\widehat\rho$,
\[
    \widehat\sigma
    =
    \frac{P_1\widehat\rho P_1}{\Tr(P_1\widehat\rho)}.
\]
By \Cref{lem:recover-sigma-stability}, an $\varepsilon$-accurate estimate
of $\rho_\sigma$ yields an $O(\varepsilon)$-accurate estimate of $\sigma$.
\\[6pt]

\textbf{Result}
&
{ A Gaussian learner using $N$ copies and $k$-copy measurements induces a
finite-dimensional learner using $O(N)$ copies and measurement blocks
with total expected weighted cost $O(N\min\{\sqrt k,n\})$,
uniformly over the promised inputs. This finite-dimensional learner must satisfy the lower bound in
\Cref{lem:fd-k-copy-lower}, thus implying the corresponding lower bound in the theorem below in the (passive) Gaussian setting.}
\\

\end{tabularx}
\end{minipage}%
}
\caption{Main steps of the black-box reduction in \Cref{th:few-copy-passive-lower} from finite-dimensional
tomography to passive Gaussian-state tomography. A tomography protocol for
the encoded passive Gaussian state $\rho_\sigma$ can be used as a subroutine
to learn the underlying finite-dimensional state $\sigma$.}
\label{fig:fd-gaussian-reduction}
\end{figure}

{ \begin{theorem}[Few-copy lower bound for cold Gaussian states]
\label{th:few-copy-passive-lower}
There exist universal constants $\varepsilon_0>0$ and $n_0\in\mbb N$
such that the following holds. Let $n\ge n_0$,
$0<\varepsilon<\varepsilon_0$, and let $k\ge1$ be any integer.
Any possibly adaptive protocol that learns every $n$-mode Gaussian state
to trace-distance error $\varepsilon$ with probability at least $2/3$,
using only $k$-copy measurements, requires
\bb
    N=\Omega\left(
        \frac{n^3}{\varepsilon^2\min\{\sqrt{k},n\}}
    \right)
\ee
copies. This remains true under the promise that the unknown state is
passive and its covariance matrix satisfies
$\Sigma\le(\frac12+\frac1{2n-1})\I_{2n}$.
\end{theorem}}

{  
\begin{proof}
Let $d_0$ and $\varepsilon_{\mathrm{fd}}$ be as in
\Cref{lem:fd-k-copy-lower}. Choose $n_0\ge\max\{d_0,2\}$ and
$0<\varepsilon_0\le1/100$ such that
$16\varepsilon_0\le\varepsilon_{\mathrm{fd}}$.
Let $\mc A$ be a possibly adaptive $k$-copy learner for the promised
passive Gaussian states, using at most $N$ copies and achieving
trace-distance error $\varepsilon$ with probability at least $2/3$.
We use it to learn an unknown $\sigma\in\mc D(\mbb C^n)$ satisfying
$\sigma\le2\I_n/n$, following \Cref{fig:fd-gaussian-reduction}.
First, $\rho_\sigma$ belongs to the promised Gaussian family.
Indeed, its fugacity matrix is ${\sf R}=\sigma/4\le\I_n/(2n)$.
By the thermal normal form in \Cref{le:fugacity}, its ordinary covariance
eigenvalues are $\frac12+r_j/(1-r_j)$, each repeated twice, where
$r_j$ are the eigenvalues of ${\sf R}$. Hence
\[
    \Sigma_{\rho_\sigma}
    \le\left(\frac12+\frac1{2n-1}\right)\I_{2n}.
\]
To simulate $\mc A$, we must choose each finite-dimensional block size
before accessing its fresh inputs. We therefore group the single trials from
\Cref{lem:exact-passive-simulator} into batches.
Recall that one trial samples $L\ge0$ with
$\Pr(L=\ell)=\frac34\,4^{-\ell}$ and projects $L$ fresh copies of $\sigma$ onto
$\mr{Sym}^L(\mbb C^n)$, with $L=0$ counted as a success and producing
the vacuum.
Its success probability is
$q_\sigma=3/[4\det(\I-\sigma/4)]\ge3/4$, and its subnormalized
successful output, averaged over $L$, is $q_\sigma\rho_\sigma$.
All independence statements below are for a fixed input $\sigma$.
Whenever $\mc A$ requests a measurement on $1\le b\le k$ Gaussian
copies, sample $2b$ independent lengths $L_1,\ldots,L_{2b}$ and set
\[
    t\coleq\max\left\{1,\sum_{a=1}^{2b}L_a\right\}.
\]
Take $t$ fresh copies and perform the trials on disjoint groups,
discarding the extra copy when all lengths vanish.
If at least $b$ trials succeed, apply the requested POVM to the first
$b$ successful outputs and pass only its outcome to $\mc A$.
Otherwise, discard the entire batch and repeat without advancing
$\mc A$. Every quantum system is discarded at the end of each batch.
For each preselected length vector, these operations form one POVM
on $t$ fresh copies, so the construction respects the block model.
Let $p_b(\sigma)$ be the probability that a batch succeeds.
Since each trial succeeds with probability
$q_\sigma\ge3/4$, it fails with probability at most $1/4$.
Let $F$ denote the number of failed trials among the $2b$ trials in
the batch. By linearity of expectation,
\[
    \E F=2b(1-q_\sigma)\le\frac b2.
\]
The batch fails precisely when fewer than $b$ trials succeed, or
equivalently when $F\ge b+1$. Markov's inequality therefore gives
\[
    1-p_b(\sigma)
    =
    \Pr(F\ge b+1)
    \le
    \frac{\E F}{b+1}
    \le
    \frac{b}{2(b+1)}
    \le\frac12.
\]
Thus every batch succeeds with probability at least $1/2$.

We next verify that, conditioned on success, the batch supplies exactly
$b$ independent copies of $\rho_\sigma$. Fix a particular success/failure
pattern with $s$ successful trials. For one trial, after averaging over
its sampled length $L$, the subnormalized successful output is
$q_\sigma\rho_\sigma$, while the failure event has total probability
$1-q_\sigma$. Since the $2b$ trials are independent, after discarding
the failed outputs the corresponding subnormalized state on the
$s$ successful outputs is
\[
    q_\sigma^s(1-q_\sigma)^{2b-s}
    \rho_\sigma^{\otimes s}.
\]
Whenever $s\ge b$, we keep the first $b$ successful outputs and discard
the remaining $s-b$, leaving
\[
    q_\sigma^s(1-q_\sigma)^{2b-s}
    \rho_\sigma^{\otimes b}.
\]
Summing this expression over all success patterns with $s\ge b$ simply
sums their probabilities, whose total is $p_b(\sigma)$. Hence the
subnormalized state retained by a successful batch is
\[
    p_b(\sigma)\rho_\sigma^{\otimes b},
\]
and therefore, conditioned on batch success, the retained state is
exactly $\rho_\sigma^{\otimes b}$.
This identity holds after averaging over the sampled lengths; it need
not hold after conditioning on the batch size $t$. We therefore do not
pass the lengths or the trial success flags to $\mc A$: its choices
of block sizes and POVMs use only its own previous measurement outcomes
and internal randomness. Repeating failed batches uses independent fresh
inputs, so the first successful batch has this same conditional output
state. Consequently, the classical record seen by $\mc A$, and hence
its final estimate, has exactly the same distribution as in a run on
fresh copies of $\rho_\sigma$.

We first analyze the exact simulation without imposing a total copy
cutoff. Consider one batch used to simulate a request of $b$ Gaussian
copies. Let
\[
    X\coleq\sum_{a=1}^{2b}L_a,
\]
so that, by construction, $t=\max\{1,X\}$. Since the $L_a$'s are
independent and satisfy
$\E L=1/3$ and $\var(L)=4/9$, we have
\[
    \E X=\frac{2b}{3},
    \qquad
    \E X^2
    =
    \var(X)+(\E X)^2
    =
    \frac{8b}{9}+\frac{4b^2}{9}.
\]
Moreover, $t\le1+X$ and $t^2\le1+X^2$. Therefore,
\[
    \E t
    \le
    1+\frac{2b}{3}
    \le
    \frac{5b}{3},
    \qquad
    \E t^2
    \le
    1+\frac{4b^2+8b}{9}
    \le
    \frac{7b^2}{3},
\]
where the last inequalities use $b\ge1$.

We next bound the weighted block cost appearing in
\Cref{lem:fd-k-copy-lower}. Recall that
\[
    w_n(t)
    =
    t\sqrt{\min\{t,n^2\}}
    =
    \min\{t^{3/2},nt\}.
\]
Using separately the two terms in this minimum gives
\[
    \E w_n(t)
    \le
    \min\{\E t^{3/2},\,n\E t\}.
\]
Since $x\mapsto x^{3/4}$ is concave, Jensen's inequality gives
\[
    \E t^{3/2}
    =
    \E (t^2)^{3/4}
    \le
    (\E t^2)^{3/4}
    \le
    C b^{3/2},
\]
while the bound on $\E t$ gives
\[
    n\E t\le Cnb.
\]
Combining the two estimates,
\[
    \E w_n(t)
    \le
    C\min\{b^{3/2},nb\}
    =
    Cb\sqrt{\min\{b,n^2\}},
\]
where $C>0$ is a universal constant, which may change from line to line.

These estimates concern a single batch. A request for $b$ Gaussian copies
may require several independent batches before one succeeds. For the
analysis, view these batches as an infinite independent sequence, of
which only the batches up to the first success are actually performed.
Let $R_b$ be the index of the first successful batch, and let $t_\ell$
be the size of batch $\ell$ in this sequence. Here $\E_\sigma$ averages
over the simulation's randomness and measurement outcomes for a fixed
input $\sigma$, while $\E t$ and $\E w_n(t)$ average over the lengths
sampled for one fresh batch.

Since $\{R_b\ge\ell\}$ means that the first $\ell-1$ batches have failed,
it depends only on the previous batches and is therefore independent of
the fresh batch size $t_\ell$. Hence
\[
\begin{aligned}
    \E_\sigma\left[\sum_{\ell=1}^{R_b} t_\ell\right]
    &=
    \sum_{\ell\ge1}
    (1-p_b(\sigma))^{\ell-1}\E t
    =
    \frac{\E t}{p_b(\sigma)}
    \le 2\E t,
\end{aligned}
\]
and similarly
\[
\begin{aligned}
    \E_\sigma\left[\sum_{\ell=1}^{R_b} w_n(t_\ell)\right]
    &=
    \sum_{\ell\ge1}
    (1-p_b(\sigma))^{\ell-1}\E w_n(t)
    =
    \frac{\E w_n(t)}{p_b(\sigma)}
    \le 2\E w_n(t).
\end{aligned}
\]
In summary, failed batches increase both the expected number of copies
and the expected weighted cost of simulating one request by at most a
factor of two.

We now sum these bounds over all the Gaussian measurement blocks requested
by $\mc A$. Let $b_i$ denote the number of Gaussian copies used in the
$i$-th request. Since $\mc A$ is a $k$-copy protocol using at most $N$
copies in total, every request actually made satisfies $1\le b_i\le k$,
and $\sum_i b_i\le N$ on every run.
In particular, $\mc A$ makes at most $N$ requests.

For each request $i$, let $S_i$ be the total number of
finite-dimensional copies of $\sigma$ used to simulate that request,
including all failed batches before the first successful one. Similarly,
let $W_i$ be the sum of the weighted costs $w_n(t)$ of all batches used
for that request. For bookkeeping, set $b_i=S_i=W_i=0$ after $\mc A$
terminates, and let all sums over $i$ run from $1$ to $N$.

The size $b_i$ may depend on previous outcomes and internal randomness.
Let $H_i$ denote the classical information available after $\mc A$
has chosen its $i$-th block and before its simulation starts, including
any randomness used for that choice. Conditional on $H_i$, the size
$b_i$ and the requested POVM are fixed, and the simulation uses fresh
copies and new independent randomness. The single-request bounds
therefore give
\[
    \E_\sigma
    \!\left[
        S_i
        \,\middle|\,
        H_i
    \right]
    \le
    2\cdot\frac{5}{3}b_i
    =
    \frac{10}{3}b_i,
\]
and
\[
    \E_\sigma
    \!\left[
        W_i
        \,\middle|\,
        H_i
    \right]
    \le
    C b_i\sqrt{\min\{b_i,n^2\}}.
\]
The factor $2$ in the first bound accounts for the possible failed
batches before the first successful one; in the second bound the same
factor has been absorbed into the universal constant $C$.

We can now add these bounds over all requests. Let
\[
    S\coleq\sum_i S_i
\]
be the total number of finite-dimensional copies used by the exact
simulation. By the law of total expectation and the conditional bounds
above,
\[
\begin{aligned}
    \E_\sigma S
    &=
    \E_\sigma\sum_i S_i
    \\
    &=
    \sum_i\E_\sigma\!\left[
        \E_\sigma[S_i\mid H_i]
    \right]
    \\
    &\le
    \frac{10}{3}\,
    \E_\sigma\sum_i b_i
    \\
    &\le
    \frac{10N}{3}.
\end{aligned}
\]
Here the last inequality uses $\sum_i b_i\le N$ on every run.
Similarly, since the total weighted cost is the sum of the $W_i$'s,
\[
\begin{aligned}
    \E_\sigma\sum_j w_n(t_j)
    &=
    \E_\sigma\sum_i W_i
    \\
    &=
    \sum_i\E_\sigma\!\left[
        \E_\sigma[W_i\mid H_i]
    \right]
    \\
    &\le
    C\,\E_\sigma
    \sum_i b_i\sqrt{\min\{b_i,n^2\}}
    \\
    &\le
    C\sqrt{\min\{k,n^2\}}\,
    \E_\sigma\sum_i b_i
    \\
    &\le
    C N\sqrt{\min\{k,n^2\}}.
\end{aligned}
\]
Here the second-to-last inequality uses $b_i\le k$ for every request.
Notice that $i$ indexes the Gaussian measurement blocks requested by
$\mc A$, whereas $j$ indexes all the finite-dimensional batches used to
simulate them, including failed batches. We have therefore proved
\bb\label{eq:batch-total-costs}
    \E_\sigma S
    &\le
    \frac{10N}{3},
    \\
    \E_\sigma\sum_j w_n(t_j)
    &\le
    C N\sqrt{\min\{k,n^2\}}.
\ee
Both bounds hold uniformly for every $\sigma\le2\I_n/n$.
In particular, the total copy cost $S$ is finite almost surely.

The only remaining problem is that the total number $S$ of
finite-dimensional copies is so far controlled only in expectation,
whereas \Cref{lem:fd-k-copy-lower} requires a deterministic total copy
budget. To fix this, set
\[
    M_0\coleq40N.
\]
Before starting each new batch, its size $t$ has already been sampled
and is therefore known. If executing that batch would make the total
number of used copies exceed $M_0$, we abort without taking those copies.

To compare this capped protocol with the exact simulation, run the two
protocols using the same random choices and measurement outcomes up to
the possible abort. If the capped protocol aborts, then the exact
simulation would have used more than $M_0$ copies. Markov's inequality
and \eqref{eq:batch-total-costs} therefore give
\[
    \Pr_\sigma(\mathrm{abort})
    \le
    \Pr_\sigma(S>M_0)
    \le
    \frac{\E_\sigma S}{M_0}
    \le
    \frac{10N/3}{40N}
    =
    \frac1{12}.
\]
The cutoff cannot increase the weighted cost: under the coupling above,
the capped protocol performs exactly the same batches as the exact
simulation until it possibly aborts, and after an abort it simply stops.
Since all weights $w_n(t_j)$ are nonnegative, the weighted cost of the
capped protocol is therefore no larger than that of the exact simulation
on every run. Taking expectations and using the bound above gives
\[
    \E_{\sigma,\mathrm{capped}}
    \sum_j w_n(t_j)
    \le
    \E_{\sigma}
    \sum_j w_n(t_j)
    \le
    C N\sqrt{\min\{k,n^2\}}.
\]
Moreover, under the coupling above, the capped and exact simulations are
identical whenever no abort occurs. The exact simulation reproduces
the output distribution of $\mc A$, and hence returns an
estimate $\widehat\rho$ satisfying
$D_{\rm tr}(\widehat\rho,\rho_\sigma)\le\varepsilon$ with probability at
least $2/3$. Therefore, the capped protocol can fail to produce such an
estimate only if either the exact simulation fails or an abort occurs.
Using the union bound and $\Pr_\sigma(\mathrm{abort})\le1/12$, its success
probability is at least
\[
    1-\frac13-\frac1{12}
    =
    \frac7{12}
    >
    \frac12.
\]
Whenever the simulation finishes with an estimate $\widehat\rho$, we
return the normalized one-photon block
\[
    \widehat\sigma
    \coleq
    \frac{P_1\widehat\rho P_1}
         {\Tr(P_1\widehat\rho)}.
\]
On abort, or if the denominator vanishes, we return an arbitrary fixed
state in $\mc D(\mbb C^n)$. This is classical postprocessing of the
description of $\widehat\rho$ and requires no additional copies.
By \Cref{lem:recover-sigma-stability}, whenever $\widehat\rho$ is within
trace distance $\varepsilon$ of $\rho_\sigma$, the denominator is positive
and $\widehat\sigma$ is within trace distance $16\varepsilon$ of $\sigma$.
Thus the capped finite-dimensional protocol learns every
$\sigma\le2\I_n/n$ to error $16\varepsilon$ with probability at least
$7/12>1/2$.

We can therefore apply \Cref{lem:fd-k-copy-lower}. Indeed, the capped
protocol has a deterministic total copy budget $M_0$, chooses each
block size and POVM before accessing its fresh copies, and retains no
quantum memory between blocks. Moreover, by our initial choice of
$n_0$ and $\varepsilon_0$, we have $n\ge d_0$ and
$16\varepsilon<16\varepsilon_0\le\varepsilon_{\mathrm{fd}}$. Hence
\[
    \Omega\left(
        \frac{n^3}{(16\varepsilon)^2}
    \right)
    \le
    \sup_{\substack{
        \sigma\in\mc D(\mbb C^n)\\
        \sigma\le2\I_n/n
    }}
    \E_{\sigma,\mathrm{capped}}
    \sum_j w_n(t_j).
\]
Combining this lower bound with the weighted-cost upper bound obtained
above gives
\[
    \Omega\left(
        \frac{n^3}{(16\varepsilon)^2}
    \right)
    \le
    C N\sqrt{\min\{k,n^2\}}.
\]
Since
$\sqrt{\min\{k,n^2\}}=\min\{\sqrt{k},n\}$, rearranging yields
\[
    N
    =
    \Omega\left(
        \frac{n^3}
        {\varepsilon^2\min\{\sqrt{k},n\}}
    \right),
\]
as claimed.
\end{proof}
}

{ Setting $k=1$ gives $N=\Omega(n^3/\varepsilon^2)$, proving the adaptive
part of Theorem~\ref{th:main_lower_1_copy}. The next section establishes the dependence on temperature by analyzing single-copy measurements directly (without the reduction with finite-dimensional tomography).}

\section{Tight single-copy lower bound for all-temperature Gaussian states}\label{sec:adaptive}
\begin{theorem}\label{th:adaptive-lower}
    For any sufficiently small $\varepsilon>0$, any single-copy {adaptive} scheme that can learn any $n$-mode passive Gaussian state, given the promise that $\Sigma\ge (\frac12 + \nu)\I$ for any $\frac1n \le \nu \le 1$, to trace distance precision $\varepsilon$ with at least $2/3$ probability requires at least $N=\Omega\!\left(\cfrac{n^2}{\nu\varepsilon^2}\right)$ copies.
\end{theorem}

For $\nu$ as small as $1/n$, the above theorem gives a lower bound of $\Omega(n^3/\varepsilon^2)$; For $\nu$ as large as $1$, it gives the classical lower bound of $\Omega(n^2/\varepsilon^2)$; In the intermediate regime $\nu \in (1/n,1)$, it gives a smooth crossover. For all $\nu$, this lower bound is achieved by heterodyne measurement, see Theorem~\ref{th:warm-heterodyne-upper}.

\paragraph{Proof Sketch.}

Our proof follows the Bayesian posterior anticoncentration approach developed in~\cite{chen2023does}.
At a high-level, we design an appropriate prior distribution of passive bosonic Gaussian states. Then we show that, with high probability over a ground-truth state sampled from the prior and over $N$ outcomes from any (potentially adaptive) measurement schedules, the posterior distribution will not concentrate over a $\varepsilon$-trace distance ball centered at the ground-truth unless $N = \Omega(n^2/(\nu\varepsilon^2))$. This then implies the desired learning lower bound. 
Several novel techniques from representation theory and differential geometry are needed to make this argument work in our bosonic Gaussian case, which are not needed in the original finite-dimensional case~\cite{chen2023does}.

\bigskip
\noindent The proof for Theorem~\ref{th:adaptive-lower} will be presented in the following subsections.

\subsection{Definitions}

Throughout the proof, we use $C$ or $C_1,C_2,C_3,\cdots$ or $C', C'', \cdots$ to denote positive constants that are independent of $n$ or $\varepsilon$. We sometimes slightly abuse the notation of $C$'s to denote different constants, but those should be clear from the context and do not affect the final scaling.

\medskip
\noindent We need the following facts from random matrix theory to introduce our prior, similar to \cite{chen2023does,chen2024optimal2}.

\begin{definition}\label{de:GUE} A Gaussian Unitary Ensemble, denoted by $\mr{GUE}(n)$, is a distribution over $n$-by-$n$ Hermitian matrices where a sample $G\sim\mr{GUE}(n)$ has independent Gaussian entries on and above its diagonal: 
    \bb
        &G_{jj} \sim \mc N(0,2/n),\quad&&\forall j\in[n].\\
        &G_{j_1j_2} \sim \mc N(0,1/n) + i\mc N(0,1/n),\quad&&\forall 1\le j_1<j_2 \le n.
    \ee
A traceless Gaussian Unitary Ensemble, denoted by $G'\sim\mr{GUE}^*(n)$, is given by $G'\coleq G - \frac{\Tr G}{n} \I$.
\end{definition}

\begin{lemma}[E.g.~{\cite[Section~2.6.2]{anderson2010introduction}}]\label{le:GUE_fact}
    For $G\sim\mr{GUE}^*(n)$, it holds that $\|G\|_\mr{op}\le 3$ with probability at least $1 - e^{-\Omega(n)}$.
\end{lemma}

\noindent Let $\xi\in(0,1/32)$ be a sufficiently small parameter to be chosen later, which should be understood as of the same order as $\varepsilon/\sqrt{n\nu}$. 
Define the following family of fugacity matrices:
\bb
\msf R_X\coleq r(\I + \xi X)\quad\mr{for}\quad X\in\mr{Herm_0}(n),
\ee
where $\mr{Herm_0}(n)$ denotes the space of $n$-by-$n$ traceless Hermitian matrices. Also define $G_\mr{supp}\coleq\{X\in\mr{Herm_0}(n):\|X\|_\mr{op}\le 4\}$ and $G_\mr{good}\coleq\{X\in\mr{Herm_0}(n):\|X\|_\mr{op}\le 3\}$. The parameter $r$ is chosen as
\bb
    r\coleq \frac{\nu}{(1-4\xi)(1+\nu)}.
\ee
Using $\xi\le\frac1{32}$ and $\nu\le 1$, one can verify that $\frac{\nu}{1+\nu}\I\le R_X\le \frac{2\nu}{1+2\nu}\I$ for all $X\in G_\mr{supp}$. Thus, this family of passive Gaussian states satisfy $(\frac12+\nu)\I \le \Sigma\le (\frac12+2\nu)\I$.
We will also denote $\bar\nu\coleq r/(1-r)$ for later use. It is straightforward to verify that $\nu\le \bar\nu \le \frac43\nu$.

\medskip

Sample $X\sim\mr{GUE}^*(n)$ conditioned on $X\in G_\mr{supp}$. We take the resulting distribution of $X$ as our prior distribution, denoted by $X\sim \mu$. 
The density of $X$ with respect to the Lebesgue measure on $\mr{Herm_0}(n)$ is given by~\footnote{
    Here and below, we slightly abuse notation by using $\mu$ and $\nu_{\bm x}$ to denote both the corresponding distributions and their density. 
    The meaning should be clear from context.
}
\bb
    \mu(X) = Z_\mu^{-1} \exp\left(-\frac{n}{4}\|X\|_\mr{F}^2\right) \mathds 1(\|X\|_\mr{op}\le 4).
\ee
Here $Z_\mu$ is the normalization factor.
Note this induces a prior distribution over passive Gaussian states $\rho_{\msf R_X}$. 
Thanks to Lemma~\ref{le:GUE_fact}, we also have $\Pr_\mu(X\in G_\mr{good})\ge 1-e^{-\Omega(n)}$.
Another fact to be used later is the flatness condition of the density: $Z_{\mu}^{-1} \exp(-4n^2) \le\mu(X)\le Z_{\mu}^{-1}$ for all $X\in G_\mr{supp}$.

A non-entangled adaptive learning protocol on $N$ copies can be specified by a family of sequences of POVMs, $\left\{\left\{M_{x_1}\right\}_{x_1},\cdots,\left\{M^{(\bm x_{<t})}_{x_t}\right\}_{x_t},\cdots,\left\{M^{(\bm x_{<N})}_{x_N}\right\}_{x_N}\right\}$. Here, the $t$-th POVM can depend on the first $t-1$ measurement outcomes $\bm x_{<t}$, indicated by the superscripts. An alternative description is given by the learning tree formalism, see e.g. \cite{chen2023does,chen2022exponential}. We use $\bm x\coleq \bm x_{1:N}$ to denote the complete transcript of measurement outcomes.

Fix an adaptive protocol. Let $\mc T_X(\bm x)$ denote the likelihood of observing $\bm x$ when the underlying state is $\rho_{\msf R_X}$. Let $\nu_{\bm x}$ denote the posterior distribution of $X$ after observing $\bm x$. By Bayes' rule:
\bb
   \nu_{\bm x}(X)
    =
    \frac{
        \mc T_X(\bm x)\mu(X)
    }{
        \int_{G_\mr{supp}}
        \mc T_{X'}(\bm x)\mu(X')\,\mr dX'
    }.
\ee
Finally, for $X_0\in\mr{Herm}_0(n)$, define the trace-norm ball $B_{\mr{tr},a}(X_0)\coleq\{X\in\mr{Herm}_0(n)~:~\|X-X_0\|_\mr{tr}\le a\}$ and the operator-norm ball $B_{\mr{op},a}(X_0)\coleq\{X\in\mr{Herm}_0(n)~:~\|X-X_0\|_\mr{op}\le a\}$. These will later be used to show posterior anticoncentration.

\subsection{A relative local parametrization}
To show posterior anticoncentration around a ground-truth state $\rho_0$ in their good set, \cite{chen2023does} directly upper bound the posterior ratio (for some constant $C>1$):
$$
    \frac{\nu_{\bm x}(B_{\mr{tr},\varepsilon}(\rho_0))}{\nu_{\bm x}(B_{\mr{tr},C\varepsilon}(\rho_0))}
$$
Here the balls centered at $\rho_0$ are taken in the affine space of trace-one Hermitian matrices. The denominator is lower bounded by the posterior mass of additive perturbations $\rho_0 + Z$ with $Z\in B_{\mr{op},C\varepsilon/n}(0)$ and $\rho_0+Z$ in the support of the finite-dimensional prior, where $n$ is the state dimension in this comparison; indeed, $\|Z\|_{\mr{tr}}\le n\|Z\|_{\mr{op}}\le C\varepsilon$. In our case, it is more convenient to work with a form of relative perturbation instead, which will become clear in later sections.
For any $X_0\in G_\mr{supp}$, we define 
\bb
    &\Phi_{X_0}: \{Z\in\mr{Herm}_0(n):1+\xi^2\Tr(X_0Z)/n>0\}\longrightarrow \mr{Herm}_0(n),\\
    s.t.\quad  &\I+\xi\Phi_{X_0}(Z) = \frac{(\I+\xi X_0)^{\frac12}(\I+\xi Z)(\I+\xi X_0)^{\frac12}}{1+ \xi^2\Tr(X_0 Z)/n}.
\ee
This open domain contains $B_{\mr{op},0.1}(0)$, since $1+\xi^2\Tr(X_0Z)/n\ge1-0.4\xi^2>0$ on that ball for $X_0\in G_\mr{supp}$. The denominator normalizes the trace, so $\Phi_{X_0}(Z)$ is traceless whenever $X_0$ and $Z$ are traceless. One can also verify that $\Phi_{X_0}(0) = X_0$.
The following lemma describes the basic properties of this relative parametrization: it is locally bi-Lipschitz and has a controlled Jacobian.
\begin{lemma}\label{le:Phi}
    For any $X_0 \in G_\mr{good}$ and $Z,Z_1,Z_2 \in B_{\mr{op}, b}(0)$ with $b\le0.1$, and for every $p\in[1,\infty]$,
    \bb
        \|\mc D\Phi_{X_0}(Z)[H] - H\|_p \le C\xi \|H\|_p,
    \ee
    for any $H \in \mr{Herm}_0(n)$.  
    Here $\mc D\Phi_{X_0}(Z)[H]\coleq \frac{d}{dt}|_{t=0}\Phi_{X_0}(Z + tH)$. $\|\cdot\|_p$ denotes the Schatten $p$-norm, with $\|\cdot\|_\infty=\|\cdot\|_\mr{op}$.
    Consequently, 
    \bb
        \frac12 \|Z_1 - Z_2\|_p \le \|\Phi_{X_0}(Z_1) - \Phi_{X_0}(Z_2)\|_p \le  2\|Z_1 - Z_2\|_p,
    \ee
    Also,
    \bb
        \Phi_{X_0}(Z) \in G_\mr{supp}.
    \ee
    Finally, the Jacobian determinant of $\Phi_{X_0}$ satisfies
    \bb
        J_{\Phi_{X_0}}(Z)\coleq \left|\det(\mc D\Phi_{X_0}(Z))\right| \ge \exp(-C'n^2).
    \ee
\end{lemma}

\medskip

\begin{proof}
    For notational simplicity, define
    \bb
        S\coleq (\I+\xi X_0),\quad T_Z\coleq(\I+\xi Z), \quad d_Z\coleq 1+\xi^2\Tr(X_0Z)/n.
    \ee
    First compute the derivative,
    \bb
        \mc D\Phi_{X_0}(Z)[H] &= \frac{d}{dt}\bigg|_{t=0} \frac{1}{\xi}\left(d_{Z+tH}^{-1}S^\frac12(T_Z + t\xi H)S^\frac12-\I\right)
        \\&= d_Z^{-1}S^{\frac12}HS^{\frac12} -\frac{\xi d_Z^{-2}}{n}S^{\frac12}T_ZS^{\frac12}\Tr X_0H.
    \ee
    Thus,
    \bb
        \left\|\mc D\Phi_{X_0}(Z)[H] - H\right\|_p &\le \left\|d_Z^{-1}S^{\frac12}HS^{\frac12} - H\right\|_p + \frac{\xi d_Z^{-2}}{n}\|S^{\frac12}T_ZS^{\frac12}\|_p |\Tr X_0H|.
    \ee
    For the first term, note that
    \bb
    |d_Z^{-1} - 1| \le C_0|d_Z - 1| \le C_0\|X_0\|_\mr{op}\|Z\|_\mr{op}\xi^2 \le C_1\xi^2
    \ee
    And, by defining $B = S^\frac12 - \I$, for sufficiently small $\xi$,
    \bb
        \|B\|_\mr{op} &= \|(\I +\xi X_0)^\frac12 - \I\|_\mr{op} \le \xi\|X_0\|_\mr{op} \le C_2\xi.
    \ee
    Therefore, the first term can be upper bounded as
    \bb
        \left\|d_Z^{-1}S^{\frac12}HS^{\frac12} - H\right\|_p &\le \left\|(d_Z^{-1}-1)S^{\frac12}HS^{\frac12} \right\|_p + \left\|S^{\frac12}HS^{\frac12} - H\right\|_p
        \\&\le |d_Z^{-1}-1|\|S\|_\mr{op}\|H\|_p + 2\|H\|_{p}\|B\|_\mr{op} + \|H\|_{p}\|B\|^2_\mr{op}
        \\&\le C_3 \xi \|H\|_p.
    \ee
    For the second term, let $q\in[1,\infty]$ be such that $1/q+1/p=1$, 
    \bb
        \frac{\xi d_Z^{-2}}{n}\|S^{\frac12}T_ZS^{\frac12}\|_p |\Tr X_0H| 
        &\leqt{(i)}  \frac{\xi d_Z^{-2}}{n}\|S^{\frac12}T_ZS^{\frac12}\|_p \|X_0\|_q\|H\|_p 
        \\&\leqt{(ii)}\frac{\xi d_Z^{-2}}{n}n^{\frac1p}\|S^{\frac12}T_ZS^{\frac12}\|_\mr{op}n^\frac1q\|X_0\|_\mr{op}\|H\|_p 
        \\&\le C_4\xi\|H\|_p.
    \ee
    Here (i) is by H\"older's inequality; (ii) uses $\|\cdot\|_p \le n^\frac1p\|\cdot\|_\mr{op}$.
    Combining both terms yields, as claimed,
    \bb
        \|\mc D\Phi_{X_0}(Z)[H] - H\|_p \le C\xi \|H\|_p.
    \ee
    Now, let $\Delta\coleq Z_2-Z_1$. The fundamental theorem of calculus gives that
    \bb
        \Phi_{X_0}(Z_2) - \Phi_{X_0}(Z_1) &= \int_{0}^{1}\mc D\Phi_{X_0}(Z_1+t\Delta)[\Delta]\mr d t
        \\&= \Delta + \int_{0}^{1}(\mc D\Phi_{X_0}(Z_1+t\Delta)[\Delta]-\Delta)\mr d t.
    \ee
    Thus, by the convexity of the Schatten $p$-norm,
    \bb
        \|\Phi_{X_0}(Z_2) - \Phi_{X_0}(Z_1)\|_p \le \int_{0}^{1}\|\mc D\Phi_{X_0}(Z_1+t\Delta)[\Delta]\|_p\mr d t \le (1 + C\xi)\|\Delta\|_p.
    \ee
    \bb
        \|\Phi_{X_0}(Z_2) - \Phi_{X_0}(Z_1)\|_p \ge \|\Delta\|_p - \int_{0}^{1}\|\mc D\Phi_{X_0}(Z_1+t\Delta)[\Delta]-\Delta\|_p\mr d t \ge (1 - C\xi)\|\Delta\|_p.
    \ee
    Take $\xi$ sufficiently small such that $C\xi\le 1/2$ gives as claimed,
    \bb
       \frac12\|Z_2-Z_1\|_p\le \|\Phi_{X_0}(Z_2) - \Phi_{X_0}(Z_1)\|_p \le 2\|Z_2-Z_1\|_p
    \ee
    In particular, this implies
    \bb
        \|\Phi_{X_0}(Z)\|_\mr{op} \le \|\Phi_{X_0}(Z)- \Phi_{X_0}(0)\|_\mr{op} +\|X_0\|_\mr{op} \le 2\|Z\|_\mr{op} + \|X_0\|_\mr{op} < 4,
    \ee
    which implies $\Phi_{X_0}(Z)\in G_\mr{supp}$. Finally, note what we have shown implies
    \bb
        (1-C\xi)\|H\|_2\le\|\mc D\Phi_{X_0}(Z)[H]\|_2 \le (1+C\xi)\|H\|_2,
    \ee
    which equivalently says the singular values of $\mc D\Phi_{X_0}(Z)$ as a linear map on $\mr{Herm}_0(n)$ are within $[1-C\xi, 1+C\xi]$. The Jacobian determinant of $\Phi_{X_0}$ is thus lower bounded by
    $$
        J_{\Phi_{X_0}}(Z) \coleq |\det \mc D\Phi_{X_0}(Z)| \ge (1-C\xi)^{n^2-1} \ge \exp(-C'n^2)
    $$
    uniformly in $Z\in B_{\mr{op},b}(0)$.

\end{proof}

\noindent We define the following ball induced by $\Phi_{X_0}$: $B_{\Phi,b}(X_0)\coleq\{\Phi_{X_0}(Z) : Z\in B_{\mr{op},b}(0)\}.$ By Lemma~\ref{le:Phi}, $B_{\Phi,b}(X_0)\subseteq G_\mr{supp}$ for every $X_0\in G_\mr{good}$.

\subsection{Reduce posterior anti-concentration to likelihood ratio bound}
The following proposition is our starting point for showing posterior anticoncentration. Here, $a$ and $b$ are sufficiently small positive constants that will be chosen later.
\begin{proposition}\label{prop:p2l}
    Fix $X_0\in G_\mr{good}$.   
    For any $X\in G_\mr{supp}$ and outcomes $\bm x$, define the likelihood ratio
    \bb
        L_X(\bm x) \coleq \frac{\mc T_{X}(\bm x)}{\mc T_{X_0}(\bm x)}.
    \ee
    Then, with probability at least $1 - e^{-n^2}$ over $\bm x\sim\mc T_{X_0}$,
    \bb
        \frac{\nu_{\bm x}\!\left(B_{\mr{tr},na}(X_0)\right)}{\nu_{\bm x}\!\left(B_{\Phi,b}(X_0)\right)}
        \le 
        \exp(-\Omega(n^2))\left(\E_{Z\sim\mr{Unif}(B_{\mr{op},b}(0))} L_{\Phi_{X_0}(Z)}(\bm x)\right)^{-1}.
    \ee
\end{proposition}
\begin{proof}
    Write $X\coleq \Phi_{X_0}(Z)$.
    By Bayes' rule, 
    \bb
        \frac{\nu_{\bm x}\!\left(B_{\mr{tr},na}(X_0)\right)}{\nu_{\bm x}\!\left(B_{\Phi,b}(X_0)\right)} &=
        \frac{\int_{B_{\mr{tr}, na}(X_0)}L_X(\bm x)\mu(X)\mr dX}{\int_{B_{\Phi, b}(X_0)}L_X(\bm x)\mu(X)\mr dX}.
    \ee
    For the denominator, Lemma~\ref{le:Phi} gives both injectivity of $\Phi_{X_0}$ on $B_{\mr{op},b}(0)$ and $B_{\Phi,b}(X_0)\subseteq G_\mr{supp}$, so one can change the variable $Z = \Phi_{X_0}^{-1}(X)$. Thus,
    \bb
        \int_{B_{\Phi, b}(X_0)}L_X(\bm x)\mu(X)\mr dX &\eqt{(i)} 
        \int_{B_{\mr{op}, b}(0)}L_{\Phi_{{X_0}}(Z)}(\bm x)\mu(\Phi_{{X_0}}(Z)) J_{\Phi_{X_0}}(Z) \mr dZ
        \\ &\geqt{(ii)} Z_\mu^{-1}\exp(-C_1 n^2) \int_{B_{\mr{op}, b}(0)}L_{\Phi_{{X_0}}(Z)}(\bm x) \mr dZ
        \\ &\eqt{(iii)} Z_\mu^{-1}\exp(-C_1 n^2)\,\mr{Vol}(B_{\mr{op},b})\E_{Z\sim\mr{Unif}(B_{\mr{op},b}(0))} L_{\Phi_{X_0}(Z)}(\bm x). 
    \ee
    Here, (i) uses $\mr dX = J_{\Phi_{X_0}}(Z)\mr dZ$ where $J_{\Phi_{X_0}}(Z)$ is the Jacobian determinant; (ii) uses the support guarantee $\Phi_{X_0}(Z)\in G_\mr{supp}$, the Jacobian determinant lower bound from Lemma~\ref{le:Phi}, and the expression of the prior density; In (iii), $\mr{Vol}(B_{\mr{op},b})$ denotes the Lebesgue volume of the operator-norm ball. We omit the center as the volume is independent of it.
    
    \medskip
    \noindent For the numerator, taking the expectation of $\bm x\sim\mc T_{X_0}$ yields
    \bb
        \E_{\bm x\sim\mc T_{X_0}}\int_{B_{\mr{tr}, na}(X_0)}L_X(\bm x)\mu(X)\mr dX = \int_{B_{\mr{tr}, na}(X_0)}\mu(X)\mr dX = \mu(B_{\mr{tr}, na}(X_0)).
    \ee
    Therefore, Markov's inequality implies that, with probability at least $1-\exp(-n^2)$ over $\bm x\sim\mc T_{X_0}$, it holds that\footnote{Although $B_\mr{tr,na}(X_0)$ is not guaranteed to lie in $G_\mr{supp}$, the definition of $\mu$ already includes a cutoff.}
    \bb
        \int_{B_{\mr{tr}, na}(X_0)}L_X(\bm x)\mu(X)\mr dX &\le \exp(n^2) \mu(B_{\mr{tr}, na}(X_0))
        \\&\le \exp(n^2)Z_\mu^{-1} \mr{Vol}(B_{\mr{tr},na}).
    \ee
    Putting the two parts together,
    \bb
        \frac{\nu_{\bm x}\!\left(B_{\mr{tr},na}(X_0)\right)}{\nu_{\bm x}\!\left(B_{\Phi,b}(X_0)\right)} \le \exp((1+C_1)n^2)\frac{\mr{Vol}(B_{\mr{tr},na})}{\mr{Vol}(B_{\mr{op},b})}\left(\E_{Z\sim\mr{Unif}(B_{\mr{op},b}(0))} L_{\Phi_{X_0}(Z)}(\bm x)\right)^{-1}.
    \ee
    Now, use the following volume ratio bound:\footnote{\cite{chen2023does} proved the bound for GOE. \cite{chen2024optimal2} used the bound for GUE and mentioned that the proof is exactly the same as for GOE.}
    \begin{lemma}[\cite{chen2023does, chen2024optimal2}] For the trace-norm and operator-norm ball defined in $\mr{Herm}_0(n)$,
        \bb
            \frac{\mr{Vol}(B_{\mr{tr},1})}{\mr{Vol}(B_{\mr{op},1/n})} \le \exp(10n^2).
        \ee
    \end{lemma}
    \noindent    By homogeneity,
    \bb
        \frac{\mr{Vol}(B_{\mr{tr},na})}{\mr{Vol}(B_{\mr{op},b})} = \frac{(na)^{n^2-1}}{(nb)^{n^2-1}}\frac{\mr{Vol}(B_{\mr{tr},1})}{\mr{Vol}(B_{\mr{op},1/n})}\le \exp(-(n^2-1)\log(b/a) + 10n^2).
    \ee
    By choosing $b/a$ a sufficiently large constant and putting the above back to the posterior ratio, we conclude that
    \bb
        \frac{\nu_{\bm x}\!\left(B_{\mr{tr},na}(X_0)\right)}{\nu_{\bm x}\!\left(B_{\Phi,b}(X_0)\right)} \le \exp(-\Omega(n^2)) \left(\E_{Z\sim\mr{Unif}(B_{\mr{op},b}(0))} L_{\Phi_{X_0}(Z)}(\bm x)\right)^{-1}.
    \ee
    Recall this holds with probability at least $1-e^{-n^2}$ over $\bm x\sim\mc T_{X_0}$ for fixed $X_0\in G_\mr{good}$.
\end{proof}

\subsection{Likelihood ratio bound -- fixed transcripts}
Thanks to Proposition~\ref{prop:p2l}, we just need $\E_Z L_{\Phi_{X_0}(Z)}(\bm x) \ge \exp(-o(n^2))$ with high probability for posterior anticoncentration around $X_0$, where $Z\sim\mr{Unif}(B_{\mr{op},b}(0))$. We now start to lower bound this average likelihood ratio. Lemma~\ref{le:Phi} ensures that every $\Phi_{X_0}(Z)$ appearing here lies in $G_\mr{supp}$.

Since passive Gaussian states are block-diagonal in total photon number sectors, one can without loss of generality assume all the POVMs have the same block-diagonal structure. Furthermore, we can assume all POVM elements are rank-$1$, for the same reason as before. We thus denote the $t$-th POVM by $\left\{\ketbra{\phi_{z_t,k_t}^{(\bm x_{<t})}}{\phi_{z_t,k_t}^{(\bm x_{<t})}}\right\}_{z_t,k_t}$, with outcome $x_t\coleq(z_t,k_t)$ where $k_t$ indicates the total number of photons and $\ket{\phi_{z_t,k_t}^{(\bm x_{<t})}}\in\mc H_{k_t}^{(n)}$ is an unnormalized vector satisfying $\int \mr d{z_t}\ketbra{\phi_{z_t,k_t}^{(\bm x_{<t})}}{\phi_{z_t,k_t}^{(\bm x_{<t})}} = P_{k_t}^{(n)}$ for all $k_t\in\mbb N$, where we recall $P_{k}^{(n)}$ is the projector onto the $k$-photon sector of the Hilbert space. Again, the superscript of $\bm x_{<t}$ indicates adaptivity.

Now, fix $X_0\in G_\mr{good}$ and a transcript of outcomes $\bm x$. Recall $X\coleq \Phi_{X_0}(Z)$.
\bb
\log \E_Z L_{\Phi_{X_0}(Z)}(\bm x) &= \log \E_Z \prod_{t=1}^{N}\frac{\braket{\phi_{z_t,k_t}^{(\bm x_{<t})}|\rho_{\msf R_{X}}|\phi_{z_t,k_t}^{(\bm x_{<t})}}}{\braket{\phi_{z_t,k_t}^{(\bm x_{<t})}|\rho_{\msf R_{X_0}}|\phi_{z_t,k_t}^{(\bm x_{<t})}}}
\\&\geqt{(i)} \sum_{t=1}^N \E_Z \log\frac{\braket{\phi_{z_t,k_t}^{(\bm x_{<t})}|\rho_{\msf R_{X}}|\phi_{z_t,k_t}^{(\bm x_{<t})}}}{\braket{\phi_{z_t,k_t}^{(\bm x_{<t})}|\rho_{\msf R_{X_0}}|\phi_{z_t,k_t}^{(\bm x_{<t})}}}.
\ee
(i) is by Jensen's inequality. Thus, we just need to bound the average log likelihood ratio for each copy. In the following, we focus on the $t$-th term and temporarily rewrite the POVM as $\{\ketbra{z,k}{z,k}\}_{z,k}$ for notational simplicity. By definition of $\rho_{\msf R}$, we have
\bb
    \E_Z\log\frac{\braket{z,k|\rho_{\msf R_X}|z,k}}{\braket{z,k|\rho_{\msf R_{X_0}}|z,k}} 
    &= 
    \E_Z \log\frac{\det(\I-\msf R_X)}{\det(\I - \msf R_{X_0})} + \E_Z\log\frac{\braket{z,k|\Gamma_k(\msf R_X)|z,k}}{\braket{z,k|\Gamma_k(\msf R_{X_0})|z,k}}.
\ee
Recall that $\msf R_{X_0} = r({\I + \xi X_0})$ ~and~ $\msf R_{X} = r({\I + \xi X}) = r\cfrac{(\I+\xi X_0)^{\frac12}(\I+\xi Z)(\I+\xi X_0)^{\frac12}}{1 + \xi^2\Tr(X_0 Z)/n}$.

\medskip
\noindent \textbf{Bound for the first term.} Notice that (recall we denote $\bar\nu\coleq \frac r{1-r}$),
\bb
    \log\det\!\left(\I - r({\I+\xi X})\right) = n\log(1-r) + \log\det(\I - \bar\nu\xi X).
\ee
Only the second term depends on $X$. We bound it as follows,
\bb
\left|\log\det(\I - \bar\nu\xi X)\right| &\eqt{(i)} \left|\sum_{k=2}^\infty\frac1k (\bar\nu\xi)^k\Tr X^k\right|
\\&\leqt{} \sum_{k=2}^\infty\frac1k (\bar\nu\xi)^k n\|X\|_\mr{op}^k.
\\&\leqt{(ii)} 16n\bar\nu^2\xi^2\sum_{k=0}^\infty(4\bar\nu\xi)^k
\\&= 16n\bar\nu^2\xi^2 (1 - 4\bar\nu\xi)^{-1}
\\&\leqt{(iii)} 32n\bar\nu^2\xi^2.
\ee
Here, (i) uses the Taylor expansion for $\log\det$. The $k=1$ term disappears as $X$ is traceless; (ii) uses Lemma~\ref{le:Phi}, which gives $\|X\|_\mr{op}\le 4$ for $X = \Phi_{X_0}(Z)$ for all $Z$ we consider; (iii) uses the observation that $\bar\nu\xi \le 1/24$. Combining this with the same bound for $X_0\in G_\mr{good}$ and the triangle inequality,
\bb\label{eq:likelihood_first_term}
    \log\frac{\det(\I-\msf R_X)}{\det(\I - \msf R_{X_0})} \ge -64n\bar\nu^2\xi^2.
\ee
Since this bound holds uniformly for all $X = \Phi_{X_0}(Z)$, it also holds in expectation.

\medskip
\noindent \textbf{Bound for the second term.} This term (before averaging over $Z$) can be further decomposed as
\bb
    &\log\frac{\braket{z,k|\Gamma_k(\msf R_X)|z,k}}{\braket{z,k|\Gamma_k(\msf R_{X_0})|z,k}} \\
    =~&\log\frac{\braket{z,k|\Gamma_k(\I+\xi X_0)^{\frac12}\Gamma_k(\I+\xi Z)\Gamma_k(\I+\xi X_0)^{\frac12}|z,k}}{\braket{z,k|\Gamma_k(\I+\xi X_0)|z,k}} -k \log\!\left(1 + \frac{\xi^2}{n}\Tr(X_0Z)\right).
\ee
The second part is non-negative after averaging over $Z\sim\mr{Unif}(B_{\mr{op},b}(0))$,
\bb\label{eq:likelihood_zero_term}
- \E_Z k\log\left(1+\frac{\xi^2}{n}\Tr(X_0Z)\right) &\geqt{(i)} -  k\log\left(1+\frac{\xi^2}{n}\E_Z\Tr(X_0Z)\right) \eqt{(ii)} 0.
\ee
(i) uses the concavity of $\log$; (ii) is because the distribution of $Z$ is centrally symmetric.

\medskip
\noindent To evaluate the first part, define a new normalized vector in $\mc H_{k}^{(n)}$,
\bb
    \ket{\bar\psi}\coleq \frac{\Gamma_k(\I+\xi X_0)^{\frac12}\ket{z,k}}{\sqrt{\braket{z,k|\Gamma_k(\I+\xi X_0)|z,k}}}.
\ee
The first part can then be concisely written as $$\E_Z\log\braket{\bar\psi|\Gamma_k(\I+\xi Z)|\bar\psi}.$$
Let the spectrum decomposition of $Z$ be $Z = UDU^\dagger$ where $D = \diag(\bm d)$. 
We first fix $D$ and average over $U$, which is possible as the distribution of $Z$ is invariant under unitary conjugation. 
Define
\bb
    F(U) = \log\braket{\bar\psi|\Gamma_k(\I+\xi UDU^\dagger)|\bar\psi}.
\ee
We need the following lemma to control the concentration of $F$.
\begin{lemma}[Haar log-Sobolev concentration, modified from \cite{meckes2013spectral}]\label{le:haar_lsi}
    Let $U\sim\mr{Haar}(\mbb U(n))$, and equip $\mbb U(n)$ with the geodesic distance induced by the Frobenius (i.e., Hilbert--Schmidt) norm.
    Suppose that $F:\mbb U(n)\to\mbb R$ is $L$-Lipschitz with respect to this geodesic distance. Then, for every $\lambda>0$,
    \bb\label{eq:haar_herbst}
        \log\E_U\exp\!\left(\lambda(F(U)-\E_UF(U))\right)
        \le \frac{3\pi^2\lambda^2L^2}{4n}.
    \ee
    In particular,
    \bb
        \E_UF(U)\ge \log\E_Ue^{F(U)}-\frac{3\pi^2L^2}{4n}.
    \ee
\end{lemma}
\begin{proof}
    By~\cite[Theorem~15 and the following paragraph]{meckes2013spectral}, Haar measure $\mu$ on $\mbb U(n)$ satisfies
    \bb
        \mr{Ent}_{\mu}(g^2)\le \frac{3\pi^2}{n}\E_{\mu}|\nabla g|^2
    \ee
    for this geodesic metric, where
    \bb
        \mr{Ent}_{\mu}(Y)
        \coleq \E_{\mu}(Y\log Y)-(\E_{\mu}Y)\log(\E_{\mu}Y).
    \ee
    Set
    \bb
        H\coleq F-\E_{\mu}F,
        \qquad
        \psi(\lambda)\coleq\log\E_{\mu}e^{\lambda H}.
    \ee
    Since $H$ differs from $F$ by a constant and $F$ is $L$-Lipschitz, $|\nabla H|=|\nabla F|\le L$. Thus, for $g=e^{\lambda H/2}$,
    \bb
        |\nabla g|^2
        =\frac{\lambda^2}{4}e^{\lambda H}|\nabla H|^2
        \le \frac{\lambda^2L^2}{4}e^{\lambda H}.
    \ee
    Applying the preceding logarithmic Sobolev inequality therefore gives, for $\lambda>0$,
    \bb
        \mr{Ent}_{\mu}(e^{\lambda H})
        \le \frac{3\pi^2\lambda^2L^2}{4n}\E_{\mu}e^{\lambda H}.
    \ee
    On the other hand,
    \bb
        \frac{\mr{Ent}_{\mu}(e^{\lambda H})}{\E_{\mu}e^{\lambda H}}
        &=\lambda\frac{\E_{\mu}(He^{\lambda H})}{\E_{\mu}e^{\lambda H}}
        -\log\E_{\mu}e^{\lambda H}
        \\&=\lambda\psi'(\lambda)-\psi(\lambda).
    \ee
    Combining the last two equations yields
    \bb
        \lambda\psi'(\lambda)-\psi(\lambda)
        \le \frac{3\pi^2\lambda^2L^2}{4n}.
    \ee
    Hence
    \bb
        \left(\frac{\psi(\lambda)}{\lambda}\right)'
        =\frac{\lambda\psi'(\lambda)-\psi(\lambda)}{\lambda^2}
        \le \frac{3\pi^2L^2}{4n}.
    \ee
    Since $\psi(0)=\psi'(0)=0$, integration proves the claim for $\lambda>0$. 
    The final claim is the case $\lambda=1$.
\end{proof}

\medskip

\noindent We now bound the Lipschitz constant for $F$. Let $U_\theta = e^{\theta K} U$ for any $K^\dagger = -K$.  Define
\bb
    B_\theta\coleq U_\theta(\I+\xi D)U_\theta^\dagger,\quad T_\theta\coleq \bra{\bar\psi}\Gamma_k(B_\theta)\ket{\bar\psi}.
\ee
Hence $F(U_\theta) = \log T_\theta$. Since $\Gamma_k$ denotes the $k$-th symmetric tensor power,
\bb
{T}_\theta = \bra{\Psi} B_\theta^{\otimes k} \ket{\Psi},
\ee
where $\ket{\Psi}$ is the first-quantized wave function of $\ket{\bar\psi}$, and dots denote differentiation with respect to $\theta$. Thus,
\bb
|{\dot T}_\theta|  &= |\bra{\Psi}{(B_\theta^{\frac12})}^{\otimes k}(\sum_{j=1}^k \I^{\otimes j-1}\otimes{B_\theta^{-\frac12}\dot B_\theta B_\theta^{-\frac12}}\otimes\I^{\otimes k-j}){(B_\theta^{\frac12})}^{\otimes k}\ket{\Psi}|
\\&\le k\|B_\theta^{-\frac12}\dot B_\theta B_\theta^{-\frac12}\|_\mr{op} \cdot T_\theta.
\ee
Further note that $\dot B_\theta = [K, B_\theta]$ and that the singular values of $B_\theta$ are bounded within, say, $[0.9,1.1]$, 
\bb
\left|\frac{\mr d}{\mr d\theta}F(U_\theta)\right| = \frac{|{\dot T}_\theta|}{T_\theta} \le 2k\|\dot B_\theta\|_\mr{op}\le 4k\xi \|K\|_\mr{op} \le 4k\xi \|K\|_\mr{F}
\ee
The Lipschitz constant for $F$ thus satisfies $L\le 4k\xi$.

It remains for us to bound $\log\E_U e^{F(U)} = \log\E_U\!\braket{\bar\psi|\Gamma_k(\I+\xi UDU^\dagger)|\bar\psi}$. Recall that $(\Gamma_k, \mr{Sym}^k(\mbb C^n))$ is an irreducible representation of $\mbb U(n)$. As a direct consequence of Schur's Lemma and the fact that Schur polynomials give the characters of $\Gamma_k$,
    \bb
            \E_U \braket{\bar\psi|\Gamma_k(U(\I+\xi D)U^\dagger)|\bar\psi} &= \braket{\bar\psi|\frac{\Tr \Gamma_k(\I+\xi D)}{\Tr P_{k}^{(n)}}P_{k}^{(n)}|\bar\psi} = {\bar h_k(\bm 1 + \xi\bm d)}.
    \ee
Here $\bar h_k(\bm y) =h_k(\bm y)/h_k(\bm 1)$ is the normalized complete homogeneous polynomial. Equivalently, $h_k=s_{(k)}$ is a Schur polynomial corresponding to the one-row partition. A useful characterization of the normalized complete homogeneous polynomials is as follows:
\begin{lemma}\cite[Proof of Theorem 220]{hardy1952inequalities}\label{le:dirichlet}
    Let $k\in\mbb N$ and $\bm y = (y_1,\cdots,y_n)$ such that $y_j>0$ for all $j\in[n]$. Let $\bm P \coleq(P_1,\cdots,P_n)\sim\mr{Dirichlet}(1,\cdots,1)$. Equivalently, $\bm P$ is sampled uniformly from the probability simplex. Then it holds that
    \bb
        \bar h_k(\bm y) = \E_{\bm P}(\sum_{j=1}^n P_jy_j)^k \equiv \E_{\bm P}(\bm P\cdot\bm y)^k.
    \ee
\end{lemma}
\noindent Using this lemma, we obtain
\bb
    \E_U e^{F(U)} &= \E_{\bm P}(\bm P\cdot(\bm 1+\xi \bm d))^k \geqt{(i)} ((\E_{\bm P}\bm P)\cdot(\bm 1+\xi \bm d))^k \eqt{(ii)} 1.
\ee
Here (i) uses Jensen's inequality and (ii) uses the fact that $\bm Z$ is traceless. We can now use Lemma~\ref{le:haar_lsi} to conclude
\bb
    \E_U F(U) \ge -\frac{12\pi^2k^2\xi^2}{n},
\ee
This proves our bound for the second term of the fixed-transcript likelihood ratio. Putting the two terms together, we obtain
\begin{proposition}[Likelihood Ratio Bound for fixed transcript]\label{prop:likelihood_fixed}
    Fix $X_0\in G_\mr{good}$. 
    For any transcript of measurement outcomes $\bm x \coleq (\bm z, \bm k)$ where $k_t$ denotes the measured photon number of the $t$-th copy, 
    \bb
        \log\E_Z L_{\Phi_{X_0}(Z)}(\bm x) \ge -C \xi^2\sum_{t=1}^N\left(n\nu^2+\frac{k_t^2}{n}\right),
    \ee
    where $Z\sim\mr{Unif}(B_{\mr{op},b}(0))$ and $C$ is some positive constant.
\end{proposition}

\subsection{Likelihood ratio bound II -- high probability argument}

Now we show a high probability bound for the log-likelihood over measurement transcripts. Note that, although the measurement can be adaptive, the total photon number for each copy is sampled from an i.i.d. distribution. More concretely, let $Y=\Phi_{X_0}(Z)$ for any $Z\in B_{\mr{op},b}(0)$. Let $k$ be a random variable corresponding to the total photon number of $\rho_{\msf R_Y}$. That is, $k\sim q_Y$ where $q_Y(k)\coleq \det(\I-\msf R_Y)\Tr(\Gamma_k(\msf R_Y))$. 

Denote the eigenvalues of $\msf R_Y$ by $\{r_j\}_{j=1}^n$.     
As is clear from the proof of Lemma~\ref{le:fugacity}, $k$ equals to the sum of $n$ independent geometric random variables ${\mr {Geom}(1-r_j)}$. That is,
    \bb
        k = \sum_{j=1}^n M_j,\quad~\text{where}~\Pr(M_j = m) = (1-r_j)r_j^m,\quad m=0,1,2,\cdots
    \ee
Recall that $\frac{\nu}{1+\nu}\le r_j \le \frac{2\nu}{1+2\nu}$. Combining the moments of geometric random variables with the assumption $1/n\le\nu\le1$ gives
\bb
    \E k^2 \le \sum_{j=1}^n\mr{Var}M_j + (\sum_{j=1}^n\E M_j)^2 \le 12 \nu^2 n^2.
\ee

\noindent 
Combining this with Proposition~\ref{prop:likelihood_fixed} gives us the following:
\begin{proposition}[High-probability Likelihood Ratio Bound]\label{prop:likelihood_hp}
    Fix $X_0\in G_\mr{good}$ and $\delta \in (0,1)$.
    With probability at least $1 - \delta$ over $\bm x\sim\mc T_{X_0}$,
    \bb
        \E_Z L_{\Phi_{X_0}(Z)}(\bm x) \ge \exp\left(-\frac{C''}{\delta}{Nn\nu^2\xi^2}\right),
    \ee
    where $C''$ is some positive constant.
    \begin{proof}
        Apply Markov's inequality to Proposition~\ref{prop:likelihood_fixed} and then exponentiate both sides.
    \end{proof}
\end{proposition}

\subsection{Putting everything together}
\begin{proposition}[Posterior anti-concentration]\label{prop:anti-concentration}
    Assume $ N = o(n/(\nu^2 \xi^2))$.
    Then, for any $X_0\in G_\mr{good}$ and a small absolute constant $\delta_0\in(0,\,0.01]$, with probability at least $1 - e^{-n^2} - \delta_0$ over $\bm x\sim\mc T_{X_0}$,
    \bb
        {\nu_{\bm x}\!\left(B_{\mr{tr},na}(X_0)\right)}
        \le 
        \exp(-\Omega(n^2)).
    \ee
\end{proposition}
\begin{proof}
    Combining Proposition~\ref{prop:p2l} and Proposition~\ref{prop:likelihood_hp} with a union bound, the following holds with probability at least $1 - e^{-n^2} - \delta_0$ for any fixed $X_0\in G_\mr{good}$.
    \bb
        \frac{\nu_{\bm x}\!\left(B_{\mr{tr},na}(X_0)\right)}{\nu_{\bm x}\!\left(B_{\Phi,b}(X_0)\right)} &\le \exp\left(-\Omega(n^2) + \frac{C''}{\delta_0}Nn\nu^2\xi^2 \right)
        \\&= \exp\left(-\Omega(n^2) + \frac{C''}{\delta_0}o(n^2) \right)
        \\&= \exp\left(-\Omega(n^2)\right).
    \ee
    The proposition follows by relaxing the denominator ${\nu_{\bm x}\!\left(B_{\Phi,b}(X_0)\right)}\le1$.
\end{proof}

\begin{proposition}[Hardness of learning $X_0$]\label{prop:hardness_X}
    Let $\{\rho_{\msf R_{X_0}}:X_0\in G_\mr{supp}\}$ be the family of passive Gaussian states as defined before, with the prior distribution  $X_0\sim\mu$.
    Let $\mc A$ denote a (possibly randomized) algorithm that takes as input the transcript of outcomes $\bm x$ and outputs an estimator $\mc A(\bm x)\in \mr{Herm}_0(n)$ for $X_0$.
    Given that $N = o(n/(\nu^2\xi^2))$, one must have
    \bb
        \Pr_{\mc A,\, X_0\sim\mu,\, \bm x\sim\mc T_{X_0}}\left(\|\mc A(\bm x) - X_0\|_\mr{tr}\le na/2\right) \le 0.1 + o(1).
    \ee
    Here $\mc A$ is also used to denote the internal randomness of the algorithm.
\end{proposition}
\begin{proof}
    By Proposition~\ref{prop:anti-concentration} and the fact that $\Pr_{X_0\sim\mu}(X_0\in G_\mr{good}) \ge 1 - e^{-\Omega(n)}$,
    \bb
        \Pr_{X_0\sim\mu,\, \bm x\sim \mc T_{X_0}}[\nu_{\bm x}(B_{\mr{tr},na}(X_0)) \le \exp(-\Omega(n^2))] \ge 1 - \delta_0 - o(1).
    \ee
    Combined with the trivial bound $\nu_{\bm x}(B_{\mr{tr},na}(X_0))\le 1$, 
    \bb
        \E_{X_0\sim\mu,\, \bm x\sim \mc T_{X_0}} \nu_{\bm x}(B_{\mr{tr},na}(X_0)) \le \delta_0 + o(1).
    \ee
    By Bayes' rule, compared to sampling first $X_0\sim\mu$ and then $\bm x\sim \mc T_{X_0}$, it is equivalent to first sampling $X'\sim\mu$ and $\bm x\sim \mc T_{X'}$, and then sampling $X_0\sim\nu_{\bm x}$.
    Thus, we can rewrite the above as
    \bb
        \E_{X'\sim\mu,\, \bm x\sim\mc T_{X'}}\E_{X_0\sim\nu_{\bm x}} \nu_{\bm x}(B_{\mr{tr},na}(X_0)) \le \delta_0 + o(1).    
    \ee
    Now, for any fixed $\bm x$, sample independently $X_1,X_2\sim\nu_{\bm x}$. 
    Let $\mc A(\bm x)$ be a realization of the estimator conditioned on its internal randomness.
    Consider the following two events:
    \begin{itemize}
        \item $E_1$: $X_1 \in B_{\mr{tr},na/2}(\mc A(\bm x))$ and $X_2 \in B_{\mr{tr},na/2}(\mc A(\bm x))$. 
        \item $E_2$: $X_2 \in B_{\mr{tr},na}(X_1)$.
    \end{itemize}
    By the triangle inequality,  $E_1\Rightarrow E_2$, and thus $\Pr(E_1) \le \Pr(E_2)$. That is,
    \bb
        \nu_{\bm x}(B_{\mr{tr},na/2}(\mc A(\bm x)))^2 \le \E_{X_1\sim\nu_{\bm x}}\nu_{\bm x}(B_{\mr{tr}, na}(X_1)).
    \ee
    Taking the expectation over $X'\sim\mu,\,\bm x\sim\mc T_{X'}$ and the internal randomness of $\mc A$ on both sides,
    \bb
        \E_{\bm x,\mc A} \nu_{\bm x}(B_{\mr{tr},na/2}(\mc A(\bm x)))^2 &\le \E_{\bm x}\E_{X_1\sim\nu_{\bm x}}\nu_{\bm x}(B_{\mr{tr}, na}(X_1)) \\&\le \delta_0 + o(1).
    \ee
    By Cauchy-Schwarz,
    \bb
        \E_{\bm x,\mc A} \nu_{\bm x}(B_{\mr{tr},na/2}(\mc A(\bm x))) \le \sqrt{\delta_0 + o(1)}.
    \ee
    Again by Bayes' rule, the L.H.S. equals to, 
    $$   
        \Pr_{\mc A,\, X_0\sim\mu,\, \bm x\sim\mc T_{X_0}}\left(X_0\in B_{\mr{tr},na/2}(\mc A(\bm x) )  \right)
        =
        \Pr_{\mc A,\, X_0\sim\mu,\, \bm x\sim\mc T_{X_0}}\left(\|\mc A(\bm x) - X_0\|_\mr{tr}\le na/2\right).
    $$
    The proof is finalized by noting that $\delta_0\le0.01$.
\end{proof}

\noindent Finally, we need the following lemma that relates estimation of $X$ to estimation of $\rho_{\msf R_X}$.
\begin{lemma}[Trace distance bound via fugacity matrix]\label{le:fugacity_stab}
    Assume $1/n\le\nu\le1$. For every $a>0$, there exist constants $C_a>0, \alpha_a>0$ such that, whenever $\xi\le \alpha_a/\sqrt{n\nu}$, if $X_1,X_2\in G_\mr{supp}$ and
    \bb
        \|X_2-X_1\|_\mr{tr} \ge an.
    \ee
    Then,
    \bb
        \|\rho_{\msf R_{X_2}} - \rho_{\msf R_{X_1}}\|_\mr{tr} \ge C_a \xi\sqrt{n\nu}.
    \ee
\end{lemma}

\noindent We postpone the proof for Lemma~\ref{le:fugacity_stab} to Sec.~\ref{sec:postponed_stability}. Below, we prove the main lower bound.

\medskip

\begin{proof}[Proof of Theorem~\ref{th:adaptive-lower}]
    Fix $b=0.1$ and choose a sufficiently small constant $a>0$ such that Proposition~\ref{prop:p2l} holds with appropriate constants. Also fix $\delta_0=0.01$.
    
    Assume for contradiction that $N = o(n^2/(\nu\varepsilon^2))$ copies suffice to learn an $n$-mode passive Gaussian state $\rho$ to $\varepsilon$ trace distance with $2/3$ probability using a single-copy adaptive protocol (denoted by $\mc T$).
    Consider the family $\{\rho_{\msf R_X}:X\in G_\mr{supp}\}$ with prior $\mu$, where we set $\xi \coleq \dfrac{6\varepsilon}{C_{a/2}\sqrt{n\nu }}$ for $C_{a/2}$ to be the constant from Lemma~\ref{le:fugacity_stab}. 
    By taking the theorem's constant upper bound on $\varepsilon$ small enough, this choice satisfies all earlier smallness assumptions on $\xi$.
    Note this family satisfies $(\frac12+\nu)\I\le\Sigma\le(\frac12+2\nu)\I$ as has been shown above.
    Denote the measurement outcome by $\bm x$, and the state estimator by $\hat\rho(\bm x)$ that allows internal randomness. By assumption,
    \bb
        \Pr_{X\sim\mu,\,\bm x\sim\mc T_{X},\,\hat\rho}\left(\frac12\|\hat\rho(\bm x) - \rho_{\msf R_X}\|_\mr{tr} \le \varepsilon\right) \ge 2/3.
    \ee
    We now define an estimator for $X$ from $\hat\rho(\bm x)$. We do so by introducing a finite covering net to avoid complication from measure theory. 
    Let $\eta = \varepsilon$.  Since $G_\mr{supp}$ is compact and the map $X\mapsto \rho_{\msf R_X}$ is continuous in trace norm (which can be shown by appropriate photon number cutoff), the image $\{\rho_{\msf R_X} : X\in G_\mr{supp}\}$ has an $\eta$-covering net in trace norm: Let $\mc N_\eta\subset G_\mr{supp}$ be a finite set such that for every $X \in G_\mr{supp}$ there exists $\mc N(X)\in\mc N_\eta$ such that
    \bb
        \|\rho_{\msf R_X} - \rho_{\msf R_{\mc N(X)}} \|_\mr{tr} \le \eta.
    \ee
    Our estimator of $X$ is defined as follows,
    \bb
        \hat X(\bm x)\coleq\arg\min_{Y \in \mc N_\eta}\|\hat\rho(\bm x) - \rho_{\msf R_Y}\|_\mr{tr},
    \ee
    where we break ties in an arbitrary fixed order.
    Now, conditioned on the event that $\|\hat\rho(\bm x)-\rho_{\msf R_X}\|_\mr{tr}/2\le\varepsilon$,
    \bb
        \|\rho_{\msf R_{\hat X(\bm x)}} - \rho_{\msf R_{X}}\|_\mr{tr} &\leqt{(i)} \|\rho_{\msf R_{\hat X(\bm x)}} - \hat\rho(\bm x)\|_\mr{tr} + \|\hat\rho(\bm x) - \rho_{\msf R_{X}}\|_\mr{tr}
        \\&\leqt{(ii)}  \|\rho_{\msf R_{\mc N(X)}} - \hat\rho(\bm x)\|_\mr{tr} + \|\hat\rho(\bm x) - \rho_{\msf R_{X}}\|_\mr{tr}
        \\&\leqt{(iii)} \|\rho_{\msf R_{\mc N(X)}} - \rho_{\msf R_{X}}\|_\mr{tr} + 2\|\hat\rho(\bm x) - \rho_{\msf R_{X}}\|_\mr{tr}
        \\&\leqt{(iv)} \eta + 4\varepsilon = 5\varepsilon < C_{a/2}\xi\sqrt{n\nu}.
    \ee
    Here, (i) and (iii) is by the triangle inequality; (ii) uses the definition of $\hat X(\bm x)$; (iv) uses our assumption and the definition of the $\eta$-covering net. Lemma~\ref{le:fugacity_stab} then implies that
    \bb
        \|\rho_{\msf R_{\hat X(\bm x)}} - \rho_{\msf R_{X}}\|_\mr{tr} < C_{a/2}\xi\sqrt{n\nu} ~\Longrightarrow~ 
        \|\hat X(\bm x) - X\|_\mr{tr} < an/2.
    \ee
    That is, $\|\hat X(\bm x) - X\|_\mr{tr} < an/2$ holds with probability at least $2/3$. However, Proposition~\ref{prop:hardness_X} states that with $N = o(n/(\nu^2\xi^2)) = o(n^2/(\nu\varepsilon^2))$ samples, this can only succeed with no more than $0.1 + o(1)$ probability. This leads to contradiction. We thus conclude that $N=\Omega(n^2/(\nu\varepsilon^2))$ samples are necessary for the learning task.
\end{proof}

\subsection{Delayed proof for Lemma~\ref{le:fugacity_stab}}\label{sec:postponed_stability}

We first introduce two lemmas that are needed for our concentration argument.

\begin{lemma}[A coarse L\'evy concentration bound]\label{le:levy_concentration}
    Let $W$ be a real-valued random variable and set
    \bb\label{eq:conv-levy-scale}
        s^2\coleq 1+\mr{Var}(W).
    \ee
    For $h>0$, define its concentration function
    \bb\label{eq:conv-levy-concentration-function}
        \mc Q_W(h)\coleq&
        \sup_{t\in\mbb R}\Pr(t-h<W\le t).
    \ee
    If $1\le h\le s$, then
    \bb\label{eq:conv-levy-concentration}
        \mc Q_W(h)\ge c\frac{h}{s}
    \ee
    for an absolute constant $c>0$.
\end{lemma}

\begin{proof}
    Let $\mu\coleq\E W$. Since $s^2=1+\mr{Var}(W)$, Chebyshev's inequality gives
    \bb\label{eq:conv-levy-chebyshev}
        \Pr(|W-\mu|<2s)
        &\ge 1-\frac{\mr{Var}(W)}{4s^2}
        \ge \frac34.
    \ee
    The interval $(\mu-2s,\mu+2s)$ can be covered by at most
    \bb\label{eq:conv-levy-cover}
        \left\lceil\frac{4s}{h}\right\rceil
        &\le \frac{5s}{h}.
    \ee
    Here, the last inequality uses $h\le s$. At least one of these intervals therefore has probability at least $3h/(20s)$, proving Eq.~\eqref{eq:conv-levy-concentration}. This is the elementary form of the concentration--variance principle introduced by L\'evy; see, e.g., \cite{FHS90} for sharp versions and extensions.
\end{proof}

\begin{lemma}[Threshold separation under convolution order]\label{le:convolution_threshold}
    Let $K_-$ and $K_+$ be nonnegative integer-valued random variables. Suppose that on an auxiliary probability space one can write
    \bb\label{eq:conv-abstract-convolution-coupling}
        K_+\overset{\mathrm d}{=}K_-+D,
        \qquad
        D\ \text{is nonnegative and integer-valued},
        \qquad
        D\ \text{is independent of }K_-.
    \ee
    Set
    \bb\label{eq:conv-abstract-scales}
        \delta&\coleq\E D,
        \qquad
        s^2\coleq 1+\mr{Var}(K_-).
    \ee
    Assume that $\delta>0$, $\delta\le s$, and
    \bb\label{eq:conv-abstract-D-second-moment}
        \E D^2\le C_D(\delta+\delta^2)
    \ee
    for some fixed constant $C_D$. Then there exists $t\in\mbb R$ such that
    \bb\label{eq:conv-abstract-threshold-gap}
        \Pr(K_+>t)-\Pr(K_->t)
        \ge c_{C_D}\frac{\delta}{s},
    \ee
    where $c_{C_D}>0$ depends only on $C_D$.
\end{lemma}

\begin{proof}
    The representation in Eq.~\eqref{eq:conv-abstract-convolution-coupling} is precisely the convolution-order relation $K_-\le_{\mathrm{conv}}K_+$. See \cite[Section~1.D]{SS07}. For every $h>0$ and $t\in\mbb R$, monotonicity of the increment gives
    \bb\label{eq:conv-tail-gap-concentration-function}
        \Pr(K_+>t)-\Pr(K_->t)
        &=\Pr(K_-\le t<K_-+D)
        \\&\ge \Pr(D\ge h)\Pr(t-h<K_-\le t),
    \ee
    where independence is used in the second line. Taking the supremum over $t$,
    \bb\label{eq:conv-tail-gap-Q}
        \sup_t\bigl[\Pr(K_+>t)-\Pr(K_->t)\bigr]
        &\ge \Pr(D\ge h)\mc Q_{K_-}(h).
    \ee
    We choose $h$ according to the size of $\delta$. If $0<\delta\le2$, take $h=1$. Since $D$ is nonnegative and integer-valued,
    \bb\label{eq:conv-small-delta-cauchy-schwarz}
        \delta
        &=\E[D\mathbf 1_{\{D>0\}}]
        \\&\le (\E D^2)^{1/2}\Pr(D>0)^{1/2}.
    \ee
    Equation~\eqref{eq:conv-abstract-D-second-moment} thus gives
    \bb\label{eq:conv-small-delta-increment-prob}
        \Pr(D\ge1)
        &\ge\frac{\delta^2}{\E D^2}
        \ge c_{C_D}\delta.
    \ee
    Because $1\le s$, Lemma~\ref{le:levy_concentration} gives $\mc Q_{K_-}(1)\ge c/s$. Substitution into Eq.~\eqref{eq:conv-tail-gap-Q} yields Eq.~\eqref{eq:conv-abstract-threshold-gap}.

    If $\delta>2$, take $h=\delta/2$. The assumption $\delta\le s$ ensures $1<h\le s$. Paley--Zygmund gives
    \bb\label{eq:conv-large-delta-increment-prob}
        \Pr\!\left(D\ge\frac\delta2\right)
        &\ge \frac{(1-1/2)^2(\E D)^2}{\E D^2}
        \ge c_{C_D}.
    \ee
    By Lemma~\ref{le:levy_concentration}, $\mc Q_{K_-}(\delta/2)\ge c\delta/s$. Equation~\eqref{eq:conv-tail-gap-Q} again yields Eq.~\eqref{eq:conv-abstract-threshold-gap}. Finally, because $K_-$ and $K_+$ are integer-valued, the supremum of their tail-probability difference is attained at some threshold $t$.
\end{proof}

\medskip
\noindent Now we proceed to the proof of Lemma~\ref{le:fugacity_stab}.
\begin{proof}[Proof of Lemma~\ref{le:fugacity_stab}]
    Assume that $\|X_1-X_2\|_\mr{tr}\ge na$. We will construct a two-outcome photon-counting measurement whose outcome distributions under $\rho_{\msf R_{X_1}}$ and $\rho_{\msf R_{X_2}}$ have total variation distance at least a constant multiple of $\xi\sqrt{n\nu}$. Below, $c,C>0$ denote absolute constants, while $c_a,C'_a>0$ may depend only on $a$.

    For $X\in G_\mr{supp}$, recall the photon occupation matrix
    \bb\label{eq:fugacity-stab-occupation}
        \msf N_X
        &\coleq \msf R_X(\I-\msf R_X)^{-1}
        \\&=\bar\nu(\I+\xi X)(\I-\bar\nu\xi X)^{-1}
        \\&=(1+\bar\nu)(\I-\bar\nu\xi X)^{-1}-\I,
    \ee
    where we used $\bar\nu=r/(1-r)$. Define
    \bb
        \Delta\coleq \msf N_{X_1}-\msf N_{X_2}.
    \ee
    The resolvent identity gives
    \bb\label{eq:fugacity-stab-resolvent}
        \Delta
        =(1+\bar\nu)\bar\nu\xi
        (\I-\bar\nu\xi X_1)^{-1}
        (X_1-X_2)
        (\I-\bar\nu\xi X_2)^{-1}.
    \ee
    Since $\|X_i\|_\mr{op}\le4$, $\bar\nu\le4/3$, and $\xi\le1/32$, the two resolvent factors in the preceding display and their inverses have uniformly bounded operator norms. In particular, it can be rewritten as
    \bb
        X_1-X_2
        =\frac{(\I-\bar\nu\xi X_1)\Delta(\I-\bar\nu\xi X_2)}{(1+\bar\nu)\bar\nu\xi}
    \ee
    and using the ideal property of the trace norm, we obtain
    \bb\label{eq:fugacity-stab-delta-bounds}
        \|\Delta\|_\mr{tr}
        &\ge c\bar\nu\xi\|X_1-X_2\|_\mr{tr}
        \ge c_a n\bar\nu\xi,
        \\ 
        \|\Delta\|_\mr{op}&\le C\bar\nu\xi.
    \ee
    The same spectral bounds applied to Eq.~\eqref{eq:fugacity-stab-occupation} show that, uniformly over $X\in G_\mr{supp}$,
    \bb\label{eq:fugacity-stab-occupation-bounds}
        c\bar\nu\I\le \msf N_X\le C\bar\nu\I.
    \ee

    Let $\Pi_+$ denote the projector onto the positive-eigenvalue subspace of $\Delta$, and define $\Delta_+\coleq\Pi_+\Delta\Pi_+$. For a Hermitian matrix, at least one of the positive parts of $\Delta$ and $-\Delta$ has trace at least $\|\Delta\|_\mr{tr}/2$. By interchanging $X_1$ and $X_2$ if necessary, and then redefining $\Pi_+$ and $\Delta_+$ accordingly, we may therefore assume that
    \bb
        \Tr\Delta_+\ge\frac12\|\Delta\|_\mr{tr},
    \ee
    Regard the following compressions as operators on $\operatorname{ran}\Pi_+$:
    \bb
        B_+\coleq\Pi_+\msf N_{X_1}\Pi_+,
        \qquad
        B_-\coleq\Pi_+\msf N_{X_2}\Pi_+.
    \ee
    Set
    \bb
        m\coleq\mr{rank}(\Pi_+),
        \qquad
        \delta_\mr{ph}\coleq\Tr(B_+-B_-)=\Tr\Delta_+.
    \ee
    By construction,
    \bb\label{eq:fugacity-stab-compression}
        B_+-B_-=\Pi_+\Delta\Pi_+=\Delta_+\ge0,
        \qquad
        \delta_\mr{ph}\ge c_a n\bar\nu\xi.
    \ee
    In particular, this implies $B_+\ge B_-$. Moreover,
    \bb\label{eq:fugacity-stab-rank-bounds}
        m
        \ge\frac{\delta_\mr{ph}}{\|\Delta\|_\mr{op}}
        \ge c_a n,
        \qquad
        c\bar\nu\Pi_+\le B_\pm\le C\bar\nu\Pi_+.
    \ee

    We now measure the total photon number in the mode subspace $\operatorname{ran}\Pi_+$. Equivalently, one may first apply a passive Gaussian unitary that maps this subspace to the first $m$ modes, and then photon-count those modes. Let $\widehat K$ denote the corresponding photon-number operator, and let $K_X$ denote its measurement outcome on $\rho_{\msf R_X}$. For $X\in G_\mr{supp}$, write $B_X\coleq\Pi_+\msf N_X\Pi_+$, viewed as an operator on $\operatorname{ran}\Pi_+$, so that $B_{X_1}=B_+$ and $B_{X_2}=B_-$. The reduced state on these $m$ modes is again a passive Gaussian state, and its photon occupation matrix is $B_X$: taking the marginal simply restricts the two-point function $\msf N_X$ to $\operatorname{ran}\Pi_+$. Let $\lambda_1(X)\le\cdots\le\lambda_m(X)$ be the eigenvalues of $B_X$. A passive Gaussian unitary within $\operatorname{ran}\Pi_+$ diagonalizes $B_X$ without changing the total photon number. By the thermal normal form (Eq.~\eqref{eq:thermal_normal_form}), in that basis the reduced state is a product of single-mode thermal states with mean photon numbers $\lambda_i(X)$. Hence the individual mode counts $M_{i,X}$ are independent geometric random variables satisfying
    \bb\label{eq:fugacity-stab-single-mode-pgf}
        \Pr(M_{i,X}=k)
        &=\frac{1}{1+\lambda_i(X)}
        \left(\frac{\lambda_i(X)}{1+\lambda_i(X)}\right)^k,
        \qquad k\in\{0,1,2,\ldots\},
        \\
        \E z^{M_{i,X}}
        &=g_{\lambda_i(X)}(z)
        \coleq\frac{1}{1+\lambda_i(X)(1-z)},
        \qquad 0\le z\le1.
    \ee
    The total count is their sum, and independence makes the probability generating functions multiply. Therefore,
    \bb\label{eq:fugacity-stab-pgf}
        K_X&\overset{\mathrm d}{=}\sum_{i=1}^m M_{i,X},
        \\
        \E_{\rho_{\msf R_X}}z^{K_X}
        &=\prod_{i=1}^m\frac{1}{1+\lambda_i(X)(1-z)}
        =\det_{\operatorname{ran}\Pi_+}\!\left(
        \I_{\operatorname{ran}\Pi_+}+(1-z)B_X
        \right)^{-1}.
    \ee

    Set $\lambda_i^+\coleq\lambda_i(X_1)$ and $\lambda_i^-\coleq\lambda_i(X_2)$. Since $B_+\ge B_-$, Weyl's monotonicity theorem yields $\lambda_i^+\ge\lambda_i^-$ for every $i$. Writing $M_i^+\coleq M_{i,X_1}$ and $M_i^-\coleq M_{i,X_2}$, Eq.~\eqref{eq:fugacity-stab-pgf} gives
    \bb\label{eq:fugacity-stab-geometric-sum}
        K_+\coleq K_{X_1}\overset{\mathrm d}{=}\sum_{i=1}^m M_i^+,
        \qquad
        K_-\coleq K_{X_2}\overset{\mathrm d}{=}\sum_{i=1}^m M_i^-,
    \ee
    where the variables within each sum are independent and $M_i^\pm$ is geometric on $\{0,1,2,\ldots\}$ with mean $\lambda_i^\pm$.

    We next construct the convolution coupling required by Lemma~\ref{le:convolution_threshold}.
    By Eq.~\eqref{eq:fugacity-stab-rank-bounds}, $\lambda_i^+\ge\lambda_i^-\ge c\bar\nu>0$. Taking the ratio between the PGFs of $M_i^\pm$ gives
    \bb\label{eq:fugacity-stab-geometric-quotient}
        \frac{g_{\lambda_i^+}(z)}{g_{\lambda_i^-}(z)}
        &=\frac{1+\lambda_i^-(1-z)}{1+\lambda_i^+(1-z)}
        \\
        &=\frac{\lambda_i^-}{\lambda_i^+}
        +\left(1-\frac{\lambda_i^-}{\lambda_i^+}\right)
        g_{\lambda_i^+}(z).
    \ee
    Since $0<\lambda_i^-/\lambda_i^+\le1$, the right-hand side is a convex combination of the PGF $1$ of the point mass at zero and the geometric PGF $g_{\lambda_i^+}$. Hence it is itself a PGF of a random variable, denoted by $J_i\ge0$. On an auxiliary probability space, we may consequently choose mutually independent variables $M_i^-$ and $J_i$ such that
    \bb
        M_i^-+J_i\overset{\mathrm d}{=}M_i^+.
    \ee
    Defining
    \bb\label{eq:fugacity-stab-coupling}
        K_-\coleq\sum_{i=1}^m M_i^-,
        \qquad
        J\coleq\sum_{i=1}^m J_i,
        \qquad
        K_+\coleq K_-+J,
    \ee
    realizes the two photon-count distributions on a common probability space, with $J\ge0$ independent of $K_-$. 
    Since $J_i$ is a convex combination between the point mass at $0$ and a geometric random variable with mean $\lambda_i^+$, we can obtain that (recall a geometric random variable of mean $\lambda$ has variance $\lambda(1+\lambda)$)
    \bb
        \E J_i&=\lambda_i^+-\lambda_i^-,
        \\ 
        \E J_i^2&=(\lambda_i^+-\lambda_i^-)(1+2\lambda_i^+).
    \ee
    Consequently,
    \bb\label{eq:fugacity-stab-increment-moments}
        \E J&=\sum_{i=1}^m(\lambda_i^+-\lambda_i^-)=\delta_\mr{ph},
        \\ 
        \E J^2
        &=\delta_\mr{ph}^2+\sum_{i=1}^m\mr{Var}(J_i)
        \le C(\delta_\mr{ph}+\delta_\mr{ph}^2),
    \ee
    where the last inequality uses $\lambda_i^+\le C$.

    Lemma~\ref{le:convolution_threshold} shows that the relevant scale is $\delta_\mr{ph}/s$. We now estimate this ratio and verify the remaining hypothesis $\delta_\mr{ph}\le s$. Set
    \bb
        s^2\coleq& 1+\mr{Var}(K_-)\\
        =& 1 + \sum_{i=1}^m\lambda_i^{-}(1+\lambda_i^{-}).
    \ee
    Eqs.~\eqref{eq:fugacity-stab-rank-bounds} and  $n\bar\nu\ge n\nu\ge1$ imply
    \bb\label{eq:fugacity-stab-fluctuation}
        c_a n\bar\nu\le s^2\le Cn\bar\nu.
    \ee
    On the other hand, Eqs.~\eqref{eq:fugacity-stab-delta-bounds} and~\eqref{eq:fugacity-stab-compression} give
    \bb
        c_a n\bar\nu\xi
        \le\delta_\mr{ph}
        \le m\|\Delta\|_\mr{op}
        \le Cn\bar\nu\xi.
    \ee
    Therefore,
    \bb\label{eq:fugacity-stab-shift-scale}
        c_a\xi\sqrt{n\bar\nu}
        \le\frac{\delta_\mr{ph}}{s}
        \le C'_a\xi\sqrt{n\bar\nu}.
    \ee
    Since $\bar\nu\le4\nu/3$, the hypothesis $\xi\le\alpha_a/\sqrt{n\nu}$ ensures $\delta_\mr{ph}/s\le C'_a\sqrt{4/3}\,\alpha_a$. By decreasing $\alpha_a$ if necessary, we may and do assume that $\delta_\mr{ph}\le s$.

    The coupling in Eq.~\eqref{eq:fugacity-stab-coupling}, the moment estimate in Eq.~\eqref{eq:fugacity-stab-increment-moments}, and the bounds above satisfy all the requirements of Lemma~\ref{le:convolution_threshold} with $D=J$ and $\delta=\delta_\mr{ph}$. Therefore, there exists a threshold $\tau\in\mbb R$ such that
    \bb\label{eq:fugacity-stab-threshold}
        \Pr_{\rho_{\msf R_{X_1}}}(K_{X_1}>\tau)
        -\Pr_{\rho_{\msf R_{X_2}}}(K_{X_2}>\tau)
        \ge c\frac{\delta_\mr{ph}}s
        \ge c_a\xi\sqrt{n\bar\nu}.
    \ee
    Coarse-graining the photon count into the two outcomes $\{\widehat K>\tau\}$ and $\{\widehat K\le\tau\}$ gives a classical total variation distance equal to the absolute value of the left-hand side of Eq.~\eqref{eq:fugacity-stab-threshold}. By data processing for trace distance and $\bar\nu\ge\nu$,
    \bb
        \frac12\|\rho_{\msf R_{X_1}}-\rho_{\msf R_{X_2}}\|_\mr{tr}
        \ge c_a\xi\sqrt{n\bar\nu}
        \ge c_a\xi\sqrt{n\nu}.
    \ee
    Renaming $2c_a$ as $C_a$ proves the stated full trace-norm bound.
\end{proof}

\section{Alternative single-copy lower bound for non-adaptive scheme}\label{sec:non-adaptive}

This section presents an alternative proof for the single-copy lower bound. This proof only works for \emph{non-adaptive} schemes, but it holds even if the symplectic eigenvalues of the unknown Gaussian states are a priori known.
For simplicity, here we focus only on cold Gaussian states (i.e., $\nu =\Theta(1/n)$) and do not try to derive a $\nu$-dependent lower bound. 

\begin{theorem}\label{th:non-adaptive-lower-complete}
    For sufficiently large $n$, any single-copy non-adaptive scheme that can learn any $n$-mode passive Gaussian state $\rho(0,\Sigma)$, with the promise that $\frac1n\I\le\Sigma-\frac12\I\le\frac2n\I$, to trace distance precision $0<\varepsilon\le0.01$ with at least $2/3$ probability requires $N=\Omega(n^3/\varepsilon^2)$ copies.
    This holds even if the symplectic eigenvalues of the state are a priori known.
\end{theorem}

This bound can be achieved by the heterodyne measurement~\cite{bittel2025energy}.

\medskip
\noindent 
The following representation theory result is crucial to our proof.
\begin{lemma}\cite[Theorem 6.3]{fulton2013representation}\label{le:character_schur}
    For any $X\in\mr{GL}_n(\mbb C)$, the trace of $\Gamma_{\lambda}(X)$ on $S_\lambda^{(n)}$ is given by the Schur polynomial on the eigenvalues $x_1,\cdots,x_n$ of $X$ on $\mbb C^n$, i.e.,
    \bb
      \Tr\!\left(\Gamma_\lambda(X)\right) = s_\lambda(x_1,\cdots,x_n).
    \ee
    In particular, the dimension of $S_\lambda^{(n)}$ is given by
    $\Tr(P_\lambda^{(n)}) = \Tr(\Gamma_\lambda{(\I_n)}) =  s_\lambda(1,\cdots,1)$.
\end{lemma}

\begin{proof}[Proof of Theorem~\ref{th:non-adaptive-lower-complete}]
  It suffices to consider sufficiently large even $n$: for odd $n$, use the construction below on $n-1$ modes and append a known thermal mode of occupation $1/n$. The choice $\xi=66\varepsilon\le0.66$ made below ensures that this construction still satisfies the stated covariance bounds for sufficiently large $n$.
  Consider passive Gaussian states of the following form:
  \bb
    \rho_U \coleq \Gamma(U)\rho_0\Gamma(U)^\dagger,\quad\mathrm{where}~\rho_0\coleq\tau_{\nu_1}^{\otimes s}\otimes\tau_{\nu_0}^{\otimes s}.
  \ee
  Here $s=n/2$, $\nu_0 = \frac{1}{n}$, $\nu_1=\frac{1 + \xi}{n}$, where $\xi\le1$ is a parameter to be chosen later. 
  
  We need the following fact: There exists an ensemble of unitary $\{U_a\}_{a=1}^M$ with $M=2^{\Theta(n^2)}$ such that 
    \bb\label{eq:packing-net-condition-2}
      \Tr\left(\Pi_{> s}U_a^\dagger U_b \Pi_{\le s} U_b^\dagger U_a \right) \ge \frac{n}{8},\quad\forall a\neq b\in[M],
    \ee
  which is used to define our packing net. 
  Here $\Pi_{\le s}$ and $\Pi_{> s}$ are the projectors onto the first $s$ modes and the last $n-s$ modes, respectively.
    This is directly implied by Lemma 3.2, Lemma 3.3 and Corollary 3.4 from \cite{lowe2025lower}. Basically, it is known that
\bb
\Pr_{U\sim\mr{Haar}}\left(\Tr\left(\Pi_{> s}U \Pi_{\le s} U^\dagger \right)\le \frac n8\right) \le \exp(-\frac{n^2}{32}).
\ee
Suppose one has obtained a collection of $M$ $U_i$ that satisfies Eq.~\eqref{eq:packing-net-condition-2}. Now sample a Haar random $U$ and put it into the ensemble. The probability that the condition is still satisfied is at least $1 - M\exp(-\frac{n^2}{32})$ by the union bound. One can thus inductively construct an ensemble of size $M=2^{\Theta(n^2)}$ that satisfies Eq.~\eqref{eq:packing-net-condition-2} with a non-zero probability.

  To show the packing property, for any $a\neq b\in[M]$, consider the following measurement: First undo $U_a$, then project the last $s$ mode to vacuum. For $\rho_{U_a}$, the probability of seeing vacuum is clearly,
  \bb
    \Pr(\mr{Vac}_{>s}|U_a) = (1+\nu_0)^{-s}. 
  \ee
  For $\rho_{U_b}$, define $W\coleq U_a^\dagger U_b$. In the following determinant, the operator and identity are restricted to the $s$-dimensional subspace $\operatorname{ran}(\Pi_{>s})$. Then, using the fidelity formula for Gaussian states,
  \bb
    \Pr(\mr{Vac}_{>s}|U_b) &= \det\!\left((1+\nu_0)\I + (\nu_1-\nu_0)\Pi_{>s}W\Pi_{\le s}W^\dagger\Pi_{>s}\right)^{-1}
    \\&\leqt{(i)} (1+\nu_0)^{-s}\left( 1 + \frac{\nu_1 - \nu_0}{1+\nu_0} \Tr(\Pi_{>s}W\Pi_{\le s}W^\dagger\Pi_{>s}) \right)^{-1}
    \\&\leqt{(ii)} \left(1+\nu_0\right)^{-s}\left(1+ \frac{\xi}{16}\right)^{-1}.
    \\&\leqt{(iii)} \left(1+\nu_0\right)^{-s}\left(1 - \frac{\xi}{17}\right).
  \ee
  Here (i) uses $\det(\I+X)\ge 1 + \Tr X$ for any $X\ge0$; (ii) uses the condition from Eq.~\eqref{eq:packing-net-condition-2}; (iii) uses $(1+\xi/16)^{-1}=1-\xi/(16+\xi)\le 1-\xi/17$ for $0\le\xi\le1$. Thus,
  \bb
    D_\mr{tr}(\rho_{U_a},\rho_{U_b})&\geqt{(i)} \Pr(\mr{Vac}_{>s}|U_a) - \Pr(\mr{Vac}_{>s}|U_b)
    \\&\geqt{}\frac{\xi}{17}(1+\nu_0)^{-s}
    \\&\geqt{{(ii)}} \frac{\xi}{17}\exp(-s\nu_0)
    \\&\ge \frac{\xi}{29}.
  \ee
  Here (i) uses the data processing inequality for trace distance; (ii) uses $\log(1+\nu_0)\le\nu_0$; the last inequality uses $s\nu_0=1/2$ and $e^{-1/2}/17>1/29$.
  Thus the packing separation is at least $\xi/29$. For a target trace-distance accuracy $\varepsilon$, choose $\xi=66\varepsilon$ so that the pairwise distances are at least $\frac94\varepsilon>2\varepsilon$, allowing minimum-distance decoding whenever the tomography error is at most $\varepsilon$. Since $\varepsilon\le0.01$, this choice satisfies $\xi\le0.66<1$. We keep $\xi$ explicit below and substitute $\xi=\Theta(\varepsilon)$ at the end.

  \medskip
  \noindent
    Next, we will use Fano's argument. Consider the following communication task between Alice and Bob: Alice samples $a\in[M]$ uniformly at random and sends $N$ copies of $\rho_{U_a}$ to Bob, who then performs a single-copy non-adaptive measurement and outputs $\hat a\in[M]$ as his guess of $a$. 
    If there exists a single-copy non-adaptive scheme that can use $N$ copies to learn any passive Gaussian state in the promised family to trace distance precision $\varepsilon$ with $2/3$ success probability, then Bob can guess correctly with at least $2/3$ average success probability.
    Denote Bob's measurement outcomes by $\bm y\coleq(y_1,\cdots,y_N)$.
    By Fano's inequality, we must have 
    \bb
      I(a:\bm y)\ge I(a:\hat a)\ge \Omega(\log M) = \Omega(n^2). 
    \ee
    We now use an argument from \cite{lowe2025lower} and \cite{Haah2017}. For any $W\in\mbb U(n)$, the ensemble $\{\rho_{WU_a}\}_{a=1}^M$ also satisfies pairwise trace distance at least $\frac94\varepsilon$ thanks to the unitary invariance of trace distance. For a fixed non-adaptive scheme, use $\bm y_{W}$ to denote the measurement outcome when the ensemble is $\{\rho_{WU_a}\}_{a=1}^M$. Also, use $\bm z$ to denote the measurement outcome on $\rho_U^{\otimes N}$ with a Haar random $U$.
    Let $p(\cdot|V)$ here denote the full transcript distribution on $\rho_V^{\otimes N}$ for any $V\in\mbb U(n)$.
    We claim that,
    \bb
    \E_{W\sim\mr{Haar}}I(a:\bm y_W) \le I(U:\bm z).
    \ee
    Indeed,
    \bb
    \E_W I(a:\bm y_W) &\eqt{(i)} \E_W H(\E_a p(\cdot|W U_a)) - \E_{W,a}H(p(\cdot|W U_a))\\
    &\leqt{(ii)} H(\E_{W,a} p(\cdot|W U_a)) - \E_{W,a}H(p(\cdot|W U_a))\\
    &\eqt{(iii)} H\left(\E_{U\sim\mr{Haar}}p(\cdot|U)\right) - \E_{U\sim\mr{Haar}} H(p(\cdot|U))\\
    &= I(U:\bm z).
    \ee
    Here, (i) is by definition of mutual information; (ii) uses the concavity of entropy; (iii) uses the invariance of Haar measure. The above inequality implies that there exists at least one $W_*$ such that $I(a:\bm y_{W_*}) \le I(U:\bm z)$. Thus, we can replace our packing net by $\{\rho_{W_{*}U_a}\}_{a=1}^M$.\footnote{It is worth noticing that the choice of $W_*$ depends on the specific measurement scheme being used.} In the following, we derive an upper bound on $I(U:\bm z)$.
    
    \medskip 
    
    \noindent Since Bob's scheme is non-adaptive, each outcome $z_j$ is independent conditioned on $U$. Thus,
    \bb
      I(U:\bm z) &\eqt{(i)} \sum_{j=1}^N I(U:z_j|\bm z_{<j})
      \\&= \sum_{j=1}^N H(z_j|\bm z_{<j}) - H(z_j|U, \bm z_{<j})
      \\&\leqt{(ii)} \sum_{j=1}^N H(z_j) - H(z_j|U)
      \\&= \sum_{j=1}^N I(U:z_j).
    \ee
    Here (i) uses the chain rule of mutual information; (ii) uses the conditional independence of $z_j$'s on $U$ and the fact that conditioning does not increase entropy.
    What's left is to find an appropriate upper bound on $I(U:z)$ for any single-copy measurement outcome $z$.
    Use $p(z|U)$ to denote the probability density of the measurement $\{M_{z}\}$ when applying the scheme to $\rho_U$.
    We have
    \bb\label{eq:combine_with_this}
      I(U:z) &\eqt{(i)} \E_U \mr{KL}\!\left(p(\cdot|U)\,\Big\|\,\E_{U'}p(\cdot|U')\right)
      \\&\leqt{(ii)} \E_U \chi^2\!\left(p(\cdot|U)\,\Big\|\,\E_{U'}p(\cdot|U')\right)
      \\&= \left(\int\!dz\,\frac{\E_U p^2(z|U)}{\E_{U'} p(z|U')}\right) - 1 \eqcol\bigstar.
    \ee
    Here, (i) is by the definition of mutual information; (ii) uses the relation between KL divergence and $\chi^2$-divergence. The expectation over $U$ should be understood as the Haar average.
  The goal is to upper bound $\bigstar$.
  
  Any single-copy measurement can be described by a POVM $\{M_z\}_z$ such that $\int\!dz\,M_z=\I$, where each $M_z$ is a PSD operator acting on $\mc H^{(n)}$. We have $p(z|U)=\Tr(M_z \rho_U)$. Since $\rho_U$ is block-diagonal in the photon-number sector decomposition, replacing $\{M_z\}_z$ by $\{\sum_{m=0}^\infty P_m^{(n)} M_z P_m^{(n)}\}_z$ will not change the measurement outcome distributions, where $P_m^{(n)}$ is the projector onto the $m$-photon sector $\mathcal{H}^{(n)}_m$. Therefore, we can assume without loss of generality that each $M_z$ is block-diagonal in the photon-number sector.

    Furthermore, we can always assume each $M_z$ is rank-$1$, because otherwise one can replace $M_z$ by its spectral decomposition and treat each eigen-projector as a separate measurement outcome, which will not decrease the RHS of Eq.~\eqref{eq:combine_with_this}. 
    Consequently, we only need to consider POVMs of the form $\{\ketbra{z;m}{z;m}\}_{z,m}$ such that $\ket{z;m}\in\mathcal{H}^{(n)}_m$ and $\int\!dz\,\ketbra{z;m}{z;m}=P_m^{(n)}$ for all $m\in\mbb N$.

  Write $\rho_0\coleq\bigoplus_{m=0}^\infty q_m\rho_{0,m}$ where $\rho_{0,m}\in \mc L(\mc H_{m}^{(n)})$ and $\Tr(\rho_{0,m})=1$. Here $q_m$ is the probability distribution of $\rho_0$ in the $m$-photon sector. 
  This also implies $\rho_U = \bigoplus_{m=0}^\infty q_m\Gamma_m(U)\rho_{0,m}\Gamma_m(U)^\dagger$.
  Thus (suppose we have chosen the optimal measurement),
  \bb
  \bigstar + 1 &= \sum_{m=0}^\infty q_m \int\!dz \frac{\E_U\expval{z,m|\Gamma_m(U)\rho_{0,m}\Gamma_m(U)^\dagger|z,m}^2}{\E_U\expval{z,m|\Gamma_m(U)\rho_{0,m}\Gamma_m(U)^\dagger|z,m}}
  \\&\eqcol \sum_{m=0}^\infty q_m \int\!dz \frac{\E_U\expval{z,m|\rho_{U,m}|z,m}^2}{\E_U\expval{z,m|\rho_{U,m}|z,m}}.
  \ee
  To proceed, let us first derive an exact formula of $\rho_{0,m}$.
  \bb\label{eq:rho_0_m}
    \rho_{0,m} &\propto \sum_{|\bm m| = m}\left(\frac{\nu_1}{\nu_1+1}\right)^{|\bm m_{\le s}|}\left(\frac{\nu_0}{\nu_0+1}\right)^{|\bm m_{> s}|}\ketbra{\bm m}{\bm m}
    \\&\propto \sum_{k=0}^m\alpha^k \left(\sum_{|\bm m_{\le s}|=k}\ketbra{\bm m_{\le s}}{\bm m_{\le s}}\right)\otimes\left(\sum_{|\bm m_{> s}|=m-k}\ketbra{\bm m_{> s}}{\bm m_{> s}}\right)
    \\&= \sum_{k=0}^m\alpha^k P_k^{(\le s)}\otimes P_{m-k}^{(> s)}.
  \ee
  Here we have defined $\alpha\coleq (\frac{\nu_1}{\nu_1+1})/(\frac{\nu_0}{\nu_0+1})>1$ for notational simplicity.
  Crucially, the last line can be interpreted as $\Gamma_m(g_\alpha)$ where 
  \bb
  g_\alpha\coleq\mr{diag}(\underbrace{\alpha,\cdots,\alpha}_{s},\underbrace{1,\cdots,1}_{s})\in\mr{GL}_n(\mbb C).
  \ee
  Indeed, $\Gamma_m(g_\alpha)\ket{\bm m} = \alpha^{|\bm m_{\le s}|}\ket{\bm m}$ for any $|\bm m|=m$, which can also be seen via the tensor product action of $g_\alpha$ on $(\mbb C^n)^{\otimes m}$ and then symmetrizing. The equivalence between this and the last line of Eq.~\eqref{eq:rho_0_m} is obvious. Taking into account the normalization factor, we conclude that
  \bb
    \rho_{0,m} = \Gamma_m(g_\alpha)/\Tr\left(\Gamma_m(g_\alpha)\right).
  \ee
  Now we compute the following Haar averages: 
  \bb
    \E_{U}\Gamma_m(U)\rho_{0,m}\Gamma_m(U)^\dagger &\eqt{(i)} \frac{\Tr(\rho_{0,m}P_{m}^{(n)})}{\Tr P_{m}^{(n)}}P_{m}^{(n)} 
    &= \frac{1}{\Tr P_{m}^{(n)}}P_{m}^{(n)}.
  \ee
  \bb\label{eq:m_2nd_average}
    \E_{U}\Gamma_m(U)^{\otimes 2}\rho_{0,m}^{\otimes 2}\Gamma_m(U)^{\dagger,\otimes 2} &\eqt{(ii)} \sum_{r=0}^m \frac{\Tr\!\left(\rho_{0,m}^{\otimes 2}P_{(2m-r,r)}^{(n)}\right)}{\Tr P_{(2m-r,r)}^{(n)}}P_{(2m-r,r)}^{(n)}
    \\&= \sum_{r=0}^m \frac{\Tr\!\left(\Gamma_m(g_\alpha)^{\otimes 2}P_{(2m-r,r)}^{(n)}\right)}{\Tr^2\Gamma_m(g_\alpha)\Tr P_{(2m-r,r)}^{(n)}}P_{(2m-r,r)}^{(n)}
    \\&\eqt{(iii)}\sum_{r=0}^m \frac{\Tr\Gamma_{(2m-r,r)}(g_\alpha)}{\Tr^2\Gamma_m(g_\alpha)\Tr P_{(2m-r,r)}^{(n)}}P_{(2m-r,r)}^{(n)}
    \\&\leqt{(iv)} \left(\max_{0\le r\le m}\frac{\Tr\Gamma_{(2m-r,r)}(g_\alpha)}{\Tr P_{(2m-r,r)}^{(n)}}\right)\frac{\sum_{r=0}^m P^{(n)}_{(2m-r,r)}}{\Tr^2\Gamma_m(g_\alpha)}
    \\&\eqt{(v)} \left(\max_{0\le r\le m}\frac{\Tr\Gamma_{(2m-r,r)}(g_\alpha)}{\Tr P_{(2m-r,r)}^{(n)}}\right)\frac{P^{(n),\otimes 2}_{m}}{\Tr^2\Gamma_m(g_\alpha)}
  \ee
  Here (i) is by Schur's Lemma; (ii) is by Schur's Lemma and Pieri's rule: $\mr{Sym}^m(\mbb C^n)^{\otimes 2}\cong\bigoplus_{r=0}^m S_{(2m-r,r)}^{(n)}$ as the irreducible representation decomposition of $\mr{GL}_n(\mbb C)$; (iii) uses Pieri's rule again; (iv) replace all coefficients with a uniform upper bound; Finally, (v) uses Pieri's rule again in the reversed direction. 
  Note that the expression inside the bracket can be understood as \emph{normalized Schur polynomials}, thanks to Lemma~\ref{le:character_schur},
  \bb
    \frac{\Tr\Gamma_{(2m-r,r)}(g_\alpha)}{\Tr P_{(2m-r,r)}^{(n)}} 
    =
    \frac{s_{(2m-r,r)}(\alpha,\cdots,\alpha,1,\cdots,1)}{s_{(2m-r,r)}(1,\cdots,1,1,\cdots,1)}\eqcol \bar{s}_{(2m-r,r)}(\alpha,\cdots,\alpha,1,\cdots,1).
  \ee
  The following result, conjectured in~\cite{cuttler2011inequalities} and proven in~\cite{sra2016inequalities}, characterizes the monotonicity of normalized Schur polynomials:
  \begin{lemma}[\cite{sra2016inequalities,cuttler2011inequalities}]\label{le:sra}
    Let $\lambda$ and $\mu$ be partitions of $m$ of length no more than $n$. The inequality
    \bb
      \bar{s}_{\lambda}(x_1,\cdots,x_n) \le \bar{s}_{\mu}(x_1,\cdots,x_n)
    \ee
    holds for all $x_1,\cdots,x_n\ge0$ if and only if $\lambda\preceq\mu$. Here $\lambda\preceq\mu$ means $\lambda$ is majorized by $\mu$. That is,
    $\sum_{j=1}^k \lambda_j \le \sum_{j=1}^k \mu_j$ for all $1\le k\le n$. Note that we assume $\lambda$ and $\mu$ are sorted in a non-increasing order. 
  \end{lemma}
  \noindent This immediately yields that $r=0$ maximizes the last line of Eq.~\eqref{eq:m_2nd_average}. Putting everything into the expression of $\bigstar$:
  \bb
    \bigstar + 1 &\le \sum_{m=0}^\infty q_m\int\! dz\frac{\braket{z,m|z,m}^2 \cfrac{\Tr\Gamma_{2m}(g_\alpha)}{\Tr P^{(n)}_{2m}\Tr^2\Gamma_{m}(g_\alpha)}}{\braket{z,m|z,m}\cdot\cfrac1{\Tr P_m^{(n)}}}
    \\& = \sum_{m=0}^\infty q_m \frac{\Tr^2\!P_m^{(n)}\Tr\Gamma_{2m}(g_\alpha)}{\Tr^2{\Gamma_m(g_\alpha)}\Tr\!P_{2m}^{(n)}}
    \\& = \sum_{m=0}^\infty q_m\frac{\bar s_{2m}(\alpha,\cdots,\alpha,1,\cdots,1)}{\bar s_m(\alpha,\cdots,\alpha,1,\cdots,1)^2}.
    \\&\eqcol \sum_{m=0}^\infty q_m\frac{\bar s_{2m}(\bm\alpha_s)}{\bar s_m(\bm\alpha_s)^2},
  \ee
  where we define $\bm\alpha_s\coleq (\alpha,\cdots,\alpha,1,\cdots,1)$ for notational simplicity. 

  \medskip
  \noindent Now we derive a more explicit characterization of $\bar s_m(\bm\alpha_s)$,
  \bb
    \bar s_m(\bm\alpha_s) &= \frac{\Tr\Gamma_{m}(g_\alpha)}{\Tr P_m^{(n)}} \eqt{{(i)}} \sum_{k=0}^m\alpha^k\frac{\binom{k+s-1}{k}\binom{m-k+s-1}{m-k}}{\binom{m+n-1}{m}} \eqt{(ii)} \E_K[\alpha^K].
  \ee
  Here (i) uses $\Gamma_m(g_\alpha) = \sum_{k=0}^m\alpha^k P_k^{(\le s)}\otimes P_{m-k}^{(>s)}$ and the dimension formula for symmetric subspaces;
  For (ii), we introduce $K$ as a random variable with a \emph{Beta-binomial} distribution (see e.g.~\cite{johnson2005univariate}), i.e.,
  $$K\sim\mr{BetaBinomial}(m,s,s).$$
  Equivalently, let $P \sim\mr{Beta}(s,s)$, then $K|P\sim\mr{Binomial}(m,P)$. Using the moment generating function of the binomial distribution, we obtain that
  \bb
    \bar s_m(\bm\alpha_s) =\E_P\E[e^{K\log\alpha}|P] = \E_P\left(1+(\alpha-1)P\right)^m
  \ee
  Same for $\bar s_{2m}(\bm\alpha_s)$.
  
  \medskip
  \noindent Next, we derive some properties of $q_m$. View $m$ as a random variable of distribution $q_m$. Since $\rho_0 = \tau_{\nu_1}^{\otimes{s}}\otimes\tau_{\nu_0}^{\otimes{s}}$, we can write $m = m_1 + m_0$ where $m_1\sim\mr{NB(s,1/(\nu_1+1))}$ and $m_0\sim\mr{NB(s,1/(\nu_0+1))}$ are two independent \emph{negative binomial} random variables. The first and second moment and moment generating function of $m$ can then be derived as
  \bb
    \E_{m\sim q} m &= (\nu_0 + \nu_1)s.
    \\\E_{m\sim q} m^2 &= (\nu_0^2 + \nu_0 + \nu_1^2 + \nu_1)s+ (\nu_0 + \nu_1)^2s^2 = O(1).
    \\\E_{m\sim q} e^{tm} &= (1-\nu_0(e^t-1))^{-s}(1-\nu_1(e^t-1))^{-s}.
  \ee
  We are now ready to bound $\bigstar$. We will divide and bound separately the sum over $m$ into $m\le \sqrt{n}$ and $m>\sqrt{n}$. For $m\le\sqrt n$, 
  \bb
    \frac{\bar s_{2m}(\bm\alpha_s)}{\bar s_{m}(\bm\alpha_s)^2} &= \frac{\E_P\!\left(1+(\alpha-1)P\right)^{2m}}{\left(\E_P\!\left(1+(\alpha-1)P\right)^m\right)^2}
    \\&\leqt{(i)} \frac{\E_P\!\left(1+(\alpha-1)P\right)^{2m}}{\!\left((\alpha+1)/2\right)^{2m}}
    \\&= \E_P \left(1+ 2\frac{\alpha-1}{\alpha+1}\left(P-\frac12\right)\right)^{2m}
    \\&\le \E_P \exp\!\left(4m\frac{\alpha-1}{\alpha+1}(P-\frac12)\right).
    \\&\leqt{(ii)} \exp\!\left(\frac{2m^2(\alpha-1)^2}{(n+1)(\alpha+1)^2}\right)
    \\&\leqt{(iii)}\exp\left(\frac{m^2\xi^2}{2n}\right)
    \\&\leqt{(iv)} 1 + \frac{m^2\xi^2}{n}
  \ee
  Here, (i) uses Jensen's inequality and the fact that $\E P = 1/2$; (ii) uses the sub-gaussianity of Beta distribution as given in the following lemma:
  \begin{lemma}\cite[Theorem 1, simplified]{marchal2017sub}
    For any $s>0$, $P\sim\mr{Beta}(s,s)$ is sub-Gaussian with variance proxy ${1/(4(2s+1))}$. That is,
    \bb
      \E\exp(\lambda(P -\E P)) \le \exp\left(\frac{\lambda^2}{8(2s+1)}\right).
    \ee
  \end{lemma}
  \noindent
  (iii) is by substituting $\alpha$ with $\nu_0$ and $\nu_1$, and then with $n$ and $\xi$. In particular, it is straightforward to verify that $(\alpha-1)/(\alpha+1)\le\xi/2$.
  (iv) uses our assumption that $m\le\sqrt n$ and $\xi\le1$ and the fact that $\exp(x)\le 1+2x$ for $x\in[0,1/2]$. Therefore,
  \bb
    \sum_{m\le\sqrt n}q_m\left( \frac{\bar s_{2m}(\bm\alpha_s)}{\bar s_{m}(\bm\alpha_s)^2}-1\right)&\le \sum_{m\le\sqrt n}q_m\frac{m^2\xi^2}{n}
    \le \E_{m\sim q}[m^2]\frac{\xi^2}{n}
    = O\left(\frac{\xi^2}{n}\right).
  \ee
  The last equality uses the second order moment of $m\sim q$ derived above.

  \medskip
  \noindent For $m>\sqrt n$, we use a different upper bound:
  \bb
     \frac{\bar s_{2m}(\bm\alpha_s)}{\bar s_{m}(\bm\alpha_s)^2}-1 &= \frac{\mr{Var}_P\!\left(1+(\alpha-1)P\right)^m}{\left(\E_P\!\left(1+(\alpha-1)P\right)^m\right)^2}
     \\& \leqt{(i)}\frac{(\alpha^m - 1)^2/4}{((\alpha+1)/2)^{2m}}
     \\& \leqt{(ii)} m^2(\alpha-1)^24^{m}
     \\& \leqt{(iii)} m^2\xi^24^m.
  \ee
  For (i), the denominator uses Jensen's inequality, while the numerator uses the fact that $\mr{Var}[X]\le (\max X - \min X)^2/4$ for any bounded random variable $X$.
  (ii) uses the inequality $\alpha^m-1\le m(\alpha-1)\alpha^{m-1}$ which can be seen by differentiating $x^m$. (iii) uses that $\alpha-1\le\xi$ that can be directly verified.
  Therefore,
  \bb
    \sum_{m>\sqrt n}q_m \left(\frac{\bar s_{2m}(\bm\alpha_s)}{\bar s_{m}(\bm\alpha_s)^2}-1\right) &\le \sum_{m>\sqrt n}q_mm^2\xi^24^{m} = \E_{m\sim q}[\mathds 1_{[m>\sqrt n]}m^2\xi^24^{m}] 
    \\&\leqt{(i)}\E_{m\sim q}[e^{m-\sqrt n}m^2\xi^24^m]
    \\&\leqt{} \xi^2 e^{-\sqrt n}\E_{m\sim q}[e^{(2+\log 4) m}]
    \\&\eqt{(ii)} \xi^2 e^{-\sqrt n}(1-\nu_1(4e^2-1))^{-s}(1-\nu_0(4e^2-1))^{-s}
    \\&\le \xi^2 e^{-\sqrt n} \exp\left(-n \log\left(1-\frac{\xi+1}{n}(4e^2-1)\right)\right)
    \\&\eqt{(iii)}O(\xi^2 e^{-\sqrt n}).
  \ee
  Here (i) uses $\mathds 1_{[m>\sqrt n]} \le e^{m-\sqrt n}$; (ii) uses the moment generating function of $m\sim q$ derived above; (iii) follows from $\xi\le1$ and the bound $\log(1-x)\ge -Cx$ for $0\le x\le x_0<1$, where $C$ is a constant that only depends on $x_0$. This holds for sufficiently large $n$.

  \medskip
  \noindent Combining the two parts, we conclude that
  \bb
    \bigstar = O\left(\frac{\xi^2}{n}\right).
  \ee
  Applying Eq.~\eqref{eq:combine_with_this} to each single-copy measurement and using the uniform bound above gives $\Omega(n^2)\le I(U:\bm z)\le\sum_{j=1}^N I(U:z_j)\le N\,O(\xi^2/n)$. Hence the necessary number of samples must satisfy
  \bb
    N = \Omega\left(\frac{n^3}{\xi^2}\right)=\Omega\left(\frac{n^3}{\varepsilon^2}\right),
  \ee
  where the last equality uses our choice $\xi=\Theta(\varepsilon)$.
  This completes the proof for Theorem~\ref{th:non-adaptive-lower-complete}.
\end{proof}

\newpage

\section{Upper bound on learning warm Gaussian states}
\label{sec:warm-upper}

In this section, we prove the upper bound in Theorem~\ref{th:main_transition}. We use the quadrature ordering
$\widehat{\boldsymbol r}
=(\widehat q_1,\ldots,\widehat q_n,\widehat p_1,\ldots,\widehat p_n)$,
the symplectic form
$\Omega=\left(\begin{smallmatrix}0&\I_n\\-\I_n&0\end{smallmatrix}\right)$,
and the covariance convention so that the vacuum covariance matrix is $\I/2$ as also used in previous sections.

The proof is based on a simple property of the estimator used below. After slightly inflating the
empirical heterodyne covariance, the estimated covariance $\widehat\Sigma$ satisfies
$\widehat\Sigma\succeq\Sigma$ on the concentration event. Thus, the estimation error only adds noise
to the state. We first show that such one-sided covariance perturbations lead to a particularly stable
trace-distance bound. We then apply this bound to the empirical mean and covariance obtained from
heterodyne measurements.

{
    The proof strategy based on the resolvent-based covariance-to-Hamiltonian analysis and the quantum Pinsker inequality is first explored in~\cite{fanizza2025} for bosonic Gaussian states and Hamiltonian learning. 
    Our Theorem~\ref{th:one-sided-relative-covariance} improves upon \cite[Theorem 4.1]{fanizza2025} for one-sided covariance perturbation, thanks to a tighter analysis on the relative Frobenius norm. This improvement is necessary to obtain our optimal temperature-dependent upper bound.
    Our analysis of the heterodyne estimators builds upon \cite{bittel2025energy}.
}

\subsection{A one-sided relative covariance trace-distance bound}
\label{sec:one-sided-relative-covariance}

A centered Gaussian state whose symplectic eigenvalues are strictly larger than $1/2$ is faithful
(i.e., it has trivial kernel) and can be written as a Gibbs state of a positive quadratic Hamiltonian.
More precisely,
\begin{align}
    \rho(0,\Sigma)
    =
    \frac{
        \exp\!\left[
            -\frac12
            \widehat{\boldsymbol r}^{\,T}
            H(\Sigma)
            \widehat{\boldsymbol r}
        \right]
    }{
        \operatorname{Tr}
        \exp\!\left[
            -\frac12
            \widehat{\boldsymbol r}^{\,T}
            H(\Sigma)
            \widehat{\boldsymbol r}
        \right]
    }.
    \label{eq:warm-Gibbs-form}
\end{align}
To express the Hamiltonian matrix $H(\Sigma)$ in terms of the covariance matrix, set
$J:=i\Omega$. Since $\Omega^T=-\Omega$ and $\Omega^2=-\I$, the matrix $J$ is Hermitian and
unitary. In our covariance convention, Eqs.~(5)--(6) of Ref.~\cite{BBP15} give
\begin{align}
    H(\Sigma)
    =
    J f(\Sigma J),
    \qquad
    f(z)
    :=
    2\,\operatorname{arccoth}(2z)
    =
    \log\!\left(
        \frac{z+\frac12}{z-\frac12}
    \right).
    \label{eq:warm-covariance-Hamiltonian-map}
\end{align}

The form of $f$ that will be useful below is its resolvent representation. For every
$z\notin[-1/2,1/2]$, direct integration gives
\begin{align}
    f(z)
    =
    \int_{-1/2}^{1/2}
    \frac{dt}{z-t}.
    \label{eq:warm-scalar-resolvent-formula}
\end{align}
Indeed,
$\int_{-1/2}^{1/2}(z-t)^{-1}dt
=
[-\log(z-t)]_{-1/2}^{1/2}
=
\log((z+1/2)/(z-1/2))$.

We can apply the same identity to the matrix $\Sigma J$. To see this explicitly, observe first that
$\Sigma J$ is similar to the Hermitian matrix $\Sigma^{1/2}J\Sigma^{1/2}$: indeed,
$\Sigma J
=
\Sigma^{1/2}(\Sigma^{1/2}J\Sigma^{1/2})\Sigma^{-1/2}$.
Hence $\Sigma J$ is diagonalizable and has real spectrum. Moreover, if
$\Sigma=SDS^T$ is a Williamson decomposition~\cite[Eq.~(44)]{weedbrook2012gaussian}, with
$D=\operatorname{diag}(\nu_1,\ldots,\nu_n,\nu_1,\ldots,\nu_n)$ in our quadrature ordering,
then symplecticity of $S$ gives $\Sigma J=S(DJ)S^{-1}$. Since $D$ commutes with $J$ and
$J$ has eigenvalues $\pm1$, the eigenvalues of $\Sigma J$ are precisely
$\{\pm\nu_j\}_{j=1}^n$.

For a faithful Gaussian state, every $\nu_j>1/2$, so the spectrum of $\Sigma J$ is disjoint from
$[-1/2,1/2]$. Consequently, $\Sigma J-t\I$ is invertible for every
$t\in[-1/2,1/2]$. Diagonalizing $\Sigma J$ and applying
Eq.~\eqref{eq:warm-scalar-resolvent-formula} to each eigenvalue therefore gives
\begin{align}
    f(\Sigma J)
    =
    \int_{-1/2}^{1/2}
    (\Sigma J-t\I)^{-1}\,dt.
    \label{eq:warm-matrix-resolvent-formula}
\end{align}
This resolvent representation is the starting point for the stability estimate below.

We now state the trace-distance bound that will be used in the tomography proof.

\begin{theorem}[One-sided relative covariance bound]
\label{th:one-sided-relative-covariance}
Let $\bar\Sigma$ be an $n$-mode covariance matrix whose symplectic eigenvalues are all at least
$\nu_->1/2$, and let $\Sigma'$ be another covariance matrix satisfying
$\Sigma'\succeq\bar\Sigma$. Define
\begin{align}
    \gamma_-
    &:=
    1-\frac{1}{2\nu_-},
    &
    \Delta_{\mathrm{rel}}
    &:=
    \left\|
        \bar\Sigma^{-1/2}
        (\Sigma'-\bar\Sigma)
        \bar\Sigma^{-1/2}
    \right\|_{\mathrm F}.
    \label{eq:warm-gamma-and-relative-error}
\end{align}
Then
\begin{align}
    \frac12
    \left\|
        \rho(0,\Sigma')
        -
        \rho(0,\bar\Sigma)
    \right\|_1
    \le
    \frac{1}{2\sqrt{2\gamma_-}}\,
    \Delta_{\mathrm{rel}}.
    \label{eq:warm-one-sided-trace-bound}
\end{align}
\end{theorem}

\begin{proof}
By Williamson's theorem, there is a real symplectic matrix $S$ such that
\begin{align}
    \bar\Sigma
    =
    SDS^T,
    \qquad
    D
    :=
    \operatorname{diag}
    (\bar\nu_1,\ldots,\bar\nu_n,
     \bar\nu_1,\ldots,\bar\nu_n),
    \qquad
    \bar\nu_j\ge\nu_-.
    \label{eq:warm-Williamson-decomposition}
\end{align}
We also have that $D$ commutes with $J=i\Omega$.

We express $\Sigma'$ in the same coordinates and measure its error relative to $D$:
\begin{align}
    \Sigma_0
    &:=
    S^{-1}\Sigma'S^{-T},
    &
    E
    &:=
    D^{-1/2}(\Sigma_0-D)D^{-1/2}.
    \label{eq:warm-relative-perturbation}
\end{align}
The order relation $\Sigma'\succeq\bar\Sigma$ is preserved under congruence by an invertible matrix.
Therefore $\Sigma_0\succeq D$, and hence $E\succeq0$. Equivalently,
\begin{align}
    \Sigma_0
    =
    D^{1/2}(\I+E)D^{1/2},
    \qquad
    E\succeq0.
    \label{eq:warm-relative-form}
\end{align}
This positivity is the main reason for introducing the one-sided estimator.

The matrix $E$ has exactly the same Frobenius norm as the relative error in the statement of the
theorem. To see this, define
$O:=\bar\Sigma^{-1/2}SD^{1/2}$. Since $\bar\Sigma=SDS^T$, we have $OO^T=\I$, so $O$ is
orthogonal. Moreover,
\begin{align}
    \bar\Sigma^{-1/2}
    (\Sigma'-\bar\Sigma)
    \bar\Sigma^{-1/2}
    =
    OEO^T.
    \label{eq:warm-relative-error-conjugation}
\end{align}
The Frobenius norm is invariant under orthogonal conjugation, and consequently
\begin{align}
    \Delta_{\mathrm{rel}}
    =
    \|E\|_{\mathrm F}.
    \label{eq:warm-relative-error-E}
\end{align}
Notice also that
\(\Sigma_0\succeq D\succeq \nu_- \I\). We now show that this ordinary matrix lower bound implies the same lower bound on the symplectic eigenvalues of \(\Sigma_0\).

Let \(M>0\) have symplectic eigenvalues \(\nu_1(M),\ldots,\nu_n(M)\). By Williamson's theorem, the eigenvalues of \(\Omega M\) are
\(\{\pm i\nu_j(M)\}_{j=1}^n\). Since
\(M^{1/2}\Omega M^{1/2}\) is similar to \(\Omega M\), it has the same eigenvalues. Moreover,
\(M^{1/2}\Omega M^{1/2}\) is real antisymmetric and therefore normal, so its singular values are the absolute values of its eigenvalues. Hence
\[
    \nu_{\min}(M)
    =
    s_{\min}\!\left(M^{1/2}\Omega M^{1/2}\right).
\]

Using the standard inequality
\(s_{\min}(AB)\ge s_{\min}(A)s_{\min}(B)\) twice, we obtain
\[
    s_{\min}\!\left(M^{1/2}\Omega M^{1/2}\right)
    \ge
    s_{\min}(M^{1/2})^2\,s_{\min}(\Omega).
\]
If \(M\succeq m\I\), then
\(s_{\min}(M^{1/2})=\sqrt{\lambda_{\min}(M)}\ge\sqrt m\), while
\(\Omega^T\Omega=\I\) implies \(s_{\min}(\Omega)=1\). Therefore
\[
    \nu_{\min}(M)\ge m.
\]

Applying this to \(M=\Sigma_0\) and \(m=\nu_-\) shows that every symplectic eigenvalue of
\(\Sigma_0\) is at least \(\nu_-\). Since
\(\Sigma_0=S^{-1}\Sigma'S^{-T}\) is related to \(\Sigma'\) by a symplectic congruence, and symplectic eigenvalues are invariant under such congruences, the same lower bound holds for \(\Sigma'\).

Thus both \(\bar\Sigma\) and \(\Sigma'\) have symplectic eigenvalues strictly larger than \(1/2\). The corresponding Gaussian states are therefore faithful, and the Hamiltonian formula in
Eq.~\eqref{eq:warm-covariance-Hamiltonian-map} is well defined for both.

We now compare the two Hamiltonian matrices. Set
\begin{align}
    A:=\Sigma_0J,
    \qquad
    B:=DJ.
    \label{eq:warm-A-and-B}
\end{align}
For $t\in[-1/2,1/2]$, the commutation relation $DJ=JD$ gives
\begin{align}
    D^{1/2}(B-t\I)^{-1}D^{1/2}
    &=
    (J-tD^{-1})^{-1},
    \nonumber\\
    D^{1/2}(A-t\I)^{-1}D^{1/2}
    &=
    \bigl((\I+E)J-tD^{-1}\bigr)^{-1}.
    \label{eq:warm-weighted-resolvents}
\end{align}
Define $a(t):=1-|t|/\nu_-$. Since $|t|\le1/2$ and $\nu_->1/2$, we have $a(t)>0$.
To bound the first resolvent in Eq.~\eqref{eq:warm-weighted-resolvents}, we use the fact that
$J$ is unitary. Right multiplication by $J$ therefore leaves singular values unchanged, while it turns
the denominator into a Hermitian matrix:
\[
    (J-tD^{-1})J
    =
    \I-tD^{-1}J.
\]
Since $D^{-1}$ commutes with $J$, the matrix $D^{-1}J$ is Hermitian. Moreover, on the
two-dimensional subspace corresponding to the $(q_j,p_j)$ coordinates,
\[
    D^{-1}J
    =
    \frac{1}{\bar\nu_j}
    \begin{pmatrix}
        0&i\\
        -i&0
    \end{pmatrix},
\]
so its eigenvalues are $\pm1/\bar\nu_j$. It follows that the eigenvalues of
$\I-tD^{-1}J$ are $1\mp t/\bar\nu_j$. Since $\bar\nu_j\ge\nu_-$ and
$t\in[-1/2,1/2]$, all of them are bounded below by
\[
    a(t)
    :=
    1-\frac{|t|}{\nu_-}
    >0.
\]
Hence $\I-tD^{-1}J\succeq a(t)\I$. Equivalently, the smallest singular value of
$J-tD^{-1}$ is at least $a(t)$, and therefore
\[
    \left\|
        D^{1/2}(B-t\I)^{-1}D^{1/2}
    \right\|_{\mathrm{op}}
    =
    \left\|
        (J-tD^{-1})^{-1}
    \right\|_{\mathrm{op}}
    \le
    \frac{1}{a(t)}.
\]

The same argument gives the corresponding bound for the perturbed resolvent. Indeed,
\[
    \bigl((\I+E)J-tD^{-1}\bigr)J
    =
    \I+E-tD^{-1}J.
\]
The matrix on the right is Hermitian, and the one-sided assumption enters precisely here:
since $E\succeq0$,
\[
    \I+E-tD^{-1}J
    \succeq
    \I-tD^{-1}J
    \succeq
    a(t)\I.
\]
Thus the positive relative perturbation $E$ can only increase the relevant spectral gap. Since
multiplication by $J$ again preserves singular values, we conclude that
\begin{align}
    \left\|
        D^{1/2}(A-t\I)^{-1}D^{1/2}
    \right\|_{\mathrm{op}},
    \quad
    \left\|
        D^{1/2}(B-t\I)^{-1}D^{1/2}
    \right\|_{\mathrm{op}}
    \le
    \frac{1}{a(t)}.
    \label{eq:warm-resolvent-bounds}
\end{align}
This is the only place in the proof where the order relation
$\Sigma'\succeq\bar\Sigma$ is needed. In Williamson coordinates, it becomes $E\succeq0$, which
ensures that the perturbed resolvent is no closer to the singular interval $[-1/2,1/2]$ than the
unperturbed one.

Using the matrix resolvent representation
\[
    f(X)
    =
    \int_{-1/2}^{1/2}(X-t\I)^{-1}\,dt,
\]
which applies to both $A$ and $B$, we have
\[
    f(A)-f(B)
    =
    \int_{-1/2}^{1/2}
    \Bigl[(A-t\I)^{-1}-(B-t\I)^{-1}\Bigr]\,dt.
\]
We now use the elementary resolvent identity $ X^{-1}-Y^{-1}
    =
    X^{-1}(Y-X)Y^{-1}$ with $X=A-t\I$ and $Y=B-t\I$. This gives
\[
    (A-t\I)^{-1}-(B-t\I)^{-1}
    =
    (A-t\I)^{-1}(B-A)(B-t\I)^{-1},
\]
and therefore
\begin{align}
    f(A)-f(B)
    =
    \int_{-1/2}^{1/2}
    (A-t\I)^{-1}(B-A)(B-t\I)^{-1}\,dt.
    \label{eq:warm-resolvent-identity}
\end{align}

We next express the perturbation $B-A$ in terms of the relative covariance error $E$.
Since $\Sigma_0=D^{1/2}(\I+E)D^{1/2}$ and $D$ commutes with $J$,
\[
    B-A
    =
    (D-\Sigma_0)J
    =
    -D^{1/2}EJD^{1/2}.
\]
Substituting this into the integrand of
Eq.~\eqref{eq:warm-resolvent-identity} and inserting the weights $D^{1/2}$ on both sides gives
\begin{align}
& D^{1/2}(A-t\I)^{-1}(B-A)(B-t\I)^{-1}D^{1/2}
\nonumber\\
&\qquad=
-
\Bigl[D^{1/2}(A-t\I)^{-1}D^{1/2}\Bigr]
EJ
\Bigl[D^{1/2}(B-t\I)^{-1}D^{1/2}\Bigr].
\label{eq:warm-weighted-integrand}
\end{align}

We can now apply
$\|XYZ\|_{\mathrm F}\le
\|X\|_{\mathrm{op}}\|Y\|_{\mathrm F}\|Z\|_{\mathrm{op}}$.
Since $J$ is unitary, $\|EJ\|_{\mathrm F}=\|E\|_{\mathrm F}$, while
Eq.~\eqref{eq:warm-resolvent-bounds} bounds each weighted resolvent by
$1/a(t)$, with $a(t)=1-|t|/\nu_-$. Hence, for every
$t\in[-1/2,1/2]$,
\[
\begin{aligned}
&
\left\|
    D^{1/2}(A-t\I)^{-1}(B-A)(B-t\I)^{-1}D^{1/2}
\right\|_{\mathrm F}
\\
&\qquad\le
\frac{\|E\|_{\mathrm F}}{a(t)^2}
=
\frac{\|E\|_{\mathrm F}}
{\left(1-|t|/\nu_-\right)^2}.
\end{aligned}
\]
Taking the Frobenius norm in
Eq.~\eqref{eq:warm-resolvent-identity} and using the triangle inequality for the integral, we obtain
\begin{align}
\left\|
    D^{1/2}\bigl(f(A)-f(B)\bigr)D^{1/2}
\right\|_{\mathrm F}
&\le
\|E\|_{\mathrm F}
\int_{-1/2}^{1/2}
\frac{dt}{\left(1-|t|/\nu_-\right)^2}
\nonumber\\
&=
\frac{1}{\gamma_-}\|E\|_{\mathrm F}.
\label{eq:warm-matrix-function-bound}
\end{align}
To evaluate the integral, we use symmetry around $t=0$:
\[
\begin{aligned}
\int_{-1/2}^{1/2}
\frac{dt}{\left(1-|t|/\nu_-\right)^2}
&=
2\int_0^{1/2}
\frac{dt}{(1-t/\nu_-)^2} \\
&=
2\nu_-
\left(
    \frac{1}{\gamma_-}-1
\right)
=
\frac{1}{\gamma_-},
\end{aligned}
\]
where $\gamma_-=1-(2\nu_-)^{-1}$.

The exact $t$-dependence is important here. Replacing
$a(t)$ by its minimum value $\gamma_-$ before integrating would give the weaker factor
$1/\gamma_-^2$. Keeping the full resolvent gap throughout the integral instead gives the sharper
$1/\gamma_-$ dependence.

We now translate the matrix-function estimate
Eq.~\eqref{eq:warm-matrix-function-bound} into a bound on the corresponding quadratic
Hamiltonians. Recall from Eq.~\eqref{eq:warm-covariance-Hamiltonian-map} that
\[
    H(\Sigma_0)=Jf(A),
    \qquad
    H(D)=Jf(B),
\]
because $A=\Sigma_0J$ and $B=DJ$. Therefore
\[
\begin{aligned}
    D^{1/2}\bigl(H(\Sigma_0)-H(D)\bigr)D^{1/2}
    &=
    D^{1/2}J\bigl(f(A)-f(B)\bigr)D^{1/2} \\
    &=
    JD^{1/2}\bigl(f(A)-f(B)\bigr)D^{1/2},
\end{aligned}
\]
where in the second equality we used that $D$ commutes with $J$, and hence so does $D^{1/2}$.
Since $J$ is unitary, left multiplication by $J$ does not change the Frobenius norm. Thus,
using Eq.~\eqref{eq:warm-matrix-function-bound} and
$\|E\|_{\mathrm F}=\Delta_{\mathrm{rel}}$, we obtain
\begin{align}
    \left\|
        D^{1/2}
        \bigl(H(\Sigma_0)-H(D)\bigr)
        D^{1/2}
    \right\|_{\mathrm F}
    \le
    \frac{1}{\gamma_-}\Delta_{\mathrm{rel}}.
    \label{eq:warm-Hamiltonian-bound-Williamson}
\end{align}

We next return from Williamson coordinates to the original covariance matrices
$\Sigma'$ and $\bar\Sigma$. The Hamiltonian matrix transforms covariantly under the same
symplectic change of coordinates. In particular, since
$\Sigma_0=S^{-1}\Sigma'S^{-T}$, we claim that
\[
    H(\Sigma_0)=S^T H(\Sigma')S.
\]

To see this, recall that $S$ is symplectic, so
$S^TJS=J$. Equivalently,
$S^TJ=JS^{-1}$ and $S^{-T}J=JS$. Hence
\[
    \Sigma_0J
    =
    S^{-1}\Sigma'S^{-T}J
    =
    S^{-1}(\Sigma'J)S.
\]
Thus $\Sigma_0J$ is similar to $\Sigma'J$. The same similarity relation is inherited by the
matrix function $f$. Indeed, for every $t\in[-1/2,1/2]$,
\[
    (\Sigma_0J-t\I)^{-1}
    =
    S^{-1}(\Sigma'J-t\I)^{-1}S,
\]
and integrating the resolvent representation of $f$ gives
$f(\Sigma_0J)=S^{-1}f(\Sigma'J)S$.

Using the covariance-to-Hamiltonian formula and $JS^{-1}=S^TJ$, we therefore obtain
\[
\begin{aligned}
    H(\Sigma_0)
    =
    Jf(\Sigma_0J) =
    JS^{-1}f(\Sigma'J)S =
    S^TJf(\Sigma'J)S =
    S^TH(\Sigma')S.
\end{aligned}
\]
Applying exactly the same argument to
$D=S^{-1}\bar\Sigma S^{-T}$ gives
$H(D)=S^TH(\bar\Sigma)S$.

Recall now the orthogonal matrix
$O=\bar\Sigma^{-1/2}SD^{1/2}$ introduced above. From its definition,
$O^T\bar\Sigma^{1/2}=D^{1/2}S^T$ and
$\bar\Sigma^{1/2}O=SD^{1/2}$. Therefore
\begin{align}
&
O^T
\bar\Sigma^{1/2}
\bigl(H(\Sigma')-H(\bar\Sigma)\bigr)
\bar\Sigma^{1/2}
O
\nonumber\\
&\qquad=
D^{1/2}
S^T
\bigl(H(\Sigma')-H(\bar\Sigma)\bigr)
S
D^{1/2}
\nonumber\\
&\qquad=
D^{1/2}
\bigl(H(\Sigma_0)-H(D)\bigr)
D^{1/2}.
\label{eq:warm-Hamiltonian-coordinate-change}
\end{align}
Since $O$ is orthogonal, conjugation by $O$ preserves the Frobenius norm. Combining
Eqs.~\eqref{eq:warm-Hamiltonian-coordinate-change} and
\eqref{eq:warm-Hamiltonian-bound-Williamson} therefore yields
\begin{align}
    \left\|
        \bar\Sigma^{1/2}
        \bigl(H(\Sigma')-H(\bar\Sigma)\bigr)
        \bar\Sigma^{1/2}
    \right\|_{\mathrm F}
    \le
    \frac{1}{\gamma_-}\Delta_{\mathrm{rel}}.
    \label{eq:warm-Hamiltonian-bound}
\end{align}

It remains to convert this Hamiltonian estimate into a trace-distance estimate. For brevity, write
$H':=H(\Sigma')$ and $\bar H:=H(\bar\Sigma)$. For any Hamiltonian matrix $H$, define
\[
    K_H
    :=
    \frac12\widehat{\boldsymbol r}^{\,T}H\widehat{\boldsymbol r},
    \qquad
    Z_H
    :=
    \operatorname{Tr}(e^{-K_H}),
    \qquad
    \rho_H
    :=
    \frac{e^{-K_H}}{Z_H}.
\]
Thus $K_H$ is the quadratic Hamiltonian operator and $Z_H$ is its partition function. In
particular,
$\log\rho_H=-K_H-\log Z_H$.

Using the definition
$D(\rho\|\sigma)=\operatorname{Tr}[\rho(\log\rho-\log\sigma)]$, we obtain
\[
\begin{aligned}
D\!\left(
    \rho(0,\Sigma')
    \middle\|
    \rho(0,\bar\Sigma)
\right)
&=
\operatorname{Tr}\!\left[
    \rho(0,\Sigma')
    (K_{\bar H}-K_{H'})
\right]
+
\log Z_{\bar H}-\log Z_{H'},\\
D\!\left(
    \rho(0,\bar\Sigma)
    \middle\|
    \rho(0,\Sigma')
\right)
&=
\operatorname{Tr}\!\left[
    \rho(0,\bar\Sigma)
    (K_{H'}-K_{\bar H})
\right]
+
\log Z_{H'}-\log Z_{\bar H}.
\end{aligned}
\]
Adding the two expressions cancels the partition-function terms, leaving
\begin{align}
&
D\!\left(
    \rho(0,\Sigma')
    \middle\|
    \rho(0,\bar\Sigma)
\right)
+
D\!\left(
    \rho(0,\bar\Sigma)
    \middle\|
    \rho(0,\Sigma')
\right)
\nonumber\\
&\qquad=
\operatorname{Tr}\!\left[
    \bigl(\rho(0,\Sigma')-\rho(0,\bar\Sigma)\bigr)
    (K_{\bar H}-K_{H'})
\right].
\label{eq:warm-symmetric-relative-entropy-first}
\end{align}

We now evaluate the expectation of a quadratic Hamiltonian. For every real symmetric matrix
$G$ and every centered state with covariance matrix $\Sigma$,
\[
    \operatorname{Tr}\!\left[
        \rho(0,\Sigma)
        \frac12
        \widehat{\boldsymbol r}^{\,T}
        G
        \widehat{\boldsymbol r}
    \right]
    =
    \frac12\operatorname{Tr}(G\Sigma).
\]
Indeed, because $G$ is symmetric,
\[
    \widehat{\boldsymbol r}^{\,T}G\widehat{\boldsymbol r}
    =
    \frac12
    \sum_{j,k}
    G_{jk}
    \{\widehat r_j,\widehat r_k\},
\]
since the antisymmetric commutator contribution cancels. The identity then follows directly from
the definition of the covariance matrix.

Applying this identity to
Eq.~\eqref{eq:warm-symmetric-relative-entropy-first} gives the exact symmetric
relative-entropy identity
\begin{align}
&
D\!\left(
    \rho(0,\Sigma')
    \middle\|
    \rho(0,\bar\Sigma)
\right)
+
D\!\left(
    \rho(0,\bar\Sigma)
    \middle\|
    \rho(0,\Sigma')
\right)
\nonumber\\
&\qquad=
\frac12
\operatorname{Tr}
\left[
    (\bar H-H')
    (\Sigma'-\bar\Sigma)
\right].
\label{eq:warm-symmetric-relative-entropy}
\end{align}

The left-hand side is nonnegative, so the trace on the right-hand side is nonnegative as well.
We may therefore upper bound it by its absolute value. Inserting the relative weights
$\bar\Sigma^{1/2}$ and $\bar\Sigma^{-1/2}$ and using cyclicity of the trace gives
\[
\begin{aligned}
&
\operatorname{Tr}
\left[
    (\bar H-H')
    (\Sigma'-\bar\Sigma)
\right]
\\
&\qquad=
-\operatorname{Tr}\!\left[
    \bar\Sigma^{1/2}
    (H'-\bar H)
    \bar\Sigma^{1/2}
    \,
    \bar\Sigma^{-1/2}
    (\Sigma'-\bar\Sigma)
    \bar\Sigma^{-1/2}
\right].
\end{aligned}
\]
Applying the Frobenius Cauchy--Schwarz inequality
$|\operatorname{Tr}(XY)|\le\|X\|_{\mathrm F}\|Y\|_{\mathrm F}$ therefore gives
\begin{align}
&
\left|
\operatorname{Tr}
\left[
    (\bar H-H')
    (\Sigma'-\bar\Sigma)
\right]
\right|
\nonumber\\
&\qquad\le
\left\|
    \bar\Sigma^{1/2}
    (H'-\bar H)
    \bar\Sigma^{1/2}
\right\|_{\mathrm F}
\left\|
    \bar\Sigma^{-1/2}
    (\Sigma'-\bar\Sigma)
    \bar\Sigma^{-1/2}
\right\|_{\mathrm F}
\nonumber\\
&\qquad\le
\frac{1}{\gamma_-}\Delta_{\mathrm{rel}}^2,
\label{eq:warm-weighted-CS}
\end{align}
where the last inequality uses
Eq.~\eqref{eq:warm-Hamiltonian-bound} and the definition of
$\Delta_{\mathrm{rel}}$. Consequently,
\begin{align}
&
D\!\left(
    \rho(0,\Sigma')
    \middle\|
    \rho(0,\bar\Sigma)
\right)
+
D\!\left(
    \rho(0,\bar\Sigma)
    \middle\|
    \rho(0,\Sigma')
\right)
\le
\frac{1}{2\gamma_-}
\Delta_{\mathrm{rel}}^2.
\label{eq:warm-symmetric-relative-entropy-bound}
\end{align}

Finally, quantum Pinsker's inequality
\cite[Theorem~5.41]{Wat18}, with natural logarithms, gives
$D(\rho\|\sigma)\ge\frac12\|\rho-\sigma\|_1^2$.
Applying it once in each direction and adding the resulting inequalities yields
\[
    D(\rho\|\sigma)+D(\sigma\|\rho)
    \ge
    \|\rho-\sigma\|_1^2.
\]
Taking $\rho=\rho(0,\Sigma')$ and $\sigma=\rho(0,\bar\Sigma)$ and combining this with
Eq.~\eqref{eq:warm-symmetric-relative-entropy-bound}, we obtain
\[
    \left\|
        \rho(0,\Sigma')
        -
        \rho(0,\bar\Sigma)
    \right\|_1^2
    \le
    \frac{1}{2\gamma_-}
    \Delta_{\mathrm{rel}}^2.
\]
Taking square roots and recalling that the trace distance is one half of the trace norm gives
\[
    \frac12
    \left\|
        \rho(0,\Sigma')
        -
        \rho(0,\bar\Sigma)
    \right\|_1
    \le
    \frac{1}{2\sqrt{2\gamma_-}}
    \Delta_{\mathrm{rel}},
\]
which is precisely Eq.~\eqref{eq:warm-one-sided-trace-bound}.
\end{proof}
We also need to control errors in the first moment. This part is simpler and does not require the
warmness assumption.

\begin{lemma}[First-moment perturbation]
\label{lem:warm-first-moment}
For every valid covariance matrix $\Sigma$ and all $\mu,\mu'\in\mathbb R^{2n}$,
\begin{align}
    \frac12
    \left\|
        \rho(\mu',\Sigma)
        -
        \rho(\mu,\Sigma)
    \right\|_1
    \le
    \frac12
    \left\|
        \Sigma^{-1/2}(\mu'-\mu)
    \right\|_2.
    \label{eq:warm-first-moment-bound}
\end{align}
\end{lemma}

\begin{proof}
Let $\delta\mu:=\mu'-\mu$. The equal-covariance specialization of the Gaussian root-fidelity formula
\cite[Eq.~(9)]{BBP15} gives
\[
    F\!\left(
        \rho(\mu',\Sigma),
        \rho(\mu,\Sigma)
    \right)
    =
    \exp\!\left[
        -\frac18
        \delta\mu^T\Sigma^{-1}\delta\mu
    \right].
\]
The Fuchs--van de Graaf inequality~\cite[Theorem~3.36]{Wat18} and
$1-e^{-x}\le x$ for $x\ge0$ therefore imply
\[
\begin{aligned}
    \frac12
    \left\|
        \rho(\mu',\Sigma)
        -
        \rho(\mu,\Sigma)
    \right\|_1
    &\le
    \sqrt{
        1-
        \exp\!\left[
            -\frac14
            \delta\mu^T\Sigma^{-1}\delta\mu
        \right]
    }                                                     \\
    &\le
    \frac12
    \sqrt{
        \delta\mu^T\Sigma^{-1}\delta\mu
    },
\end{aligned}
\]
which proves the claim.
\end{proof}

\subsection{Heterodyne tomography and sample complexity}
\label{sec:warm-sample-complexity}

We now apply Theorem~\ref{th:one-sided-relative-covariance} to the covariance estimator obtained from
heterodyne measurements.

\begin{theorem}[Heterodyne tomography under a vacuum gap]
\label{th:warm-heterodyne-upper}
Let $\rho(\mu,\Sigma)$ be an $n$-mode Gaussian state satisfying
\begin{align}
    \Sigma
    \succeq
    \left(\frac12+\nu\right)\I
    \label{eq:warm-promise}
\end{align}
for some $\nu>0$. For every $\epsilon,\delta\in(0,1)$, there is a non-adaptive protocol using only
independent heterodyne measurements that always outputs a valid Gaussian state $\widehat\rho$ and,
with probability at least $1-\delta$, satisfies
\[
    \frac12
    \left\|
        \widehat\rho-\rho(\mu,\Sigma)
    \right\|_1
    \le
    \epsilon
\]
using
\begin{align}
    N
    =
    O\left(
        \left(1+\frac1\nu\right)
        \frac{
            n\bigl(n+\log(1/\delta)\bigr)
        }{\epsilon^2}
    \right)
    \label{eq:warm-sample-complexity}
\end{align}
copies.
\end{theorem}

\begin{proof}
In our covariance convention, a heterodyne measurement on $\rho(\mu,\Sigma)$ produces a classical
Gaussian outcome
\begin{align}
    Y
    \sim
    \mathcal N(\mu,C),
    \qquad
    C
    :=
    \Sigma+\frac12\I.
    \label{eq:warm-heterodyne-law}
\end{align}
This is the standard heterodyne law
\cite[Eq.~(39)]{bittel2025energy}, after the fixed permutation from the interleaved quadrature
ordering used there to the grouped ordering used in this manuscript.

From $N$ independent heterodyne outcomes $Y_1,\ldots,Y_N$, define the empirical mean and the
centered empirical covariance by
\begin{align}
    \widehat\mu
    &:=
    \frac1N\sum_{j=1}^N Y_j,
    &
    \widehat C
    &:=
    \frac1N\sum_{j=1}^N
    (Y_j-\widehat\mu)(Y_j-\widehat\mu)^T.
    \label{eq:warm-empirical-estimators}
\end{align}
The centering is important when $\mu\neq0$: the uncentered second moment estimates
$C+\mu\mu^T$, whereas the theorem places no restriction on the size of $\mu$.

Set
\begin{align}
    \chi
    &:=
    \sqrt{2n}
    +
    \sqrt{2\log(2/\delta)},
    &
    \zeta
    &:=
    \frac{2\chi}{\sqrt N}
    +
    \frac{2\chi^2}{N}.
    \label{eq:warm-concentration-parameters}
\end{align}
The Gaussian concentration bound of
\cite[Lemma~9, Eqs.~(91), (94), and (95)]{bittel2025energy} implies that, with probability at least
$1-\delta$,
\begin{align}
    \left\|
        C^{-1/2}(\widehat\mu-\mu)
    \right\|_2
    &\le
    \frac{\chi}{\sqrt N},
    \nonumber\\
    (1-\zeta)C
    &\preceq
    \widehat C
    \preceq
    (1+\zeta)C.
    \label{eq:warm-concentration-event}
\end{align}
Let $\mathcal E$ denote the event on which both inequalities in
Eq.~\eqref{eq:warm-concentration-event} hold. Thus
$\Pr(\mathcal E)\ge1-\delta$. We analyze the estimator on $\mathcal E$ throughout the rest of the
proof. Notice that both estimates hold on the same event, so we do not need to assume that
$\widehat\mu$ and $\widehat C$ are independent.

Assume $\zeta<1$ and define the inflated covariance estimator
\begin{align}
    \widetilde\Sigma
    :=
    \frac{\widehat C}{1-\zeta}
    -
    \frac12\I.
    \label{eq:warm-inflated-estimator}
\end{align}
The purpose of the factor $(1-\zeta)^{-1}$ is to turn the two-sided statistical error on
$\widehat C$ into a one-sided error on the quantum covariance. The protocol outputs
\[
    \widehat\rho
    :=
    \begin{cases}
        \rho(\widehat\mu,\widetilde\Sigma),
        &\text{if }\zeta<1
        \text{ and }
        \widetilde\Sigma+\frac{i}{2}\Omega\succeq0,\\[1mm]
        \rho(0,\I/2),
        &\text{otherwise}.
    \end{cases}
\]
The second branch is included only to guarantee that every possible measurement transcript produces a
valid Gaussian state. We will now show that it is never used on $\mathcal E$.

Indeed, the lower covariance bound in
Eq.~\eqref{eq:warm-concentration-event} gives
$\widehat C/(1-\zeta)\succeq C$. Recalling that
$C=\Sigma+\I/2$, this immediately implies
$\widetilde\Sigma\succeq\Sigma$. Similarly, the upper covariance bound gives
\[
    \widetilde\Sigma-\Sigma
    =
    \frac{\widehat C}{1-\zeta}-C
    \preceq
    \frac{1+\zeta}{1-\zeta}C-C
    =
    \frac{2\zeta}{1-\zeta}C.
\]
Hence, on $\mathcal E$,
\begin{align}
    0
    \preceq
    \widetilde\Sigma-\Sigma
    \preceq
    \frac{2\zeta}{1-\zeta}C.
    \label{eq:warm-one-sided-estimation-error}
\end{align}
In particular,
\[
    \widetilde\Sigma+\frac{i}{2}\Omega
    =
    \left(\Sigma+\frac{i}{2}\Omega\right)
    +
    (\widetilde\Sigma-\Sigma)
    \succeq0,
\]
because $\Sigma$ is a valid covariance matrix and
$\widetilde\Sigma-\Sigma\succeq0$. Thus $\widetilde\Sigma$ is physical on $\mathcal E$, and the
first branch of the estimator is selected.

We next compare the heterodyne covariance $C$ with the quantum covariance $\Sigma$. Define
\begin{align}
    \alpha_\nu
    &:=
    1+\frac{1}{1+2\nu},
    &
    \gamma_\nu
    &:=
    \frac{2\nu}{1+2\nu}.
    \label{eq:warm-alpha-gamma}
\end{align}
The promise
$\Sigma\succeq(1/2+\nu)\I$ is equivalent to
$(1+2\nu)\I/2\preceq\Sigma$, and therefore
$\I/2\preceq\Sigma/(1+2\nu)$. It follows that
\begin{align}
    C
    =
    \Sigma+\frac12\I
    \preceq
    \left(1+\frac{1}{1+2\nu}\right)\Sigma
    =
    \alpha_\nu\Sigma.
    \label{eq:warm-C-versus-Sigma}
\end{align}

The same promise also implies that every symplectic eigenvalue of $\Sigma$ is at least
$1/2+\nu$, by the ordinary-to-symplectic eigenvalue comparison proved above. Hence
Theorem~\ref{th:one-sided-relative-covariance} applies to the pair
$(\Sigma,\widetilde\Sigma)$ with
\[
    \nu_-=\frac12+\nu,
    \qquad
    \gamma_-
    =
    1-\frac{1}{1+2\nu}
    =
    \gamma_\nu.
\]

We first bound the relative covariance error entering that theorem. Conjugating
Eq.~\eqref{eq:warm-one-sided-estimation-error} by $\Sigma^{-1/2}$ and then using
Eq.~\eqref{eq:warm-C-versus-Sigma} gives
\[
    0
    \preceq
    \Sigma^{-1/2}
    (\widetilde\Sigma-\Sigma)
    \Sigma^{-1/2}
    \preceq
    \frac{2\alpha_\nu\zeta}{1-\zeta}\I.
\]
Let
\[
    R
    :=
    \Sigma^{-1/2}
    (\widetilde\Sigma-\Sigma)
    \Sigma^{-1/2}.
\]
The matrix $R$ is positive semidefinite and has size $2n\times2n$. The previous inequality says that
every eigenvalue of $R$ lies between $0$ and
$2\alpha_\nu\zeta/(1-\zeta)$. Since for a positive semidefinite matrix
$\|R\|_{\mathrm F}^2$ is the sum of the squares of its eigenvalues, we obtain
\begin{align}
    \left\|
        \Sigma^{-1/2}
        (\widetilde\Sigma-\Sigma)
        \Sigma^{-1/2}
    \right\|_{\mathrm F}
    \le
    \frac{2\alpha_\nu\zeta}{1-\zeta}
    \sqrt{2n}.
    \label{eq:warm-relative-covariance-estimation}
\end{align}

We can now apply Theorem~\ref{th:one-sided-relative-covariance}. Since a common displacement is
implemented by the same unitary on both states, trace distance is invariant under it. Therefore
\[
    \left\|
        \rho(\widehat\mu,\widetilde\Sigma)
        -
        \rho(\widehat\mu,\Sigma)
    \right\|_1
    =
    \left\|
        \rho(0,\widetilde\Sigma)
        -
        \rho(0,\Sigma)
    \right\|_1.
\]
Combining Theorem~\ref{th:one-sided-relative-covariance} with
Eq.~\eqref{eq:warm-relative-covariance-estimation} yields
\begin{align}
&
\frac12
\left\|
    \rho(\widehat\mu,\widetilde\Sigma)
    -
    \rho(\widehat\mu,\Sigma)
\right\|_1
\le
\frac{\alpha_\nu}{\sqrt{\gamma_\nu}}
\frac{\zeta}{1-\zeta}
\sqrt n.
\label{eq:warm-covariance-contribution}
\end{align}
This is the contribution to the error coming from covariance estimation.

We next control the first moment. From
$C\preceq\alpha_\nu\Sigma$, inversion reverses the positive-semidefinite order, so
$C^{-1}\succeq\alpha_\nu^{-1}\Sigma^{-1}$, or equivalently
$\Sigma^{-1}\preceq\alpha_\nu C^{-1}$. Consequently,
\[
    \left\|
        \Sigma^{-1/2}(\widehat\mu-\mu)
    \right\|_2
    \le
    \sqrt{\alpha_\nu}
    \left\|
        C^{-1/2}(\widehat\mu-\mu)
    \right\|_2.
\]
Lemma~\ref{lem:warm-first-moment} and the mean concentration bound in
Eq.~\eqref{eq:warm-concentration-event} therefore give
\begin{align}
&
\frac12
\left\|
    \rho(\widehat\mu,\Sigma)
    -
    \rho(\mu,\Sigma)
\right\|_1
\le
\frac{\sqrt{\alpha_\nu}}{2}
\frac{\chi}{\sqrt N}.
\label{eq:warm-mean-contribution}
\end{align}
Importantly, this bound depends only on the normalized estimation error
$\widehat\mu-\mu$ and not on the size of $\mu$ itself. Thus the theorem allows arbitrary first
moments.

We now combine the two errors. On $\mathcal E$, the output is
$\widehat\rho=\rho(\widehat\mu,\widetilde\Sigma)$. Inserting the intermediate state
$\rho(\widehat\mu,\Sigma)$ and applying the triangle inequality gives
\begin{align}
\frac12
\left\|
    \widehat\rho-\rho(\mu,\Sigma)
\right\|_1 &\le
\frac12
\left\|
    \rho(\widehat\mu,\widetilde\Sigma)
    -
    \rho(\widehat\mu,\Sigma)
\right\|_1
+
\frac12
\left\|
    \rho(\widehat\mu,\Sigma)
    -
    \rho(\mu,\Sigma)
\right\|_1 \\
&\le
\frac{\alpha_\nu}{\sqrt{\gamma_\nu}}
\frac{\zeta}{1-\zeta}
\sqrt n
+
\frac{\sqrt{\alpha_\nu}}{2}
\frac{\chi}{\sqrt N}.
\label{eq:warm-final-error-bound}
\end{align}

It remains to choose $N$ so that the two terms in
Eq.~\eqref{eq:warm-final-error-bound} are together at most $\epsilon$.
Recall that
\[
    \chi
    =
    \sqrt{2n}+\sqrt{2\log(2/\delta)},
    \qquad
    \zeta
    =
    \frac{2\chi}{\sqrt N}
    +
    \frac{2\chi^2}{N}.
\]
For convenience, set
\[
    x:=\frac{\chi}{\sqrt N},
\]
so that $\zeta=2x+2x^2$. We choose $N$ such that
\begin{align}
    N
    \ge
    \frac{64\alpha_\nu^2}{\gamma_\nu}
    \frac{n\chi^2}{\epsilon^2}.
    \label{eq:warm-explicit-sample-size}
\end{align}
Equivalently,
\[
    x
    =
    \frac{\chi}{\sqrt N}
    \le
    \frac{\epsilon\sqrt{\gamma_\nu}}
         {8\alpha_\nu\sqrt n}.
\]
Since $\epsilon<1$, $\gamma_\nu\le1$, $\alpha_\nu\ge1$, and $n\ge1$, the right-hand side is at
most $1/8$. Hence $x\le1/8$.

For $0\le x\le1/8$, we have
\[
    \frac{\zeta}{1-\zeta}
    =
    \frac{2x+2x^2}{1-2x-2x^2}
    \le
    4x.
\]
Thus the covariance contribution in
Eq.~\eqref{eq:warm-final-error-bound} satisfies
\[
\begin{aligned}
    \frac{\alpha_\nu}{\sqrt{\gamma_\nu}}
    \frac{\zeta}{1-\zeta}\sqrt n
    &\le
    \frac{4\alpha_\nu}{\sqrt{\gamma_\nu}}\,x\sqrt n \le
    \frac{4\alpha_\nu}{\sqrt{\gamma_\nu}}
    \frac{\epsilon\sqrt{\gamma_\nu}}
         {8\alpha_\nu\sqrt n}
    \sqrt n
    =
    \frac{\epsilon}{2}.
\end{aligned}
\]

The first-moment contribution is even smaller:
\[
\begin{aligned}
    \frac{\sqrt{\alpha_\nu}}{2}\frac{\chi}{\sqrt N} =
    \frac{\sqrt{\alpha_\nu}}{2}x \le
    \frac{\epsilon\sqrt{\gamma_\nu}}
         {16\sqrt{\alpha_\nu n}}
    \le
    \frac{\epsilon}{16},
\end{aligned}
\]
where the last inequality again uses
$\gamma_\nu\le1$, $\alpha_\nu\ge1$, and $n\ge1$.
Therefore, on the event $\mathcal E$,
\[
    \frac12
    \left\|
        \widehat\rho-\rho(\mu,\Sigma)
    \right\|_1
    \le
    \frac{\epsilon}{2}
    +
    \frac{\epsilon}{16}
    <
    \epsilon.
\]
Since $\Pr(\mathcal E)\ge1-\delta$, the estimator succeeds with the required probability.

Finally, using $(a+b)^2\le2a^2+2b^2$,
\[
    \chi^2
    \le
    4\left(
        n+\log\frac{2}{\delta}
    \right),
\]
while
\[
    \frac{\alpha_\nu^2}{\gamma_\nu}
    =
    \frac{2(1+\nu)^2}{\nu(1+2\nu)}
    \le
    2\left(1+\frac1\nu\right).
\]
Substituting these bounds into
Eq.~\eqref{eq:warm-explicit-sample-size} gives
\[
    N
    =
    O\left(
        \left(1+\frac1\nu\right)
        \frac{
            n\bigl(n+\log(1/\delta)\bigr)
        }{\epsilon^2}
    \right),
\]
which is Eq.~\eqref{eq:warm-sample-complexity}. For every fixed $\nu>0$ and constant failure
probability, this is $O_\nu(n^2/\epsilon^2)$.%
\end{proof}

\begin{corollary}[Interpolation in the vacuum gap]
\label{cor:warm-interpolation}
Let $\rho(\mu,\Sigma)$ be an $n$-mode Gaussian state satisfying
\[
    \Sigma
    \succeq
    \left(\frac12+\nu\right)\I
\]
for some $\nu>0$. For constant failure probability, there is a non-adaptive protocol using only
independent heterodyne measurements that learns $\rho(\mu,\Sigma)$ to trace-distance error
$\epsilon$ using
\begin{align}
    N
    =
    O\left(
        \frac{n^2}{\epsilon^2}
        \min\left\{
            n,\,
            1+\frac1\nu
        \right\}
    \right)
    \label{eq:warm-interpolating-sample-complexity}
\end{align}
copies.
\end{corollary}

\begin{proof}
We combine two heterodyne upper bounds. First, Theorem~\ref{th:warm-heterodyne-upper} gives
\[
    N
    =
    O\left(
        \left(1+\frac1\nu\right)
        \frac{n^2}{\epsilon^2}
    \right)
\]
for constant failure probability. On the other hand, the general heterodyne tomography bound of
\cite[Theorem~10]{bittel2025energy} applies without using the vacuum-gap promise and gives
\[
    N
    =
    O\left(
        \frac{n^3}{\epsilon^2}
    \right).
\]
Since both protocols use only independent heterodyne measurements, we may simply use whichever
of the two bounds is smaller for the given value of $\nu$. Therefore
\[
\begin{aligned}
    N
    &=
    O\left(
        \min\left\{
            \frac{n^3}{\epsilon^2},
            \left(1+\frac1\nu\right)\frac{n^2}{\epsilon^2}
        \right\}
    \right) =
    O\left(
        \frac{n^2}{\epsilon^2}
        \min\left\{
            n,\,
            1+\frac1\nu
        \right\}
    \right),
\end{aligned}
\]
which proves Eq.~\eqref{eq:warm-interpolating-sample-complexity}.
\end{proof}

Corollary~\ref{cor:warm-interpolation} makes explicit how the available heterodyne upper bound
changes with the distance from the vacuum boundary. When $\nu=O(1/n)$, the general
$O(n^3/\epsilon^2)$ bound is at least as good as the new estimate. When $\nu=\Theta(1)$,
the second case gives $O(n^2/\epsilon^2)$. Thus the upper bound interpolates between the
cold and warm regimes.

\section{Matching lower bound for learning warm Gaussian states}\label{sec:warm-lower}

\begin{theorem}\label{th:warm-lower-bound}
    There exist constants $C>0$ and $\varepsilon_0>0$ such that the following holds:
    For any $\nu>0$, $0<\varepsilon<\varepsilon_0$, and sufficiently large $n$,
    any scheme that can learn any passive Gaussian state $\rho(0,\Sigma)$ to $\varepsilon$ trace distance with $2/3$ probability with the promise that $\Sigma\ge (\frac12+\nu)\I$ requires at least $N\ge Cn^2/\varepsilon^2$ copies.
\end{theorem}

\begin{proof}
    Without loss of generality, let $n$ be a sufficiently large even number and let $s\coleq n/2$.
    Consider a similar $n$-mode passive Gaussian ensemble used in Sec.~\ref{sec:non-adaptive}:
    \bb
        \rho_U \coleq \Gamma(U) \rho_0 \Gamma(U)^\dagger,\quad \rho_0\coleq \tau_{\nu+\Delta}^{\otimes s}\otimes\tau_{\nu}^{\otimes s},
    \ee
    where $U \in \{U_a \in\mbb U(n)\}_{a=1}^M$ is a size $M=2^{\Theta(n^2)}$ ensemble that satisfies
    \bb\label{eq:packing-net-condition-another}
      \Tr\left(\Pi_{> s}U_a^\dagger U_b \Pi_{\le s} U_b^\dagger U_a \right) \ge \frac{n}{8},\quad\forall a\neq b\in[M],
    \ee
    the existence of which has been shown in Sec.~\ref{sec:non-adaptive}. We assume $\nu\ge 1$ without loss of generality, as a lower bound for a larger $\nu$ automatically implies a lower bound for a smaller $\nu$. $\Delta$ is chosen as
    \bb
        \Delta\coleq C_1\varepsilon \nu/\sqrt n,
    \ee
    with $C_1>0$ a constant to be decided later. All $\rho_U$ in this ensemble clearly satisfy $\Sigma\ge(\frac12+\nu)\I$.

    We first show this ensemble is a $\varepsilon$-packing net in trace distance.
    Assume $0<\varepsilon<\varepsilon_0$ with $\varepsilon_0$ being a sufficiently small constant. Say $\varepsilon_0\coleq1/1800$.
    For $a,b$ such that $D_\mr{tr}(\rho_{U_a},\,\rho_{U_b})\ge3\varepsilon_0$ it trivially holds that $D_\mr{tr}(\rho_{U_a},\,\rho_{U_b})>3\varepsilon$. 
    Otherwise, consider applying a heterodyne measurement on $\rho_{U_a}$. The outcome distribution is a $2n$-dimensional zero-mean Gaussian $\mc N(0,\tilde\Sigma_{a})$ with covariance matrix given by (denote the covariance matrix of $\rho_{U_a}$ by $\Sigma_{a}$):
    \bb
        \tilde\Sigma_{a} \coleq \Sigma_a + \frac12\I = G_{U_a}\left((1+\nu)\I + \Delta\Pi_{\le s}\right)^{\oplus 2}G_{U_a}^T.
    \ee
    Here $G_{U_a}\coleq\Phi_{\mr{re}}(U_a)\in\mathrm{Sp}(2n)\cap\mathrm{O}(2n)$ is the orthogonal symplectic matrix associated with $U_a$ via the realification map from Sec.~\ref{sec:pre}.
    We start to lower bound the pairwise trace distance. 
    \bb
        D_\mr{tr}(\rho_{U_a},\,\rho_{U_b})
        &\geqt{(i)} 
        \mr{TV}(\mc N(0,\tilde\Sigma_a),\, \mc N(0,\tilde\Sigma_b))
        \\&\geqt{(ii)} C_2 \| \tilde\Sigma_a^{-\frac12}(\tilde\Sigma_b-\tilde\Sigma_a)\tilde\Sigma_a^{-\frac12} \|_\mr{F}
        \\&\geqt{(iii)} C_2 \frac{1}{1+\nu+\Delta}\|\tilde\Sigma_b-\tilde\Sigma_a\|_\mr{F}.
    \ee
    Here, (i) uses the data-processing inequality; (ii) uses the TV distance bounds for classical Gaussians~\cite{devroye2018total,arbas2023polynomial} (here the TV distance is upper bounded by $1/600$ by assumption and hence the inequality can be used); (iii) uses the fact that the maximal singular value of $\tilde{\Sigma}_a$ is $(1+\nu+\Delta)$.
    Next, define $W\coleq U_a^\dagger U_b$.
    \bb
        \|\tilde\Sigma_b-\tilde\Sigma_a\|_\mr{F}^2 
        &\eqt{(i)} 2\Delta^2 \Tr\!\left((\Pi_{\le s} - W\Pi_{\le s}W^\dagger)^2\right)
        \\&\eqt{} 4\Delta^2\Tr\!\left(\Pi_{\le s}W\Pi_{>s} W^\dagger\right) 
        \\&\geqt{(ii)} \frac12\Delta^2n.
    \ee
    Here, (i) uses the realification map (see Sec.~\ref{sec:pre}); (ii) uses the defining property of our unitary ensemble (see Eq.~\eqref{eq:packing-net-condition-2}). Putting the above two inequalities together,
    \bb
        D_\mr{tr}(\rho_{U_a},\,\rho_{U_b}) \ge \frac{C_2C_1}{3\sqrt 2}\varepsilon,
    \ee
    where we have used $\nu\ge 1$ and the definition of $\Delta$. Setting $C_1\coleq 9\sqrt{2}/C_2$ yields $D_\mr{tr}(\rho_{U_a},\,\rho_{U_b}) \ge 3\varepsilon $.

    \medskip
    \noindent Now, we use the standard Fano's method to prove the claimed lower bound: Let Alice and Bob be two players of a communication game. Alice chooses $a\in[M]$ at uniform random, then sends $\rho_{U_a}^{\otimes N}$ to Bob, who can perform arbitrary quantum measurements to guess the value of $a$ with maximal success probability.
    Suppose a learning protocol satisfying the assumption of Theorem~\ref{th:warm-lower-bound} exists. Then, Bob can guess correctly with probability at least $2/3$. Denote Bob's guess by $\hat a$. Fano's inequality gives that~\cite{cover1999elements}:
    \bb
        I(a:\hat a)\ge \frac23\log M - \log 2 = \Omega(n^2).
    \ee
    On the other hand, $I(a:\hat a)$ is upper bounded by the Holevo $\chi$ quantity of the ensemble $\{\rho_{U_a}^{\otimes N}:a\sim\mr{Unif}([M])\}$ thanks to Holevo theorem~\cite{holevo1973bounds}.
    In this proof, we use the Hamiltonian convention of Ref.~\cite[Appendix C.2.1]{chen2026towards}: $\rho_{U_a}\propto\exp[-\widehat{\boldsymbol r}^{\,T}H_a\widehat{\boldsymbol r}]$, so $H_a=H(\Sigma_a)/2$ relative to Eq.~\eqref{eq:warm-Gibbs-form}. Then
    \bb
        \chi&\coleq D(\E_a \ketbra{a}{a}\otimes\rho_{U_a}^{\otimes N} \| \E_a \ketbra{a}{a}\otimes \E_{b}\rho_{U_{b}}^{\otimes N})
        \\&\leqt{(i)} N \E_{a,b}D(\rho_{U_a}\|\rho_{U_b})
        \\&= \frac{N}{2}\E_{a,b}\bigl[D(\rho_{U_a}\|\rho_{U_b})+D(\rho_{U_b}\|\rho_{U_a})\bigr]
        \\&\leqt{(ii)} N \max_{a,b}\frac12\Tr[(\Sigma_a-\Sigma_b)(H_b - H_a)].
    \ee
    Here (i) uses joint convexity and additivity of the quantum relative entropy $D$; the equality uses that $a,b$ are independent and identically distributed; and (ii) follows from the symmetrized relative-entropy identity in Ref.~\cite[Lemma C.3]{chen2026towards}, followed by bounding the average by the maximum. For passive Gaussian states, the Hamiltonian matrix is given by
    \bb
        H_a \coleq f(\Sigma_a),\quad\mr{where~} f(x)\coleq \frac12\log(\frac{2x+1}{2x-1}),~\forall x>\frac12.
    \ee
    Therefore,
    \bb
        \Tr[(\Sigma_a-\Sigma_b)(H_b - H_a)] 
        &\eqt{(i)} \Delta\,\left(f(\frac12+\nu)-f(\frac12+\nu+\Delta)\right)2\Tr\!\left((\Pi_{\le s} - W\Pi_{\le s}W^\dagger)^2\right)
        \\&\leqt{(ii)} \frac{\Delta^2}{\nu^2}n
        = C_1^2\varepsilon^2.
    \ee
    Here, (i) uses $\Sigma_a=G_{U_a}((\frac12+\nu)\I_n+\Delta\Pi_{\le s})^{\oplus2}G_{U_a}^T$ and the realification map: the trace on the right is over $\mbb C^n$, and the factor $2$ accounts for the two real quadrature blocks. For (ii), $f'(x)=-1/[2(x^2-1/4)]<0$ and $f''(x)>0$. Thus,
    \bb
        f(\frac12+\nu) - f(\frac12+\nu+\Delta) \le -\Delta f'(\frac12 + \nu) = \frac{\Delta}{2\nu(\nu+1)} \le \frac{\Delta}{2\nu^2}.
    \ee
    The factor of $n$ appears by using $|\Pi_{\le s} - W\Pi_{\le s}W^\dagger|\le\I_n$. Combining this with the factor $1/2$ in the upper bound on $\chi$ gives
    \bb
        \Omega(n^2)\le I(a:\hat a) \le \chi \le \frac{C_1^2}{2}N\varepsilon^2.
    \ee
    Therefore, $N=\Omega(n^2/\varepsilon^2)$. Note that the prefactor hidden in $\Omega$ can be chosen as an absolute constant. This completes the proof for Theorem~\ref{th:warm-lower-bound}.
\end{proof}

\begin{proof}[Proof of Theorem~\ref{th:main_transition}]
Corollary~\ref{cor:warm-interpolation}, with failure probability $1/3$, gives the claimed upper bound using non-adaptive heterodyne measurements.
For the lower bound, it suffices to restrict to passive Gaussian states, which form a subclass of the promised Gaussian family.
When $1/n\le\nu\le1$, Theorem~\ref{th:adaptive-lower} gives the lower bound $\Omega(n^2/(\nu\varepsilon^2))$.
When $0<\nu<1/n$, the hard family used in Theorem~\ref{th:adaptive-lower} with temperature parameter $1/n$ also satisfies the one-sided promise $\Sigma\succeq(\frac12+\nu)\I$, yielding $\Omega(n^3/\varepsilon^2)$.
When $\nu>1$, Theorem~\ref{th:warm-lower-bound} gives $\Omega(n^2/\varepsilon^2)$, even for arbitrary collective measurements.
These lower bounds match the upper bound up to universal constant factors.
Choosing universal $\varepsilon_0$ sufficiently small and $n_0$ sufficiently large to satisfy the hypotheses of these results completes the proof.
\end{proof}

\bibliography{ref.bib}
\bibliographystyle{alpha}

\end{document}